\documentclass[a4paper,12pt,reqno]{amsart}
\pdfoutput=1

\usepackage{a4}
\usepackage{amsthm}
\usepackage[T1]{fontenc}
\usepackage[utf8]{inputenc}
\usepackage[british]{babel}
\usepackage[unicode,psdextra]{hyperref}
\usepackage{tikz}
\usetikzlibrary{decorations.pathmorphing,shapes}
\usetikzlibrary{arrows.meta,calc}
\def\centerarc[#1](#2)(#3:#4:#5)
    { \draw[#1] ($(#2)+({#5*cos(#3)},{#5*sin(#3)})$) arc (#3:#4:#5); }

\newtheorem{definition}{Definition}
\newtheorem{example}[definition]{Example}
\newtheorem{theorem}[definition]{Theorem}
\newtheorem{lemma}[definition]{Lemma}
\newtheorem{remark}[definition]{Remark}

\newtheorem{corollary}[definition]{Corollary}
\newtheorem{proposition}[definition]{Proposition}

\newtheorem{assumption}[definition]{Assumption}

\DeclareMathOperator{\Res}{Res}
\DeclareMathOperator{\lab}{label}

\makeatletter
\@addtoreset{definition}{section}
\@addtoreset{equation}{section}
\makeatother

\makeatletter
\let\@wraptoccontribs\wraptoccontribs
\makeatother

\allowdisplaybreaks[4]

\begin{document}

\title[Blobbed topological recursion of the quartic Kontsevich 
model II]{Blobbed topological recursion of 
the quartic Kontsevich model II: Genus=0}

\author{Alexander Hock}
\author{Raimar Wulkenhaar}
\contrib[with an appendix by]{Maciej Do\L{}\k{e}ga}

\address{Mathematisches Institut der WWU, 
Einsteinstr.\ 62, 48149 M\"unster, Germany\\
{\itshape E-mail addresses:} \normalfont 
\texttt{a\_hock03@uni-muenster.de}, 
\texttt{raimar@math.uni-muenster.de}}

\address{
Institute of Mathematics, 
Polish Academy of Sciences, 
ul.\ \'{S}niadeckich 8, 
00-956 Warszawa, Poland}
\email{mdolega@impan.pl}

\begin{abstract}
  We prove that the genus-0 sector of the quartic analogue of the
  Kontsevich model is completely governed by an involution identity
  which expresses the meromorphic differential $\omega_{0,n}$ at a
  reflected point $\iota z$ in terms of all $\omega_{0,m}$ with
  $m\leq n$ at the original point $z$.  We prove that the solution of
  the involution identity obeys blobbed topological recursion, which
  confirms a previous conjecture about the quartic Kontsevich model.
\end{abstract}


\maketitle
\markboth{\hfill\textsc\shortauthors}{\textsc{{Blobbed topological
      recursion of the quartic Kontsevich model II: Genus=0}\hfill}}

\section{Introduction and main result}

\subsection{Overview}

This paper completes the solution of the genus-0 sector of the quartic
analogue of the Kontsevich model. This is a model for $N\times N$
Hermitian matrices with the same covariance as the Kontsevich model
\cite{Kontsevich:1992ti} but with quartic instead of cubic potential.
The non-linear Dyson-Schwinger equation \cite{Grosse:2009pa} for the
planar $2$-point function of the quartic Kontsevich model was solved
in a special case in \cite{Panzer:2018tvy} and then in full generality
in \cite{Grosse:2019jnv, Schurmann:2019mzu-v3}. Building on this
foundation we identified in \cite{Branahl:2020yru} three families of
correlation functions and established interwoven loop equations
between them. One family consists of meromorphic differential forms
$\omega_{g,n}$ labelled by genus $g$ and number $n$ of marked points
of a complex curve. By a lengthy evaluation of residues the solution
was found for $\omega_{0,2},\omega_{0,3},\omega_{0,4}$ and
$\omega_{1,1}$.  It strongly suggested that the family $\omega_{g,n}$
obeys blobbed topological recursion \cite{Borot:2015hna}, a systematic
extension of topological recursion \cite{Eynard:2007kz} by additional
terms holomorphic at ramification points of a covering
$x:\hat{\mathbb{C}}\to\hat{\mathbb{C}}$.

Recall that (blobbed) topological recursion (see e.g.\
\cite{Eynard:2007kz, Borot:2015hna} and references therein) starts
from a spectral curve
$(x:\Sigma\to \Sigma_0,\omega_{0,1}=ydx,\omega_{0,2})$. Here
$y:\Sigma\to \hat{\mathbb{C}}$ is regular at the ramification points
of $x$ and $\omega_{0,2}$ is a symmetric bidifferential which extends
(or is equal to) the Bergman kernel. We noticed in
\cite{Branahl:2020yru} that the two coverings $x,y$ in the quartic
Kontsevich model are related by $y(z)=-x(-z)$ (already visible in
\cite{Grosse:2019jnv, Schurmann:2019mzu-v3}) and that
$\omega_{0,2}(u,z)=-\omega_{0,2}(u,-z)$. We show in this paper that
this observation is far more than a coincidence: the properties of
$x,y,\omega_{0,2}$ under reflection $z\mapsto \iota z:=-z$ completely
characterise the genus-0 sector of the quartic Kontsevich model.
There is a single global equation (\ref{eq:flip-om}) which describes
the behaviour of the $\omega_{0,n}$ under reflection. This equation
can be solved without connecting it to the matrix model, the solution
is identical to the solution of the complicated system of
loop equations in \cite{Branahl:2020yru}, and it obeys blobbed
topological recursion \cite{Borot:2015hna} (restricted to genus
$g=0$).  We observe that the reflection formula (\ref{eq:flip-om}) is
of similar form to the relations studied in the context of $x$-$y$
symmetry in topological recursion (see for instance \cite[Proposition
3.1]{Eynard:2013csa}).

It is currently not known to us how to extend these results to genus
$g>0$. The loop equations of \cite{Branahl:2020yru} bring in another
structure which leads to poles of $\omega_{g>0,n}$ at the fixed point
of the involution $\iota$. Nevertheless we speculate that the global
involution $z\to \iota z$ characterises the quartic Kontsevich model
completely and that it describes a geometric structure of the moduli
space of complex curves.

\subsection{Statement of the result}

Let $x:\hat{\mathbb{C}}\to\hat{\mathbb{C}}$ be a ramified covering of
the Riemann sphere $\hat{\mathbb{C}}=\mathbb{C}\cup\{\infty\}$ with
simple ramification points $\beta_1,...,\beta_r$ (which solve
$dx(\beta_i)=0)$.  Let $\iota:\hat{\mathbb{C}}\to \hat{\mathbb{C}}$ be
a global involution, $\iota^2(z)=z$, which does not fix or permute any
ramification point(s).  Another ramified covering of the Riemann
sphere is introduced by\footnote{This could be generalised to
  $y(z)=c-x(\iota z)$ for any $c\in \mathbb{C}$.}
\begin{align}
y=-x\circ \iota :\hat{\mathbb{C}}\to \hat{\mathbb{C}}\;.
\label{def:y}
\end{align}
The data are completed by a unique (up to a global constant factor)
bidifferential $\omega_{0,2}$ on
$\hat{\mathbb{C}}\times \hat{\mathbb{C}}$ which is symmetric, odd
under the involution of one variable and has a double pole on the
diagonal without residue. These conditions give 
\begin{align}
  \omega_{0,2}(w,z)&= \frac{1}{2} \frac{dw\,dz}{(w-z)^2}+
  \frac{1}{2} \frac{d(\iota w) \,d(\iota z)}{(\iota w-\iota z)^2}
-\frac{1}{2} \frac{dw\,d(\iota z)}{(w-\iota z)^2}
-  \frac{1}{2} \frac{d(\iota w) \,dz}{(\iota w- z)^2}
\label{om02}
\\
&\equiv
-d_w\Big(
\frac{1}{2} \frac{dz}{(w-z)}+
  \frac{1}{2} \frac{\iota'(z)dz}{(\iota w-\iota z)}
-\frac{1}{2} \frac{\iota'(z)dz}{(w-\iota z)}
-  \frac{1}{2} \frac{dz}{(\iota w- z)}\Big)\;.
\nonumber
\end{align}
From these data we build (see the first paragraph of Section~\ref{sec:tools}
for the notation):
\begin{definition}[involution identity]
A family $\{\omega_{0,m+1}\}_{m\geq 1}$ of meromorphic
differentials on $\hat{\mathbb{C}}^{m+1}$ is introduced by
\eqref{om02} for $m=1$ and for $m\geq 2$ by
\begin{align}
&  \omega_{0,|I|+1}(I,q)
  +\omega_{0,|I|+1}(I,\iota q)
\label{eq:flip-om}
\\
&=\sum_{s=2}^{|I|} \sum_{I_1\uplus ...\uplus I_s=I}
\frac{1}{s} \Res\displaylimits_{z\to q}  \Big(
\frac{dy(q) dx(z)}{(y(q)-y(z))^{s}}  \prod_{j=1}^s
\frac{\omega_{0,|I_j|+1}(I_j,z)}{dx(z)} 
\Big)\;,
\nonumber
\end{align}
where $I=\{u_1,...,u_m\}$.
\end{definition}
We show that, under
mild assumptions, the identity (\ref{eq:flip-om}) completely determines 
$\omega_{0,|I|+1}(I,q)$, and that the required symmetry of the 
rhs of (\ref{eq:flip-om}) under $q\mapsto \iota q$ is automatic:
\begin{theorem}
\label{thm:flip}
For $\omega_{0,|I|+1}(I,z)$ with $I=\{u_1,...,u_m\}$ of length $|I|:=m$
the following conventions are given: 
  \begin{enumerate}
\item $\omega_{0,2}$ is given by \eqref{om02};
\item the meromorphic form $z\mapsto \omega_{0,|I|+1}(I,z)$ has
  for $m\geq 2$ poles at most in points $z$ where the
rhs of  \eqref{eq:flip-om} has poles;
\item $z\mapsto \omega_{0,|I|+1}(I,z)$ is for $m\geq 2$ holomorphic
  at any $z=u_k$;

\item $z\mapsto \omega_{0,|I|+1}(I,\iota z)$ is holomorphic at any
  ramification point $\beta_i$ of $x$.
  
\end{enumerate}
Then \eqref{eq:flip-om} is  for $I=\{u_1,...,u_m\}$ with $m\geq 2$
uniquely solved by
\begin{align}
  \omega_{0,|I|+1}(I,z)
  &= \sum_{i=1}^r
\Res\displaylimits_{q\to \beta_i}K_i(z,q)
  \sum_{I_1\uplus I_2=I} \omega_{0,|I_1|+1}(I_1,q)\omega_{0,|I_2|+1}(I_2,\sigma_i(q))
  \label{sol:omega}
  \\
  &-\sum_{k=1}^m d_{u_k}
 \Big[\Res\displaylimits_{q\to \iota u_k}
\sum_{I_1\uplus I_2=I}
\tilde{K}(z,q,u_k)
d_{u_k}^{-1}\big( \omega_{0,|I_1|+1}(I_1,q)
\omega_{0,|I_2|+1}(I_2,q)\big)\Big]\,. 
\nonumber
\end{align}
Here $\beta_1,...,\beta_r$ are the ramification points of $x$ and
$\sigma_i\neq \mathrm{id}$ denotes the local Galois involution
in the vicinity of $\beta_i$, i.e.\ $x(\sigma_i(z))=x(z)$,
$\lim_{z\to \beta_i}\sigma_i(z)=\beta_i$. 
By 
$d_{u_k}$ we denote the exterior differential in $u_k$, which on 1-forms
has a right inverse 
given by the primitive $d^{-1}_u \omega(u)=\int_{u'=\infty}^{u'=u}\omega(u')$.
The recursion kernels are given by
\begin{align}
K_i(z,q)&:=   \frac{\frac{1}{2} (\frac{dz}{z-q}-\frac{dz}{z-\sigma_i(q)})
}{dx(\sigma_i(q))(y(q)-y(\sigma_i(q)))}\;,\qquad
\nonumber
\\
\tilde{K}(z,q,u)&:=
\frac{\frac{1}{2}\big(\frac{d(\iota z)}{\iota z-\iota q}-\frac{d(\iota z)}{\iota z- u}\big)}{
  dx(q)(y(q)-y(\iota u))}\;.
\label{eq:kernel}
\end{align}
The solution \eqref{sol:omega}$+$\eqref{eq:kernel} implies symmetry of
\eqref{eq:flip-om} under $q\mapsto \iota q$.

\end{theorem}
\noindent
The proof is lengthy and will be divided into many steps. We rely on
combinatorial identities proved in an appendix by Maciej Do{\l}\k{e}ga.
We start to
prove uniqueness: \emph{if} a consistent solution of (\ref{eq:flip-om})
exists, it must be of the form
(\ref{sol:omega})+(\ref{eq:kernel}). Then we prove that
(\ref{sol:omega})+(\ref{eq:kernel}) implies consistency
of (\ref{eq:flip-om}).

In a second part we show that the loop equations
\cite{Branahl:2020yru} of the quartic analogue of the Kontsevich model
lead for the choice $\iota z=-z$ and
$x(z)=R(z)=z-\lambda \sum_{k=1}^d \frac{\varrho_k}{\varepsilon_k+z}$
found in \cite{Schurmann:2019mzu-v3} to exactly the same solution
(\ref{sol:omega})+(\ref{eq:kernel}).  Thereby we prove for genus $0$
the main conjecture of \cite{Branahl:2020yru} that the quartic
Kontsevich model obeys blobbed topological recursion
\cite{Borot:2015hna}.

\subsection*{Acknowledgements}

We thank Johannes Branahl for his essential contributions to
the conjectures about the loop equations and Ga\"etan Borot
and J\"org Sch\"urmann for helpful discussions. We are grateful to
Maciej Do{\l}\k{e}ga for providing proofs for combinatorial
conjectures left open in the first version. We also acknowledge
assistence by Abdelmalek Abdesselam with the proof of an earlier
conjecture that was central in Sec.~3.6. Our work was
supported\footnote{``Funded by
  the Deutsche Forschungsgemeinschaft (DFG, German Research
  Foundation) -- Project-ID 427320536 -- SFB 1442, as well as under
  Germany's Excellence Strategy EXC 2044 390685587, Mathematics
  M\"unster: Dynamics -- Geometry -- Structure."} by the Cluster of
Excellence \emph{Mathematics M\"unster} and the CRC 1442 \emph{Geometry:
  Deformations and Rigidity}.

\section{Proof of Theorem~\ref{thm:flip}}

\subsection{Tools and conventions}

\label{sec:tools}

Throughout this paper we denote by $q,u,u_k,w,z,z_k \in \hat{\mathbb{C}}$
complex numbers and by $a,a',i,j,k,l,m,n,n_0,...,n_s,p,$ $r,s,s'$
non-negative integers.  By $I=\{u_1,...,u_m\}$ we understand a (multi-)set of
length $|I|=m$ of complex numbers, which are allowed to coincide.  By
$\sum_{I_1\uplus ...\uplus I_s=I}$ we denote the sum over all
partitions of the multiset $I$ into \emph{disjoint non-empty} subsets
$I_1,...,I_s$ of any order. If we insist on a sum over ordered subsets 
we write $\sum_{\substack{I_1\uplus ...\uplus I_s=I
  \\ I_1<...<I_s}}$. We define
$\{u_{k_1},...,u_{k_n}\}<\{u_{l_1},...,u_{l_m}\}$ iff $\min(k_i)<\min(l_j)$.
This order is well-defined because the subsets  
$\{u_{k_1},...,u_{k_n}\},\{u_{l_1},...,u_{l_m}\}$ are disjoint.
We will often write $I\uplus u_k$ or $(I,u_k)$ for
$I\uplus \{u_k\}$ and $I\setminus u_k$ for
$I\setminus \{u_k\}$.
In the second part, $\hat{q}^j$ for $j\geq 1$
denotes another preimage of $q$ under $x$, i.e.\
$x(\hat{q}^j)=x(q)$.

\begin{example}\normalfont
  The set $I=\{u_1,u_2,u_3\}$ has 13 (Fubini number
  \href{https://oeis.org/A000670}{OEIS A000670})
  different partitions
  \begin{align*}
 \{u_1,u_2,u_3\},\quad
&  \{u_1\}\uplus  \{u_2,u_3\},\quad
  \{u_2\}\uplus  \{u_3,u_1\},\quad
  \{u_3\}\uplus  \{u_1,u_2\},
  \\
  &\{u_1,u_2\}\uplus  \{u_3\},\quad
  \{u_2,u_3\}\uplus  \{u_1\},\quad
  \{u_3,u_1\}\uplus  \{u_2\},
  \\
  &
  \{u_1\}\uplus  \{u_2\} \uplus \{u_3\},\quad
  \{u_2\}\uplus  \{u_3\} \uplus \{u_1\},\quad
  \{u_3\}\uplus  \{u_1\} \uplus \{u_2\},
  \\
  &\{u_1\}\uplus  \{u_3\} \uplus \{u_2\},\quad
  \{u_2\}\uplus  \{u_1\} \uplus \{u_3\},\quad
  \{u_3\}\uplus  \{u_2\} \uplus \{u_1\}
\end{align*}
and $B_3=5$ (Bell number
  \href{https://oeis.org/A000110}{OEIS A000110})
ordered partitions
\begin{align*}
& \{u_1,u_2,u_3\},\quad  \{u_1\}\uplus  \{u_2,u_3\},\quad
\{u_1,u_2\}\uplus  \{u_3\}, \quad
\{u_1,u_3\}\uplus  \{u_2\}, \quad
\\ & \{u_1\}\uplus  \{u_2\} \uplus \{u_3\}\;.
\end{align*}
\end{example}

We will often need the projection of a meromorphic 1-form $\omega$ to
the principal part $P^w\omega$ of its Laurent series about
$w\in \mathbb{C}$. This projection is obtained by the residue
\begin{align}
P^w\omega(z)=\Res\displaylimits_{q\to w} \frac{\omega(q)dz}{z-q}\;.  
\end{align}
In case of $w=\beta_i$ (a ramification point of $x$) we abbreviate
$\mathcal{P}^i_z \omega(z):=P^{\beta_i}\omega(z)$.

An important tool will be the commutation rule of two iterated residues
of a $1$-form $\omega(q,z)$ in both complex variables $q,z$:
\begin{align}
  \Res\displaylimits_{q\to w}\Res\displaylimits_{z\to q} \omega(q,z)
+  \Res\displaylimits_{q\to w}\Res\displaylimits_{z\to w} \omega(q,z)
=\Res\displaylimits_{z\to w}\Res\displaylimits_{q\to w} \omega(q,z)\;.
\label{commute-res}
\end{align}
It is an immediate consequence of contour integrations and holds under
the assumption that $q=z,q=w,z=w$ are the only poles in a sufficiently
small neighbourhood of $w$. We will encounter a situation where this
assumption does not hold. In the
vicinity of a ramification point $\beta_i$
the contour integral must also enclose the local Galois conjugate
$\sigma_i(q)$:
\begin{align}
  \Res\displaylimits_{q\to \beta_i}\Res\displaylimits_{z\to q} \omega(q,z)
+  \Res\displaylimits_{q\to \beta_i}\Res\displaylimits_{z\to \beta_i}
\omega(q,z)
+  \Res\displaylimits_{q\to \beta_i}
\Res\displaylimits_{z\to \sigma_i(q)} \omega(q,z)
 &=\Res\displaylimits_{z\to \beta_i}\Res\displaylimits_{q\to \beta_i}
  \omega(q,z)\;.
  \label{commute-res-sig}
\end{align}

The residue commutes with partial or exterior derivatives:
\begin{align}
  \Res\displaylimits_{z\to u} \partial_u f(u,z)dz
  =\partial_u\big(\Res\displaylimits_{z\to u} f(u,z)dz\big)\;,\qquad
  \Res\displaylimits_{z\to u} d_u f(u,z)dz
  =d_u\big(\Res\displaylimits_{z\to u} f(u,z)dz\big)\;.
  \label{commute-res-d}
\end{align}
To see this, let $\gamma_\epsilon(w)$ be the loop with centre $w$ and
radius $\epsilon$. Then for $\epsilon>2\delta>0$
\[
  \frac{1}{2\pi \mathrm{i} \delta}
  \Big(\int_{\gamma_{\epsilon}(u+\delta)} \!\!\! f(u+\delta, z)dz
  -\int_{\gamma_{\epsilon}(u)} \!\!\! f(u, z)dz\Big)
  =  \frac{1}{2\pi \mathrm{i}} \int_{\gamma_{\epsilon}(u)} \!\!\! \frac{
  f(u+\delta, z)-f(u,z)}{\delta} dz\;.
\]
The limit $\delta\to 0$ together with independence of all integrals
from $ \epsilon$ gives (\ref{commute-res-d}).

The residue does not change under the local Galois involution, that is
\begin{align}\label{Id}
  \Res\displaylimits_{q\to \beta_i} \omega (q)
  = \Res\displaylimits_{q\to \beta_i} \omega (\sigma_i(q))\;.
\end{align}
Invariance of the term $\frac{c_{-1}dq}{q-\beta_i}$ of the Laurent
expansion follows from
$\sigma_i(q)-\beta_i=-(q-\beta_i)+\mathcal{O}((q-\beta_i)^2)$.  For
poles of order $n>1$ the term
$\frac{c_{-n}d\sigma_i(q)}{(\sigma_i(q)-\beta_i)^n} =-\frac{1}{n}
d\frac{c_{-n}}{(\sigma_i(q)-\beta_i)^{n-1}}$ does not have a residue.

Of particular importance is the following residue:
\begin{align}
  \nabla^n\omega_{0,|I|+1}(I,q)&:=
  \Res\displaylimits_{z\to q}
  \frac{\omega_{0,|I|+1}(I,z)}{(y(z)-y(q))(x(q)-x(z))^n}
  \label{def:nabla-om}
  \\
  &\equiv \frac{(-1)^n}{n!} \lim_{z\to q} \frac{\partial^n}{\partial (x(z))^n}
  \Big(\frac{x(z)-x(q)}{y(z)-y(q)} \frac{\omega_{0,|I|+1}(I,z)}{dx(z)}\Big)\;,
  \nonumber
\end{align}
which is a function of $q$ and a 1-form in every variable in $I$.
In particular,  $\nabla^0\omega_{0,|I|+1}(I,q)=\frac{\omega_{0,|I|+1}(I,q)}{dy(q)}$.
These functions arise in the Taylor expansion
\begin{align}
\frac{x(z)-x(q)}{y(z)-y(q)} \frac{\omega_{0,|I|+1}(I,z)}{dx(z)}
=\sum_{n=0}^\infty (x(q)-x(z))^n    \nabla^n\omega_{0,|I|+1}(I,q) \;.
  \label{nabla-om-taylor}
\end{align}

\begin{lemma}
\label{lem:omega-nabla}
  The involution identity \eqref{eq:flip-om} can be expressed as
\begin{align*}
&  \omega_{0,|I|+1}(I, q)+  \omega_{0,|I|+1}(I,\iota q)
\\
&=-dy(q) \sum_{s=2}^{|I|} \sum_{I_1\uplus ...\uplus I_s=I}
\frac{1}{s} \sum_{n_1+...+n_s=s-1}
\prod_{j=1}^s \nabla^{n_j}\omega_{0,|I_j|+1}(I_j,q)\;,
\end{align*}
where $\sum_{n_1+...+n_s=s-1}$ is the sum over all partitions of $s-1$
into integers $n_i\geq 0$.
\begin{proof}
We evaluate the residue (\ref{eq:flip-om}) as limit of a partial derivative:
\begin{align*}
&\omega_{0,|I|+1}(I, q)+  \omega_{0,|I|+1}(I,\iota q)
\\
&=\sum_{s=2}^{|I|} \sum_{I_1\uplus ...\uplus I_s=I}
\frac{(-1)^s}{s} \Res\displaylimits_{z\to q}  \Big(
\frac{dy(q) dx(z)}{(x(z)-x(q))^{s}}  \prod_{j=1}^s
\frac{x(z)-x(q)}{y(z)-y(q)}
\frac{\omega_{0,|I_j|+1}(I_j,z)}{dx(z)} 
\Big)
\\
&=\sum_{s=2}^{|I|} \sum_{I_1\uplus ...\uplus I_s=I}
\frac{(-1)^s}{s!} dy(q) 
\lim_{z\to q} \frac{\partial^{s-1}}{\partial (x(z))^{s-1}}
\Big(\prod_{j=1}^s
\frac{x(z)-x(q)}{y(z)-y(q)}
\frac{\omega_{0,|I_j|+1}(I_j,z)}{dx(z)} \Big)\;.
\end{align*}
Leibniz's rule for a higher derivative of a product together with
(\ref{def:nabla-om}) give the assertion.
\end{proof}
\end{lemma}
Iterating Lemma~\ref{lem:omega-nabla} we obtain a variant of with only
a single term $\nabla\omega$:
\begin{lemma}
 \label{lem:om-onenabla} 
  The involution identity \eqref{eq:flip-om} can also be expressed as
\begin{align}
&\omega_{0,|I|+1}(I,q)+\omega_{0,|I|+1}(I,\iota q)
  \label{om-hol-qq}
\\
&= -dy(q)\sum_{s=1}^{|I|-1} \sum_{I_0\uplus I_1\uplus ...\uplus I_s=I}
\nabla^{s}\omega_{0,|I_0|+1}(I_0,q)
\prod_{j=1}^s \frac{\omega_{0,|I_j|+1}(I_j,\iota q)}{-dy(q)}\;.
\nonumber
\end{align}
\begin{proof} By induction on $|I|$, starting from the true statement
  for $|I|=1$.  For a partition $I=I_1\uplus ...\uplus I_s$ of
  $I=\{u_1,...,u_m\}$ into $s\geq 2$ subsets together with a given
  partition $n_1+...+n_s=s-1$ we let
  $u_\mu=\min(\biguplus_{\substack{k=1\\ n_k>0}}^sI_k)$ be the
  smallest element within those $I_k$ with $n_k>0$. Moving the subset
  which contains $u_\mu$ to the first place allows us to get rid of
  the $\frac{1}{s}$-factor in Lemma~\ref{lem:omega-nabla}:
\begin{align}
&\omega_{0,|I|+1}(I,q)+\omega_{0,|I|+1}(I,\iota q)
\label{om-to-nabla}
\\
&=-dy(q) \sum_{s=1}^{|I|-1} \sum_{n_0+...+n_s=s}
\sum_{\substack{I_0\uplus I_1\uplus ...\uplus I_s=I\\ u_\mu \in I_0}}
 \!\!\!\!\! \nabla^{n_0}\omega_{0,|I_0|+1}(I_0,q)
\prod_{j=1}^s \nabla^{n_j}\omega_{0,|I_j|+1}(I_j,q)\;.
\nonumber
\end{align}
By construction we have $n_0>0$.  Take in the second line of
(\ref{om-to-nabla}) a term of the form
$X\nabla^{0}\omega_{0,|I_p|+1}(I_p,q)$, for any product $X$ which
contains $I_0$.  Observe that (\ref{om-to-nabla}) then contains for
$|I_p|\geq 2$ also every term of the sum
\[
  X\sum_{k=2}^{|I_p|} \sum_{I_1'\uplus ...\uplus I_k'=I_p}
  \sum_{m_1+...+m_k=k-1} \prod_{\ell=1}^{k}
\nabla^{m_\ell}\omega_{0,|I_\ell'|+1}(I_\ell',q)\;.
\]
By induction hypothesis and with
$\nabla^{0}\omega_{0,|I_p|+1}(I_p,q)=\frac{\omega_{0,|I_p|+1}(I_p,q)}{dy(q)}$
we have
\begin{align*}
&  X\nabla^{0}\omega_{0,|I_p|{+}1}(I_p,q)
  +X\sum_{k=2}^{|I_p|}\sum_{I_1'\uplus ...\uplus I_k'=I_p}
  \sum_{m_1+...+m_k=k{-}1} \prod_{\ell=1}^{k}
  \nabla^{m_\ell}\omega_{0,|I_\ell'|+1}(I'_\ell,q)
  \\[-1ex]
  &= X \frac{\omega_{0,|I_p|+1}(I_p,\iota q)}{-dy(q)}\;,
\end{align*}
which is also true in case $|I_p|=1$.  Repeat this procedure for the
next $X\nabla^{0}\omega_{0,|I_p|+1}(I_p;q)$ for which the product $X$
does not yet contain a factor
$\frac{\omega_{0,|I_j|+1}(I_j,\iota q)}{-dy(q)}$.  At the end of this
procedure, the second line of (\ref{om-to-nabla}) is reduced to a sum
of terms each containing a factor
$\frac{\omega_{0,|I_j|+1}(I_j,\iota q)}{-dy(q)}$.

Now iterate the procedure for every
$X\nabla^{0}\omega_{|I_p|+1}(I_p,q)$ where the product $X$ contains
$\nabla^{n_0}\omega_{|I_0|+1}(I_0;q)$ and precisely one factor
$\frac{\omega_{0,|I_j|+1}(I_j,\iota q)}{-dy(q)}$.  At the end of this
step we have reduced the second line of (\ref{om-to-nabla}) to terms
of the form
$\nabla^{1}\omega_{|I_0|+1}(I_0;q)\frac{\omega_{0,|I_1|+1}(I_1,\iota
  q)}{-dy(q)}$ or with a double factor
$\prod_{j=1}^2\frac{\omega_{0,|I_j'|+1}(I_j',\iota q)}{-dy(q)}$.
Iterate again until all $\nabla^0$ in the second line of
(\ref{om-to-nabla}) are converted. This is the assertion.  Since $I_0$
is anyway distinguished we can omit the condition
$u_\mu \in I_0$.
\end{proof}
\end{lemma}

\subsection{Poles of
  \texorpdfstring{$\omega_{0,m+1}(u_1,...,u_m,z)$
  }{\omega(u1...,um,z)} at
  \texorpdfstring{$z=\iota u_k$}{z=\iota(uk)}}

\begin{lemma}
\label{lem:om-hol1}
The involution identity \eqref{eq:flip-om} together with 
convention \textup{(c)} in Theorem~\ref{thm:flip}
that $ \omega_{0,m+1}(u_1,...,u_m,z)$ is for $m{\geq} 2$
holomorphic at $z{=}u_k$ imply that
$z\mapsto \omega_{0,m+1}(u_1,...,u_m,\iota z)$ has a pole at every
$z= u_k$. The principal part of the corresponding Laurent series is
given by
\begin{align}
&\Res\displaylimits_{q\to u_k}
\frac{\omega_{0,|I|+1}(I,\iota q)dz}{z-q}
\label{om-hol-u}
\\
&=
-d_{u_k} \bigg[\sum_{s=1}^{|I|-1} \sum_{I_1\uplus ...\uplus I_s=I{\setminus} u_k}
\frac{1}{s!}\frac{\partial^s \big(\frac{1}{z-u_k}\big)}{\partial (y(u_k))^s}
\prod_{j=1}^s \frac{\omega_{0,|I_j|+1}(I_j,u_k)}{dx(u_k)} \bigg]dz\;.
\nonumber
\end{align}
Equivalently, 
\begin{align*}
\omega_{0,|I|+1}(I,z)
&=
-d_{u_k} \bigg[\sum_{s=1}^{|I|-1} \sum_{I_1\uplus ...\uplus I_s=I{\setminus} u_k}
\frac{d(\iota z)}{s!}\frac{\partial^s
  \big(\frac{1}{\iota z-u_k}\big)}{\partial (y(u_k))^s}
\prod_{j=1}^s \frac{\omega_{0,|I_j|+1}(I_j,u_k)}{dx(u_k)} \bigg]
\\
&+\text{terms which are holomorphic at $z=\iota u_k$}\;.
\end{align*}
In particular, the poles of $\omega_{0,|I|+1}(I,z)$ at $z=\iota u_k$
do not have a residue.
\begin{proof}
  We divide the involution identity (\ref{eq:flip-om}) by $w-q$ and
  take the residue at $q=u_k$. By convention (c) in
  Theorem~\ref{thm:flip} the term $ \omega_{0,|I|+1}(I,q)$ in the
  first line of (\ref{eq:flip-om}) does not contribute to the residue.
  In the second line we commute the two residues via
  (\ref{commute-res}).  Since the inner integrand is holomorphic
  at $q=u_k$, we  have
\begin{align*}
&\Res\displaylimits_{q\to u_k}
\frac{\omega_{0,|I|+1}(I,\iota q)}{w-q}
\\
&=- \Res\displaylimits_{q\to u_k} \sum_{s=2}^{|I|} \sum_{I_1\uplus ...\uplus I_s=I}
\frac{1}{s} \Res\displaylimits_{z\to u_k}  \Big(
\frac{dy(q) dx(z)}{(w-q)(y(q)-y(z))^{s}}
\prod_{j=1}^s \frac{\omega_{0,|I_j|+1}(I_j,z)}{dx(z)} 
\Big)
\\
&=- \Res\displaylimits_{q\to u_k} \sum_{s=1}^{|I|-1}
\sum_{I_1\uplus ...\uplus I_s=I{\setminus} u_k}
\Res\displaylimits_{z\to u_k}  \Big(
\frac{dy(q) \omega_{0,2}(u_k,z)}{(w-q)(y(q)-y(z))^{s+1}}
\prod_{j=1}^s \frac{\omega_{0,|I_j|+1}(I_j,z)}{dx(z)} 
\Big)\;.
\end{align*}
We implemented the convention that there is only a pole at $z=u_k$ if
a unique factor $\omega_{0,2}(u_k,z)$ is present. It can occur at all
$s$ places of the partition of $I$ into $s$ subsets, so that
$\frac{1}{s}$ cancels. We shifted $s-1\mapsto s$.  We write
$\omega_{0,2}(u_k,z)$ according to the second line of (\ref{om02}) and
commute the differential $d_{u_k}$ according to (\ref{commute-res-d})
in front of the residues:
\begin{align}
&  \Res\displaylimits_{q\to u_k}
\frac{\omega_{0,|I|+1}(I,\iota q)dw}{w-q}
\label{omIjq-hol}
\\
&=-d_{u_k} \bigg[\Res\displaylimits_{q\to u_k} \sum_{s=1}^{|I|-1}
\sum_{I_1\uplus ...\uplus I_s=I{\setminus} u_k}
\Big(
\frac{dy(q)\prod_{j=1}^s \frac{\omega_{0,|I_j|+1}(I_j,u_k)}{dx(u_k)} 
}{(w-q)(y(q)-y(u_k))^{s+1}}
\Big)\bigg]dw
\nonumber
\\
&=-
d_{u_k} \bigg[\sum_{s=1}^{|I|-1} \sum_{I_1\uplus ...\uplus I_s=I{\setminus} u_k}
\frac{1}{s!}\frac{\partial^s \big(\frac{1}{w-q}\big)}{\partial (y(q))^s}
\Big|_{q=u_k}
\prod_{j=1}^s \frac{\omega_{0,|I_j|+1}(I_j,u_k)}{dx(u_k)} \bigg]dw\,.
\nonumber
\end{align}
This is the assertion (when renaming $w\mapsto z$).
\end{proof}
\end{lemma}

We will derive an alternative formula:
\begin{proposition}
\label{prop:om-hol2}
The poles of $z\mapsto \omega_{0,|I|+1}(I,\iota z)$ at $z=u_k$ can
also be evaluated by
\begin{align*}
&\Res\displaylimits_{q\to u_k}
\frac{\omega_{0,|I|+1}(I,\iota q)dz}{z-q}
\\[-1ex]
&=-d_{u_k} \bigg[\Res\displaylimits_{q\to u_k}
\sum_{I_0\uplus I_1=I}
\frac{\frac{1}{2}\big(\frac{dz}{z-q}-\frac{dz}{z-u_k}\big)}{
  dx(\iota q)(y(\iota q)-y(\iota u_k))}
d_{u_k}^{-1}\big( \omega_{0,|I_0|+1}(I_0,\iota q)
\omega_{0,|I_1|+1}(I_1,\iota q)\big)\bigg]\;.
\end{align*}
Equivalently,
\begin{align*}
&\omega_{0,|I|+1}(I,z)
\\[-1ex]
&=-d_{u_k} \bigg[\Res\displaylimits_{q\to \iota u_k}
\sum_{I_0\uplus I_1=I}
\frac{\frac{1}{2}\big(\frac{d(\iota z)}{\iota z-\iota q}
  -\frac{d(\iota z)}{\iota z-u_k}\big)}{
  dx(q)(y(q)-y(\iota u_k))}
d_{u_k}^{-1}\big( \omega_{0,|I_0|+1}(I_0,q)
\omega_{0,|I_1|+1}(I_1,q)\big)\bigg] \\
&+ \text{ terms which are holomorphic at } z=\iota u_k\;.
\end{align*}
\begin{proof}
  We shift $s\mapsto s-1$ in (\ref{om-hol-u})
  and represent the term with $j=0$ via Lemma~\ref{lem:om-onenabla} 
for $q\mapsto \iota u_k$:
\begin{align*}
&\Res\displaylimits_{q\to u_k}
\frac{\omega_{0,|I|+1}(I,\iota q)dw}{w-q}
\\
&=-d_{u_k} \bigg[\sum_{s=0}^{|I|-2}
\sum_{I_0\uplus I_1\uplus ...\uplus I_s=I{\setminus} u_k}
\frac{1}{(s+1)!}\frac{\partial^{s+1} \big(\frac{1}{w-q}\big)}{
  \partial (y(q))^{s+1}}
\Big|_{q=u_k}
\prod_{j=1}^s \frac{\omega_{0,|I_j|+1}(I_j,u_k)}{dx(u_k)}
\\
&\times \frac{(-dy(\iota u_k))}{dx(u_k)}
\sum_{n=0}^{|I_0|-1}\sum_{I'_0\uplus ...\uplus I'_n=I_0}
\nabla^n\omega_{0,|I'_0|+1}(I'_0,\iota u_k)
\prod_{\ell=1}^n
\frac{\omega_{0,|I_\ell'|+1}(I_\ell',u_k)}{-dy(\iota u_k)}\bigg]dw\;.
\end{align*}
We have included the term
$\omega_{|I_0|+1}(I_0,\iota u_k)=dy(\iota u_k)
\nabla^0\omega_{|I_0|+1}(I_0,\iota u_k)$
as $n=0$. Implementing (\ref{def:y}), i.e.\ $-dy(\iota u_k)=dx(u_k)$
suggests to change summation variables to
$s+n\mapsto s\in [0..|I|-2]$. Then we express the derivative with
respect to $y(q)$ as a residue:
\begin{align*}
&\Res\displaylimits_{q\to u_k}
\frac{\omega_{0,|I|+1}(I,\iota q)dw}{w-q}
\\
&=-d_{u_k} \bigg[\sum_{s=0}^{|I|-2} \sum_{n=0}^{s}
\sum_{I_0\uplus I_1\uplus ...\uplus I_s=I{\setminus} u_k}
\frac{1}{(s+1-n)!}\frac{\partial^{s+1-n}
  \big(\frac{1}{w-q}\big)}{\partial (y(q))^{s+1-n}}
\Big|_{q=u_k}
\\*[-1ex]
&\hspace*{4cm} \times \nabla^n\omega_{0,|I_0|+1}(I_0,\iota u_k)
\prod_{j=1}^s\frac{\omega_{0,|I_i|+1}(I_j,u_k)}{dx(u_k)}\bigg] dw
\\
&=-d_{u_k} \bigg[\Res\displaylimits_{q\to u_k}
\sum_{s=0}^{|I|-2} \sum_{n=0}^{s}
\sum_{I_0\uplus I_1\uplus ...\uplus I_s=I{\setminus} u_k}
\frac{dy(q)}{(w-q) (y(q)-y(u_k))^{s-n+2}}
\\*[-1ex]
&\hspace*{4cm} \times \nabla^n\omega_{0,|I_0|+1}(I_0,\iota u_k)
\prod_{j=1}^s\frac{\omega_{0,|I_j|+1}(I_j,u_k)}{dx(u_k)}\bigg] dw
\\
&=-d_{u_k} \bigg[\Res\displaylimits_{q\to u_k}
\sum_{s=0}^{|I|-2}
\sum_{I_0\uplus I_1\uplus ...\uplus I_s=I{\setminus} u_k}
\Big(\frac{dy(q)}{w-q}-\frac{dy(q)}{w-u_k}\Big)
\frac{1}{(y(q)-y(u_k))^{s+2}}
\\
&\times\prod_{j=1}^s\frac{\omega_{0,|I_i|+1}(I_j,u_k)}{dx(u_k)}
\sum_{n=0}^{s}
(x(\iota u_k)-x(\iota q))^n
\nabla^n\omega_{0,|I_0|+1}(I_0,\iota u_k)
\bigg]dw\;.
\end{align*}
The term $\frac{dy(q)}{w-u_k}$ added in the last step has vanishing
residue (obvious before setting $y(q)=-x(\iota q)$).  It is added in
order to extend the $n$-summation to any $n\geq 0$, giving with
(\ref{nabla-om-taylor}) for $q\mapsto \iota u_k$, $z\mapsto \iota q$
and again (\ref{def:y})
\begin{align}
\Res\displaylimits_{q\to u_k}
\frac{\omega_{0,|I|+1}(I,\iota q)dw}{w-q}
&=-d_{u_k} \bigg[\Res\displaylimits_{q\to u_k}
\sum_{s=0}^{|I|-2}
\sum_{I_0\uplus ...\uplus I_s=I{\setminus} u_k}\!\!\!
\frac{\frac{dy(q)}{w-q}-\frac{dy(q)}{w-u_k}
}{(x(q)-x(u_k))(y(q)-y(u_k))^{s+1}}
\nonumber
\\
&\times
\frac{\omega_{0,|I_0|+1}(I_0,\iota q)}{(-dy(q))}
\prod_{j=1}^s\frac{\omega_{0,|I_j|+1}(I_j,u_k)}{dx(u_k)}\bigg]dw\;.
\label{om-holqu}
\end{align}

The first two lines of (\ref{omIjq-hol}) tell us that 
\begin{align}
&  \sum_{s=1}^{|I'|}
\sum_{I_1\uplus...\uplus I_s=I'}
\frac{dy(q)}{(y(q)-y(u_k))^{s+1}}
\prod_{i=1}^s\frac{\omega_{0,|I_i|+1}(I_i,u_k)}{dx(u_k)}
\nonumber
\\
&=-d_{u_k}^{-1} (\omega_{0,|I'|+2}(I',u_k,\iota q))
+ \text{ terms which are regular at $q=u_k$}\;.\label{duinv-om}
\end{align}
Here the inverse $d_u^{-1}$ of the exterior differential of a $1$-form
$\omega$ is its primitive,
$d_u^{-1} \omega(u) =\int_{u'=\infty}^{u'=u} \omega(u')$.  Inserted
into (\ref{om-holqu})  we confirm with $I'\cup u_k=I_1$
and symmetrisation in $I_1,I_0$ the assertion.
\end{proof}
\end{proposition}

\subsection{Symmetry of the involution identity
  I: \texorpdfstring{$q\to \iota u_k$}{q->\iota uk}
  and \texorpdfstring{$q\to u_k$}{q->uk}}
\label{sec:iota-symm-I}

We consider the $\iota$-reflection of (\ref{eq:flip-om}), 
\begin{align}
&  \omega_{0,|I|+1}(I,q)
  +\omega_{0,|I|+1}(I,\iota q)
\label{eq:flip-om-refl}
\\
&=\sum_{s=2}^{|I|} \sum_{I_1\uplus ...\uplus I_s=I}
\frac{1}{s} \Res\displaylimits_{z\to q}  \Big(
\frac{dx(q) dy(z)}{(x(q)-x(z))^{s}}
\prod_{j=1}^s \frac{\omega_{0,|I_j|+1}(I_j,\iota z)}{dy(z)} 
\Big)\;,
\nonumber
\end{align}
where (\ref{def:y}) is used. We show that the rhs has the same pole
at $q=u_k$ as the original equation (\ref{eq:flip-om}), i.e.\ that the
solution in Proposition~\ref{prop:om-hol2} satisfies
\begin{align}
& 
\Res\displaylimits_{q\to u_k}  \frac{\omega_{0,|I|+1}(I,\iota q)dw}{w-q}
\label{eq:om-hol3}
\\
&=\Res\displaylimits_{q\to u_k}  \frac{dx(q)dw}{w-q}
\sum_{s=2}^{|I|} \sum_{I_1\uplus ...\uplus I_s=I}
\frac{1}{s} \Res\displaylimits_{z\to q}  \Big(
\frac{dy(z)}{(x(q)-x(z))^{s}}
\prod_{j=1}^s \frac{\omega_{0,|I_j|+1}(I_j,\iota z)}{dy(z)} 
\Big)\;.
\nonumber
\end{align}
This is the same as 
\begin{align*}
0&=-\Res\displaylimits_{q\to u_k}  \frac{dx(q)dw}{w-q}
\sum_{s=1}^{|I|} \sum_{I_1\uplus ...\uplus I_s=I}
\frac{1}{s} \Res\displaylimits_{z\to q}  \Big(
\frac{dy(z)}{(x(q)-x(z))^{s}}
\prod_{j=1}^s \frac{\omega_{0,|I_j|+1}(I_j,\iota z)}{dy(z)} 
\Big)
\\
&=\Res\displaylimits_{q\to u_k}  \frac{dx(q)dw}{w-q}
\sum_{s=1}^{|I|} \sum_{\substack{I_1\uplus ...\uplus I_s=I\\u_k\in I_1}}
\Res\displaylimits_{z\to u_k}  \Big(
\frac{dy(z)}{(x(q)-x(z))^{s}}
\prod_{j=1}^s \frac{\omega_{0,|I_j|+1}(I_j,\iota z)}{dy(z)} 
\Big)\;,
\end{align*}
where (\ref{commute-res}) has been used. Fixing $u_k\in I_1$ gives a
factor $s$. We write
$\omega_{0,|I_1|+1}(I_1,\iota z) =d_{u_k}
(d_{u_k}^{-1}\omega_{0,|I_1|+1}(I_1,\iota z))$, move $d_{u_k}$ in
front of the residues and ignore it below. Then we expand the
denominator about $x(z)=x(u_k)$:
\begin{align}
  0  &=\sum_{p=1}^\infty \Res\displaylimits_{q\to u_k}
  \frac{dx(q)dw}{(w-q)(x(q)-x(u_k))^{p}}
\label{kern-0}
\\
&\times \Res\displaylimits_{z\to u_k}  \Big(
\sum_{s=1}^{\min(|I|,p)} \!\!\! 
\binom{p{-}1}{p{-}s} \!\!
\sum_{\substack{I_1\uplus ...\uplus I_s=I\\u_k\in I_1}}\!\!\!\!\!\!
dy(z)(x(z)-x(u_k))^{p-s}  d_{u_k}^{-1}
\Big[\prod_{j=1}^s \frac{\omega_{0,|I_j|+1}(I_j,\iota z)}{dy(z)} 
\Big]\Big)\,.
\nonumber
\end{align}
We will show that already the second line vanishes for every
$p\geq 1$.  For $p=1$ the equation to prove reduces to
$\Res\displaylimits_{z\to u_k}d_{u_k}^{-1}\omega_{0,|I|+1}(I,\iota
z)=0$, which is true by Lemma~\ref{lem:om-hol1} (only higher order
poles at $z=\iota u_k$). Next for $p=2$ we need to show
\begin{align}
0&= \!\Res\displaylimits_{z\to u_k}  (x(z){-}x(u_k))
\Big\{d_{u_k}^{-1}
\omega_{0,|I|+1}(I,\iota z)
+ \!\!\!\!\! \sum_{\substack{I_1\uplus I_2=I\\ u_k\in I_1}} \!\!\!\!
\frac{d_{u_k}^{-1}\omega_{0,|I_1|+1}(I_1,\iota z)\omega_{0,|I_2|+1}(I_2,\iota z)}{
dx(\iota z)(y(\iota z)-y(\iota u))}
\Big\}  .
\label{kern-2}
\end{align}
Indeed by Proposition \ref{prop:om-hol2} the term in braces $\{~\}$ has at
most a first-order pole at $z=u_k$, which is removed by a prefactor
$(x(z)-x(u_k))^n$ for any $n\geq 1$.  Hence (\ref{kern-2}) is true.
Next for $p=3$ we have to show
\begin{align*}
0&= \!\Res\displaylimits_{z\to u_k}  \Big((x(z)-x(u_k))^2
d_{u_k}^{-1}
\omega_{0,|I|+1}(I,\iota z)
\\
&+ 2\sum_{\substack{I_1\uplus I_2=I,~ u_1\in I_1}} 
\frac{(x(z)-x(u_k))}{dx(\iota z)}
d_{u_k}^{-1}\omega_{0,|I_1|+1}(I_1,\iota z)\omega_{0,|I_2|+1}(I_2,\iota z)
\\
&+
\sum_{\substack{I_1\uplus I_2\uplus I_3=I,~ u_1\in I_1}} 
\frac{1}{(dx(\iota z))^2}
d_{u_k}^{-1}\omega_{0,|I_1|+1}(I_1,\iota z)\omega_{0,|I_2|+1}(I_2,\iota z)
\omega_{0,|I_3|+1}(I_3,\iota z)
\Big)  .
\end{align*}
By the argument employed to prove (\ref{kern-2}) this reduces to 
\begin{align}
  0&= \!\Res\displaylimits_{z\to u_k} \frac{(x(z)-x(u_k))}{dx(\iota z)}
  \Big(
\sum_{\substack{I'_1\uplus I'_2=I,~ u_1\in I'_1}} 
d_{u_k}^{-1}\omega_{0,|I'_1|+1}(I'_1,\iota z)\omega_{0,|I'_2|+1}(I'_2,\iota z)
\nonumber
\\
&+
\sum_{\substack{I_1\uplus I_2\uplus I_3=I,~ u_1\in I_1}} 
\frac{
d_{u_k}^{-1}\omega_{0,|I_1|+1}(I_1,\iota z)\omega_{0,|I_2|+1}(I_2,\iota z)
\omega_{0,|I_3|+1}(I_3,\iota z)}{dx(\iota z)(y(\iota z)-y(\iota u))}
\Big)  .
\label{kern-3}
\end{align}
The sum in the first line will include the cases $I_2'=I_3$ and
$I_2'=I_2$. Using again the argument based on Proposition~\ref{prop:om-hol2},
in the first case
$d^{-1}_{u_k}\omega_{0,|I_1'|+1}(I_1',\iota z) +\sum_{I_1\uplus I_2=
  I_1', u_k\in I_1} \frac{d^{-1}_{u_k}\omega_{0,|I_1|+1}(I_1,\iota z)
  \omega_{0,|I_2|+1}(I_2,\iota z)}{dx(\iota z) (y(\iota z)-y(\iota
  u_k))}$ has at most a first-order pole at $z=u_k$. Multiplying
this sum by
$((x(z)-x(u_k)) \frac{\omega_{0,|I_3|+1}(I_3,\iota z)}{dx(\iota z)}$
gives a regular term without residue.  The same is true for
$I_2\leftrightarrow I_3$. This proves (\ref{kern-3}).  The same
argument together with Pascal's triangle structure eventually shows
that the second line of (\ref{kern-0}) vanishes identically for any
$p\geq 1$. In conclusion, (\ref{eq:om-hol3}) is proved, which means that
the rhs of (\ref{eq:flip-om}), minus its
reflection $q\mapsto \iota q$, is holomorphic at every $q=\iota u_k$ (and
then also at $q=u_k$).

\subsection{Linear loop equation}

\label{sec:linloopeq}

Let $\sigma_i$ be the local Galois involution defined
in a neighbourhood of the ramification point 
$\beta_i$. It satisfies
$x(z)=x(\sigma_i(z))$, $\sigma_i(z)\neq z$ for $z\neq \beta_i$ and
$\lim_{z\to \beta_i} \sigma_i(z)=\beta_i$.
\begin{proposition}
 \label{prop-linloop} 
 The meromorphic differentials $\omega_{0,m+1}$ satisfy the linear
 loop equation \cite{Borot:2013lpa}, i.e.\
 $q\mapsto \omega_{0,|I|+1}(I,q)+\omega_{0,|I|+1}(I,\sigma_i(q))$ is
 holomorphic at $q=\beta_i$.
\begin{proof}
  We start from the involution identity (\ref{eq:flip-om-refl}), which
  arises by $q\mapsto \iota q$ from the original equation
  (\ref{eq:flip-om}), and consider
\begin{align*}
&\Res\displaylimits_{q\to \beta_i}  \frac{\omega_{0,|I|+1}(I,q)dw}{w-q}
\\*
&=\sum_{s=1}^{|I|} \frac{dw}{s}   \sum_{I_1\uplus ...\uplus I_s=I}
\Res\displaylimits_{q\to \beta_i}\Res\displaylimits_{z\to q}
\frac{dx(q)dy(z)}{(w-q)(x(q)-x(z))^s}
\prod_{j=1}^s \frac{\omega_{0,|I_j|+1}(I_j,\iota z)}{dy(z)}\;,
\end{align*}
where $\omega_{0,|I|+1}(I,\iota q)$ is included as $s=1$ on the rhs.
Condition (d) in Theorem~\ref{thm:flip}, i.e.\ holomorphicity of
$\omega_{0,|I|+1}(I,\iota q)$ at $q=\beta_i$, implies that the
integrand is regular at $z=\beta_i$, but has a pole at
$z=\sigma_i(q)$. We thus have with commutation rule
(\ref{commute-res-sig})
\begin{align*}
&\Res\displaylimits_{q\to \beta_i}  \frac{\omega_{0,|I|+1}(I,q)dw}{w-q}
\\[-1ex]
&= -\sum_{s=1}^{|I|} \frac{dw}{s}   \sum_{I_1\uplus ...\uplus I_s=I}
\Res\displaylimits_{q\to \beta_i}\Res\displaylimits_{z\to \sigma_i(q)}
\frac{dx(q)dy(z)}{(w-q)(x(q)-x(z))^s}
\prod_{j=1}^s \frac{\omega_{0,|I_j|+1}(I_j,\iota z)}{dy(z)}\;.
\end{align*}
With $x(q)=x(\sigma_i(q))$ and $dx(q)=dx(\sigma_i(q))$ the inner
integral evaluates to $\omega_{0,|I|+1}(I,\sigma_i(q))$, and we end up
in
$\Res\displaylimits_{q\to \beta}
\frac{\omega_{0,|I|+1}(I,q)+\omega_{0,|I|+1}(I,\sigma_i(q))}{w-q}dw=0$.
\end{proof}
\end{proposition}

\begin{remark} \label{rem:id}
  From \eqref{Id}
  and the expansion $\sigma_i(q)-\beta_i=-(\beta_i-q)
  +\mathcal{O}((q-\beta_i)^2)$ 
  we conclude
 \begin{align*}
   \Res\displaylimits_{q\to \beta_i}
   \frac{\omega_{0,|I|+1}(I,q)+\omega_{0,|I|+1}(I,\sigma_i(q))}{q-\beta_i}
   &=   \Res\displaylimits_{q\to \beta_i}
   \frac{\omega_{0,|I|+1}(I,q)
     +\omega_{0,|I|+1}(I,\sigma_i(q))}{\sigma_i(q)-\beta_i}
   \\
   &=  -\Res\displaylimits_{q\to \beta_i}
   \frac{\omega_{0,|I|+1}(I,q)
     +\omega_{0,|I|+1}(I,\sigma_i(q))}{q-\beta_i}\;.
 \end{align*}
Hence, 
 $\frac{\omega_{0,|I|+1}(I,q)+\omega_{0,|I|+1}(I,\sigma_i(q))}{q-\beta_i}$
 is regular at $q=\beta_i$, which means that
  $\omega_{0,|I|+1}(I,q)+\omega_{0,|I|+1}(I,\sigma_i(q))$ has at
  least a first-order zero at $q=\beta_i$.
\end{remark}

\subsection{The recursion kernel}

\label{sec:tr-kernel}

We start from Lemma~\ref{lem:om-onenabla} for $q\mapsto \iota q$ where
(\ref{def:y}) is taken into account:
\begin{align}
  &\omega_{0,|I|+1}(I,q)+\omega_{0,|I|+1}(I,\iota q)
  \label{omega-q-iotaq}
  \\
&= 
dx(q) \sum_{s=1}^{|I|-1} \sum_{I_0\uplus I_1\uplus ...\uplus I_s=I}
\nabla^{s}\omega_{0,|I_0|+1}(I_0,\iota q)
\prod_{j=1}^s \frac{\omega_{0,|I_j|+1}(I_j,q)}{dx(q)}
\nonumber
\\
&=dx(q) \sum_{s=1}^{|I|-1} \sum_{I_0\uplus I_1\uplus ...\uplus I_s=I}
\Res\displaylimits_{z\to q}
\frac{\omega_{0,|I_0|+1}(I_0,\iota z)}{(x(q)-x(z))(y(z)-y(q))^s}
\prod_{j=1}^s \frac{\omega_{0,|I_j|+1}(I_j,q)}{dx(q)}\;.
\nonumber
\end{align}
We introduce \enlargethispage{0.5mm}
\begin{align}
\mathfrak{W}_{a,s,s'}(I;q)
&:=\sum_{\substack{I_0\uplus I_1\uplus...\uplus I_s
    \uplus I'_1\uplus...\uplus I'_{s'} \\ I_0=\emptyset \text{ for } a=0}}
\!\!\!\!\!\!
\big(\delta_{a,0}+(1-\delta_{a,0})
\nabla^a \omega_{0,|I_0|+1}(I_0;\iota q)\big)
\nonumber
\\*[-3.2ex]
&\hspace*{3.7cm} \times
\prod_{k=1}^{s} \frac{\omega_{0,|I_k|+1}(I_k,q)}{dx(q)}
\prod_{j=1}^{s'}\frac{\omega_{0,|I_j'|+1}(I'_j,\sigma_i(q))}{dx(\sigma_i(q))}\;,
\nonumber
 \\
\mathfrak{B}_{a,a',s,s'}(I;q,z)&
:=\sum_{I_0\uplus I_1\uplus...\uplus I_s\uplus I'_1\uplus...\uplus I'_{s'}=I }
\frac{\omega_{0,|I_0|+1}(I_0,\iota z)}{(x(q)-x(z))}
\nonumber
\\*[-1ex]
&\hspace*{0.5cm}
\times   \frac{  \prod_{k=1}^{s} \omega_{0,|I_k|+1}(I_k,q)
    \prod_{j=1}^{s'} \omega_{0,|I'_j|+1}(I'_j,\sigma_i(q))
  }{(dx(q))^s(dx(\sigma_i(q)))^{s'} (y(z)-y(q))^a(y(z)-y(\sigma_i(q)))^{a'}}\;.
\label{frakW-frakB}
\end{align}  
These are functions of $q$ and $1$-forms in every variable in $I$, 
and $\mathfrak{B}_{a,a',s,s'}(I;q,z)$ is also a $1$-form in $z$.
A lengthy calculation gives the following important tool:
\begin{lemma}
\label{lem:frakW-rec}
Residues of $\mathfrak{W}_{a,s,s'}$ satisfy for $0<a\leq s$ 
\begin{align}
&\Res\displaylimits_{q\to \beta_i}  
\frac{f_{a,s,s'}(q) dx(q)}{w-q}  
\mathfrak{W}_{a,s,s'}(I;q)
\label{eq:res-frakW}
\\
&=
\Res\displaylimits_{z\to \beta_i}
\Res\displaylimits_{q\to \beta_i}  
\frac{f_{a,s,s'}(q) dx(q)}{w-q}\Big(\mathfrak{B}_{a,0,s,s'}(I;q,z)
+ \sum_{a'=1}^{|I|-s-s'-1} 
\frac{\mathfrak{B}_{0,a',s,s'+a'}(I;q,z)}{ (y(\sigma_i(q))-y(q))^a}
\Big)
\nonumber
\\
&+\Res\displaylimits_{q\to \beta_i}  
\frac{f_{a,s,s'}(q) dx(q)}{w-q}
\Big(-\frac{\mathfrak{W}_{0,s,s'+1}(I;q)}{(y(\sigma_i(q))-y(q))^a}
-\sum_{a'=1}^{|I|-s-s'-1} 
\frac{(-1)^{a'}\mathfrak{W}_{0,s+1,s'+a'}(I;q)}{(y(\sigma_i(q))-y(q))^{a+a'}}
\nonumber
\\
&\qquad +\sum_{a'=1}^{|I|-s-s'-2} \;
\sum_{a''=1}^{|I|-s-s'-a'-1} 
\frac{(-1)^{a'}\mathfrak{W}_{a'',s+a'',s'+a'}(I;q)}{
  (y(\sigma_i(q))-y(q))^{a+a'}}\Big)\;,
\nonumber
\end{align}
for any function $f_{a,s,s'}$ meromorphic in a neighbourhood of $\beta_i$.
\begin{proof}
 We consider
for a fixed partition $I_0\uplus I_1\uplus...\uplus {I_s}
\uplus I'_1\uplus...\uplus {I'_{s'}}=I$ and some $0< a\leq s$ the residue
\begin{align}
&\Res\displaylimits_{q\to \beta_i}  
\frac{f_{a,s,s'}(q) dx(q)}{w-q}  
\nabla^a\omega_{0,|I_0|+1}(I_0,\iota q)
\prod_{k=1}^{s} \frac{\omega_{0,|I_k|+1}(I_k,q)}{dx(q)}
\prod_{j=1}^{s'}\frac{\omega_{0,|I_j'|+1}(I'_j,\sigma_i(q))}{dx(\sigma_i(q))}
\nonumber
\\
&\equiv
\Res\displaylimits_{q\to \beta_i}  
\frac{f_{a,s,s'}(q) dx(q)}{w-q}
\Res\displaylimits_{z\to q} \frac{\omega_{0,|I_0|+1}(I_0,\iota z)
  \mbox{\footnotesize$\displaystyle\prod_{k=1}^{s}$}  \omega_{0,|I_k|+1}(I_k,q)
  \mbox{\footnotesize$\displaystyle\prod_{j=1}^{s'}$}
  \omega_{0,|I_j'|+1}(I'_j,\sigma_i(q))}{
(x(q)-x(z))(y(z)-y(q))^a(dx(q))^s(dx(\sigma_i(q)))^{s'}}   
\nonumber
\\
&=
\Res\displaylimits_{z\to \beta_i}
\Res\displaylimits_{q\to \beta_i}  
\frac{f_{a,s,s'}(q) dx(q)}{w-q}
\frac{\omega_{0,|I_0|+1}(I_0,\iota z)
  \mbox{\footnotesize$\displaystyle\prod_{k=1}^{s}$}  \omega_{0,|I_k|+1}(I_k,q)
  \mbox{\footnotesize$\displaystyle\prod_{j=1}^{s'}$}
  \omega_{0,|I_j'|+1}(I'_j,\sigma_i(q))}{
(x(q)-x(z))(y(z)-y(q))^a(dx(q))^s(dx(\sigma_i(q)))^{s'}}   
\nonumber
\\
&-\Res\displaylimits_{q\to \beta_i}  
\frac{f_{a,s,s'}(q) dx(q)}{w-q}
\Res\displaylimits_{z\to \sigma_i(q)} \frac{\omega_{0,|I_0|+1}(I_0,\iota z)
  \mbox{\footnotesize$\displaystyle\prod_{k=1}^{s}$}  \omega_{0,|I_k|+1}(I_k,q)
  \mbox{\footnotesize$\displaystyle\prod_{j=1}^{s'}$}  
\omega_{0,|I_j'|+1}(I'_j,\sigma_i(q))}{
(x(q)-x(z))(y(z)-y(q))^a(dx(q))^s(dx(\sigma_i(q)))^{s'}}   \;.
\nonumber
\end{align}
We have used (\ref{def:nabla-om}) and (\ref{commute-res-sig}) and the
fact that the integrand is regular at $z=\beta_i$. The residue at
$z=\sigma_i(q)$ in the last line can be evaluated immediately and
gives rise to the function
$-\frac{\omega_{0,|I_0|+1}(I_0,\iota \sigma_i(q))}{dx(\sigma_i(q))}$
for which we insert (\ref{omega-q-iotaq}) at $q\mapsto \sigma_i(q)$:
\begin{align}
&\Res\displaylimits_{q\to \beta_i}  
\frac{f_{a,s,s'}(q) dx(q)}{w-q}  
\nabla^a\omega_{0,|I_0|+1}(I_0,\iota q)
\prod_{k=1}^{s} \frac{\omega_{0,|I_k|+1}(I_k,q)}{dx(q)}
\prod_{j=1}^{s'}\frac{\omega_{0,|I_j'|+1}(I'_j,\sigma_i(q))}{dx(\sigma_i(q))}
\nonumber
\\
&=
\Res\displaylimits_{z\to \beta_i}
\Res\displaylimits_{q\to \beta_i}  
\frac{f_{a,s,s'}(q) dx(q)}{w-q}
\frac{\omega_{0,|I_0|+1}(I_0,\iota z)
  \mbox{\footnotesize$\displaystyle\prod_{k=1}^{s}$}  \omega_{0,|I_k|+1}(I_k,q)
  \mbox{\footnotesize$\displaystyle\prod_{j=1}^{s'}$}
  \omega_{0,|I_j'|+1}(I'_j,\sigma_i(q))
}{
(x(q)-x(z))(y(z)-y(q))^a(dx(q))^s(dx(\sigma_i(q)))^{s'}}   
\nonumber
\\
&-\Res\displaylimits_{q\to \beta_i}  
\frac{f_{a,s,s'}(q) dx(q)}{w-q}
\frac{\omega_{0,|I_0|+1}(I_0,\sigma_i(q))
  \mbox{\footnotesize$\displaystyle\prod_{k=1}^{s}$}  \omega_{0,|I_k|+1}(I_k,q)
  \mbox{\footnotesize$\displaystyle\prod_{j=1}^{s'}$}  \omega_{0,|I_j'|+1}(I'_j,\sigma_i(q))
}{
dx(\sigma_i(q)) (y(\sigma_i(q))-y(q))^a(dx(q))^s(dx(\sigma_i(q)))^{s'}}   
\nonumber
\\
&+
\sum_{a'=1}^{|I_0|-1} \sum_{I_0''\uplus I_1''\uplus ...\uplus I_{a'}''=I_0}
\Res\displaylimits_{q\to \beta_i}  
\frac{f_{a,s,s'}(q) dx(q)}{w-q} \Res\displaylimits_{z\to \sigma_i(q)}
\bigg(\frac{\omega_{0,|I_0|+1}(I_0'',\iota z)
}{(x(\sigma_i(q))-x(z))}
\tag{*}
\\*[-1.5ex]
&\qquad \times
\frac{  \mbox{\footnotesize$\displaystyle\prod_{k=1}^{s}$}
  \omega_{0,|I_k|+1}(I_k,q)
  \mbox{\footnotesize$\displaystyle\prod_{j=1}^{s'}$}
  \omega_{0,|I_j'|+1}(I'_j,\sigma_i(q))
  \mbox{\footnotesize$\displaystyle\prod_{j=1}^{a'}$}
  \omega_{0,|I_j''|+1}(I''_j,\sigma_i(q))
}{  (y(\sigma_i(q))-y(q))^a(dx(q))^s(dx(\sigma_i(q)))^{s'}
  (y(z)-y(\sigma_i(q)))^{a'}}\bigg)   \;.
\tag{*}
\end{align}
We process the last two lines (*) in the same manner, i.e.\ commute
the two residues according to (\ref{commute-res-sig}). There is again
no contribution of a residue at $z=\beta_i$, but now an additional
residue at $z=q$ arises.  The resulting term
$\frac{\omega_{0,|I_0|+1}(I_0'',\iota q)}{dx(q)}$ is expressed via
(\ref{omega-q-iotaq}):
\begin{align*}
  (*)&
= \sum_{a'=1}^{|I_0|-1} \sum_{I_0''\uplus I_1''\uplus ...\uplus I_{a'}''=I_0}
\Res\displaylimits_{z\to \beta_i}  \Res\displaylimits_{q\to \beta_i}  
\frac{f_{a,s,s'}(q) dx(q)}{w-q} 
\bigg(\frac{\omega_{0,|I_0|+1}(I_0'',\iota z)
}{(x(\sigma_i(q))-x(z))}
\\*[-1.5ex]
&\quad\quad\times
\frac{  \mbox{\footnotesize$\displaystyle\prod_{k=1}^{s}$}
  \omega_{0,|I_k|+1}(I_k,q)
  \mbox{\footnotesize$\displaystyle\prod_{j=1}^{s'}$}
  \omega_{0,|I_j'|+1}(I'_j,\sigma_i(q))
  \mbox{\footnotesize$\displaystyle\prod_{j=1}^{a'}$}
  \omega_{0,|I_j''|+1}(I''_j,\sigma_i(q))
}{  (y(\sigma_i(q))-y(q))^a(dx(q))^s(dx(\sigma_i(q)))^{s'}
  (y(z)-y(\sigma_i(q)))^{a'}}\bigg)   
\\
&-\sum_{a'=1}^{|I_0|-1} \sum_{I_0''\uplus I_1''\uplus ...\uplus I_{a'}''=I_0}
\Res\displaylimits_{q\to \beta_i}  
\frac{f_{a,s,s'}(q) dx(q)}{w-q} 
\bigg(\frac{\omega_{0,|I_0|+1}(I_0'',q)
}{dx(q)}
\\*[-1ex]
&\quad\quad\times
\frac{(-1)^{a'}  \mbox{\footnotesize$\displaystyle\prod_{k=1}^{s}$}
  \omega_{0,|I_k|+1}(I_k,q)
  \mbox{\footnotesize$\displaystyle\prod_{j=1}^{s'}$}
  \omega_{0,|I_j'|+1}(I'_j,\sigma_i(q))
  \mbox{\footnotesize$\displaystyle\prod_{j=1}^{a'}$}
  \omega_{0,|I_j''|+1}(I''_j,\sigma_i(q))
}{  (y(\sigma_i(q))-y(q))^{a+a'}(dx(q))^s(dx(\sigma_i(q)))^{s'}}\bigg)   
\\
&+\sum_{a'=1}^{|I_0|-1} \sum_{I_0''\uplus I_1''\uplus ...\uplus I_{a'}''=I_0}
\sum_{a''=1}^{|I_0''|-1} \sum_{I_0'''\uplus I_1'''\uplus ...\uplus I_{a''}'''=I_0''}
\Res\displaylimits_{q\to \beta_i}  
\frac{f_{a,s,s'}(q) dx(q)}{w-q}
\\*[-0.5ex]
&\times
\bigg(
\nabla^{a''}\omega_{0,|I'''_0|+1}(I_0''',\iota q)
\frac{(-1)^{a'} 
  \mbox{\footnotesize$\displaystyle\prod_{k=1}^{a''}$}
  \omega_{0,|I_k'''|+1}(I_k''',q)
  \mbox{\footnotesize$\displaystyle\prod_{j=1}^{a'}$}
  \omega_{0,|I_j''|+1}(I''_j,\sigma_i(q))
}{  (y(\sigma_i(q))-y(q))^{a+a'}}
\\*[-1.5ex]
&\quad\quad \times \frac{ \mbox{\footnotesize$\displaystyle\prod_{k=1}^{s}$}
  \omega_{0,|I_k|+1}(I_k,q)
  \mbox{\footnotesize$\displaystyle\prod_{j=1}^{s'}$}
  \omega_{0,|I_j'|+1}(I'_j,\sigma_i(q))
}{(dx(q))^s(dx(\sigma_i(q)))^{s'}}\bigg)\;.
\end{align*}
This is inserted back into the equation we started with. We sum over
all partitions
$I_0\uplus I_1\uplus...\uplus {I_s} \uplus I'_1\uplus...\uplus
{I'_{s'}}=I$ for fixed $s,s',a$ and express the result in terms of
$\mathfrak{W},\mathfrak{B}$ introduced in (\ref{frakW-frakB}). The
result is (\ref{eq:res-frakW}).
\end{proof}
\end{lemma}
Lemma \ref{lem:frakW-rec} is our main tool to evaluate the polar part
of (\ref{omega-q-iotaq}) at $q=\beta_i$. Taking condition (d) of
Theorem~\ref{thm:flip} into account, we need to evaluate
\[
  \Res\displaylimits_{q\to \beta_i} \frac{\omega_{0,|I|+1}(I,q)dw}{w-q}
=\sum_{s=1}^{|I|-1}
\Res\displaylimits_{q\to \beta_i} \frac{\mathfrak{W}_{s,s,0}(I;q) dx(q)dw}{w-q}\;.
\]
In a first (also very lengthy) step we show:
\begin{lemma}
\label{lem:WBnew}  
  \begin{align}
    0&=
    \Res\displaylimits_{q\to \beta_i}  
\frac{dx(q)dw}{w-q}
\Big(\sum_{s=1}^{|I|-1}\mathfrak{W}_{s,s,0}(I;q)
+
\frac{\mathfrak{W}_{0,1,1}(I;q)}{y(\sigma_i(q))-y(q)}\Big)
\label{eq:Ws0-to11}
\\
&-\Res\displaylimits_{z\to \beta_i}\Res\displaylimits_{q\to \beta_i}  
\frac{dx(q)dw}{w-q}\Big(\sum_{s=1}^{|I|-1}\sum_{s'=0}^{|I|-s-1} \!\!
\mathfrak{B}_{s,s',s,s'}(I;q,z)
+\!\!\sum_{s=1}^{|I|-2}\sum_{s'=1}^{|I|-s-1}
\frac{\mathfrak{B}_{s,s'-1,s,s'}(I;q,z)}{y(\sigma_i(q))-y(q)}\Big).
\nonumber
\end{align}
\begin{proof}
  We express
  $\sum_{s=1}^{|I|-1} \Res\displaylimits_{q\to \beta_i}
  \frac{dx(q)dw}{w-q} \mathfrak{W}_{s,s,0}(I;q)$ via
  (\ref{eq:res-frakW}) at $s'=0$, $a=s$ and $f_{a,s,s'}\equiv 1$.  In
  the third line of (\ref{eq:res-frakW}), the case $a'=1$ of the
  second term cancels, when summing over $s$, every first term except
  for the single term with $s=1$ and $s'=0$. This surviving term with
  $s=a=1$ is the last term in the first line of (\ref{eq:Ws0-to11}).
  When subtracting the second line of (\ref{eq:Ws0-to11}), the term
  with $s'=0$ in (\ref{eq:Ws0-to11}) cancels directly, and then the
  term with $s=1$ (and any $s'\geq 1$) cancels after reordering
  partial fractions.  After renaming the parameters we arrive at
\begin{align}
\eqref{eq:Ws0-to11}_{\text{rhs}}
&= \Res\displaylimits_{q\to \beta_i}  
\frac{dx(q)dw}{w-q}
\Big(
-\sum_{s=2}^{|I|-2} \sum_{s'=2}^{|I|-s}(-1)^{s'}
\frac{\mathfrak{W}_{0,s,s'}(I;q)}{(y(\sigma_i(q))-y(q))^{s+s'-1}}
\label{eq:Ws0-to11-a}
\\
&\qquad\qquad
+\sum_{s=2}^{|I|-2} \sum_{s'=1}^{|I|-s-1} \sum_{a=1}^{s-1} (-1)^{s'}
\frac{\mathfrak{W}_{a,s,s'}(I;q)}{(y(\sigma_i(q))-y(q))^{s+s'-a}}\Big)
\nonumber
\\
&+ \Res\displaylimits_{z\to \beta_i} \Res\displaylimits_{q\to \beta_i}  
\frac{dx(q)dw}{w-q}
\sum_{s=2}^{|I|-2}\sum_{s'=1}^{|I|-s-1} \!\!\!
\Big(\frac{\mathfrak{B}_{0,s',s,s'}(I;q,z)}{(y(\sigma_i(q))-y(q))^s}
-\frac{\mathfrak{B}_{s-1,s',s,s'}(I;q,z)}{y(\sigma_i(q))-y(q)} \Big) .
\nonumber
\end{align}
In the second term of the last line we apply repeatedly the identity
\begin{align*}
  \frac{\mathfrak{B}_{a,a',s,s'}(I;z,q)}{(y(\sigma_i(q))-y(q))^{s+s'-a-a'}}
  &=  \frac{\mathfrak{B}_{a-1,a',s,s'}(I;q,z)
    -\mathfrak{B}_{a,a'-1,s,s'}(I;q,z)}{(y(\sigma_i(q))-y(q))^{s+s'-a-a'+1}}
\end{align*}
to express
$\frac{\mathfrak{B}_{s-1,s',s,s'}(I;q)}{(y(\sigma_i(q))-y(q))}$ as
linear combination of $\mathfrak{B}_{a,0,s,s'}(I;q,z)$ and
$\mathfrak{B}_{0,a',s,s'}(I;q,z)$.  The coefficient of
$\mathfrak{B}_{a,0,s,s'}(I;q,z)$ in this expansion is the number of
paths made of steps up or right from $(a,0)$ to $(s-1,s')$ with a
first step right.  This is the same as the number
$\binom{s+s'-a-2}{s'-1}$ of words of $s'-1$ letters $R$ and $s-1-a$
letters $U$.  Similarly, the coefficient of
$\mathfrak{B}_{0,a',s,s'}(I;q,z)$ in this expansion is the number of
up-right paths from $(0,a')$ to $(s-1,s')$ with a first step up. This
is the same as the number $\binom{s+s'-a'-2}{s-2}$ of words of $s-2$
letters $U$ and $s'-a'$ letters $R$.  A right step comes with a factor
$(-1)$. We thus get
\begin{align*}
\frac{\mathfrak{B}_{s-1,s',s,s'}(I;q,z)}{(y(\sigma_i(q))-y(q))} 
&=
\sum_{a=1}^{s-1} \binom{s+s'-a-2}{s'-1}
\frac{(-1)^{s'}\mathfrak{B}_{a,0,s,s'}(I;q,z)}{(y(\sigma_i(q))-y(q))^{s+s'-a}}
\\
&+\sum_{a'=1}^{s'} \binom{s+s'-a'-2}{s-2}
\frac{(-1)^{s'-a'}\mathfrak{B}_{0,a',s,s'}(I;q,z)}{(y(\sigma_i(q))-y(q))^{s+s'-a'}}\;.
\end{align*}
The term with $a'=s'$ cancels the first term of the last line of (\ref{eq:Ws0-to11-a})
so that we end up in the following equation in which $k\equiv 0$:
\begin{align}
\eqref{eq:Ws0-to11}_{\text{rhs}}
  &= \Res\displaylimits_{q\to \beta_i}  
\frac{dx(q)}{w-q}
\bigg\{
-\!\!\!\!\sum_{s=2+k}^{|I|-2-k} \sum_{s'=2+k}^{|I|-s} \!\!
\binom{s{-}2}{k} \binom{s'{-}2}{k} 
\frac{(-1)^{s'}\mathfrak{W}_{0,s,s'}(I;q)}{(y(\sigma_i(q))-y(q))^{s+s'-1}}
\tag{\dag}
\\
&+\sum_{s=2+k}^{|I|-2-k} \sum_{s'=1+k}^{|I|-s-1}\sum_{a=1}^{s-1-k}
\binom{{s{-}a{-}1}}{k}\! \binom{{s'{-}1}}{k}
\frac{(-1)^{s'}\mathfrak{W}_{a,s,s'}(I;q)}{(y(\sigma_i(q))-y(q))^{s+s'-a}}\bigg\}dw
\tag{\ddag}
\\
&- \Res\displaylimits_{z\to \beta_i} \Res\displaylimits_{q\to \beta_i}  
\frac{dx(q)}{w-q}
\sum_{s=2}^{|I|-2}\sum_{s'=1}^{|I|-s-1} \!\!\! 
\bigg\{
\sum_{a=1}^{s-1} \binom{s{+}s'{-}a{-}2}{s'{-}1}
\frac{(-1)^{s'}\mathfrak{B}_{a,0,s,s'}(I;q,z)}{(y(\sigma_i(q))-y(q))^{s+s'-a}}
\nonumber
\\
&\qquad +\sum_{a'=1}^{s'-1} \binom{s{+}s'{-}a'{-}2}{s{-}2}
\frac{(-1)^{s'-a'}\mathfrak{B}_{0,a',0,s,s'}(I;q,z)}{
  (y(\sigma_i(q))-y(q))^{s+s'-a'}}\bigg\}dw\;.
\label{frakWB-b}
\end{align}
Next we process the line ($\ddag$) of (\ref{frakWB-b}) via
(\ref{eq:res-frakW}).  With the exception of one term the
`hockey-stick identity'
$\sum_{a=1}^{s-k-1} \binom{s{-}a{-}1}{k} =\binom{s-1}{k+1}$ occurs:
\begin{align}
\eqref{frakWB-b}_\ddag
&= \Res\displaylimits_{q\to \beta_i}  
\frac{dx(q)dw}{w-q}
\sum_{s=2+k}^{|I|-2-k} \sum_{s'=1+k}^{|I|-s-1}
\bigg\{
-\binom{s{-}1}{k{+}1} \binom{s'{-}1}{k}
\frac{(-1)^{s'}\mathfrak{W}_{0,s,s'+1}(I;q)}{(y(\sigma_i(q))-y(q))^{s+s'}}
\nonumber
\\
&\qquad -\binom{s{-}1}{k{+}1} \binom{s'{-}1}{k}
\sum_{a'=1}^{|I|-s-s'-1} 
\frac{(-1)^{s'+a'}\mathfrak{W}_{0,s+1,s'+a'}(I;q)}{(y(\sigma_i(q))-y(q))^{s+s'+a'}}
\nonumber
\\
&\qquad +
\binom{s{-}1}{k{+}1} \binom{s'{-}1}{k}
\sum_{a'=1}^{|I|-s-s'-2}\sum_{a''=1}^{|I|-s-s'-a'-1}
\frac{(-1)^{s'+a'}
\mathfrak{W}_{a'',s+a'',s'+a'}(I;q)}{(y(\sigma_i(q))-y(q))^{s+s'+a'}}\bigg\}
\nonumber
\\
&+\Res\displaylimits_{z\to \beta_i}\Res\displaylimits_{q\to \beta_i}  
\frac{dx(q)dw}{w-q}
\sum_{s=2+k}^{|I|-2-k} \sum_{s'=1+k}^{|I|-s-1}\bigg\{
\nonumber
\\
& 
\qquad \sum_{a=1}^{s-1-k}
\binom{s{-}a{-}1}{k} \binom{s'{-}1}{k}
\frac{(-1)^{s'}\mathfrak{B}_{a,0,s,s'}(I;q,z)}{(y(\sigma_i(q))-y(q))^{s+s'-a}}
\nonumber
\\
&\qquad +
\binom{s{-}1}{k{+}1} \binom{s'{-}1}{k}
\sum_{a'=1}^{|I|-s-s'-1}
\frac{(-1)^{s'}\mathfrak{B}_{0,a',s,s'+a'}(I;q,z)}{(y(\sigma_i(q))-y(q))^{s+s'}}\bigg\}\;.
\label{frakWB-c}
\end{align}  
The following steps are performed:
\begin{itemize}
\item In the first line we shift $s'+1\mapsto s' \in [2{+}k..|I|{-}s]$.
\item In the second line we shift $s+1\mapsto s\in [3{+}k..|I|{-}k{-}1]$.
  Then we rename $s'{+}a'\mapsto s' \in [2{+}k..|I|{-}s]$  
  and sum over $a'\in [1..s'{-}k{-}1]$. Recall
  $\sum_{a'=1}^{s'-k-1} \binom{s'{-}a'{-}1}{k}=\binom{s'-1}{k+1}$.
  The new ranges restrict 
  $s\in [3{+}k..|I|{-}k{-}2]$.
  
\item In the third line we rename
  $s'{+}a'\mapsto s' \in [2{+}k..|I|{-}s]$ and sum over
  $a'\in [1..s'{-}k{-}1]$. This gives
  $\sum_{a'=1}^{s'-k-1} \binom{s'{-}a'{-}1}{k}=\binom{s'-1}{k+1}$. We
  also rename $s{+}a''\mapsto s \in [3{+}k..|I|{-}k{-}2]$ and keep the
  sum over $a''\mapsto a \in [1..s-2-k]$.

\item In the final line we rename $s'+a'\mapsto s'\in [2+k..I-s-1]$ and sum over $a'\in [1..s'{-}k{-}1]$. 
\end{itemize}
With the Pascal triangle identity
$\binom{s{-}1}{k{+}1}-\binom{s{-}2}{k} =\binom{s{-}2}{k{+}1}$ and the
corresponding adjustments of the ranges for $s,s'$ we find that the
first two lines $(\dag,\ddag)$ of (\ref{frakWB-b}), where $k=0$, equal
the same two lines $\eqref{frakWB-b}_{\dag+\ddag}$ with $k=1$, plus
the iterated residue in the last three lines of (\ref{frakWB-c}), first
for $k=0$. Iterating this procedure until $s\geq 2+k$ and $s'\geq 1+k$
becomes incompatible with the size $|I|$ gives for the first two lines
of (\ref{frakWB-b}) the identity
\begin{align}
\eqref{frakWB-b}_{\dag,\ddag}
&=
\Res\displaylimits_{z\to \beta_i}\Res\displaylimits_{q\to \beta_i}  
\frac{dx(q)dw}{w-q}
\sum_{k=0}^{[|I|/2]-2} \sum_{s=2+k}^{|I|-2-k} \sum_{s'=1+k}^{|I|-s-1} \bigg\{
\nonumber
\\
& 
\sum_{a=1}^{s-1-k}
\binom{s{-}a{-}1}{k} \binom{s'{-}1}{k}
\frac{(-1)^{s'}\mathfrak{B}_{a,0,s,s'}(I;q,z)}{(y(\sigma_i(q))-y(q))^{s+s'-a}}
\nonumber
\\
&+
\sum_{a'=1}^{s'-1-k}
\binom{s{-}1}{k{+}1} \binom{s'{-}{a}'{-}1}{k}
\frac{(-1)^{s'-a'}\mathfrak{B}_{0,a',s,s'}(I;q,z)}{(y(\sigma_i(q))-y(q))^{s+s'-a'}}\bigg\}\;.
\nonumber
\end{align}
Now we change the summation order and sum first over $k$. With
\[
  \sum_{k=0}^{\min(n,s-1)} \binom{n}{k}\binom{s{-}1}{k}=\binom{n{+}s{-}1}{s{-}1},~~
  \sum_{k=0}^{\min(n-1,s-1)} \binom{n}{k{+}1}\binom{s{-}1}{k}=\binom{n{+}s{-}1}{n{-}1}
\]
(e.g.\ \cite{Gould}[Vol.~4, eq.~(6.69)+(6.70)]) we conclude
\begin{align*}
&\eqref{frakWB-b}_{\dag,\ddag}
\\*[-2ex]
&=
\Res\displaylimits_{z\to \beta_i}\Res\displaylimits_{q\to \beta_i}  
\frac{dx(q)dw}{w-q}
\sum_{s=2}^{|I|-2} \sum_{s'=1}^{|I|-s-1} \bigg\{
\sum_{a=1}^{s-1}
\binom{s{+}s'{-}a{-}2}{s'-1} 
\frac{(-1)^{s'}\mathfrak{B}_{a,0,s,s'}(I;q,z)}{(y(\sigma_i(q))-y(q))^{s+s'-a}}
\nonumber
\\
&\hspace*{4cm} +\sum_{a'=1}^{s'-1}
\binom{s{+}s'{-}a'{-}2}{s-2}
\frac{(-1)^{s'-a'}\mathfrak{B}_{0,a',s,s'}(I;q,z)}{
  (y(\sigma_i(q))-y(q))^{s+s'-a'}}\bigg\}\;.
\nonumber
\end{align*}
Therefore, (\ref{frakWB-b}) and hence
the rhs of (\ref{eq:Ws0-to11}) are identically zero.
\end{proof}
\end{lemma}

We will prove by induction that the second line of (\ref{eq:Ws0-to11}) vanishes identically. 
This requires a rearrangement of the forms $\mathfrak{B}$. To simplify notation we
introduce the split operator 
\begin{align}
\mathsf{S}\omega_{0,|I|+1}(I,q)
&:= \sum_{I_1\uplus I_2=I}
\frac{\omega_{0,|I_1|+1}(I_1,q)\omega_{0,|I_2|+1}(I_2,\sigma_i(q))}{
 dx(\sigma_i(q)) (y(\sigma_i(q))-y(q))}
\label{calS2}
\end{align}
with $\mathsf{S}\omega_{0,2}(u,q)=0$. Then (\ref{eq:Ws0-to11})
can be written with (\ref{omega-q-iotaq}) as
\begin{align}
0  &=\Res\displaylimits_{q\to \beta_i}  
\frac{\omega_{0,|I|+1}(I,q)+ \mathsf{S}\omega_{0,|I|+1}(I,q)}{w-q}dw
\label{res-omSom}
\\
&- \sum_{\substack{I_0\uplus I' \uplus I''=I\\ \text{possibly } I''= \emptyset}}
\Res\displaylimits_{z\to \beta_i} 
\Res\displaylimits_{q\to \beta_i}  
\frac{\omega_{0,|I_0|+1}(I_0,\iota z)}{x(q)-x(z)}
\Big\{
\frac{\omega_{0,|I'|+1}(I',q)+ \mathsf{S}\omega_{0,|I'|+1}(I',q)}{(w-q)(y(z)-y(q))}
\nonumber
\\*[-2ex]
&\hspace*{8cm} \times \tilde{\mathfrak{B}}(I'';q,z) \Big\}dw
\nonumber
\end{align}
where $\tilde{\mathfrak{B}}(\emptyset ;q,z)=1$ and for $I''\neq \emptyset$
\begin{align}
  &\tilde{\mathfrak{B}}(I'';q,z)
  \label{frakBqz-1}
  \\*
  &= \!\!\sum_{s=1}^{|I''|}\sum_{s_0=0}^{s}
\sum_{I_1\uplus...\uplus I_s=I''}
\prod_{j=1}^{s_0} \frac{\omega_{0,|I_j|+1}(I_j,q)}{dx(q) (y(z)-y(q))}
\prod_{j=s_0+1}^{s} \frac{\omega_{0,|I_j|+1}(I_j,\sigma_i(q))  }{dx(\sigma_i(q)) (y(z)-y(\sigma_i(q)))}
\;.
\nonumber
\end{align}
We claim that this expression  can be reordered into
\begin{align}
\tilde{\mathfrak{B}}(I'';q,z)&=\sum_{p=1}^{|I''|}\sum_{I_1\uplus ...\uplus I_p=I''}
  \prod_{j=1}^{p} \Big\{
  \frac{\omega_{0,|I_j|+1}(I_j,q) +\omega_{0,|I_j|+1}(I_j,\sigma_i(q))}{dx(\sigma_i(q)) 
    (y(z)-y(\sigma_i(q)))}
  \nonumber
  \\
  &-\frac{y(\sigma_i(q))-y(q)}{dx(q)}
  \frac{\omega_{0,|I_j|+1}(I_j,q)+\mathsf{S}\omega_{0,|I_j|+1}(I_j,q)
    }{(y(z)-y(q))(y(z)-y(\sigma_i(q)))}\Big\}\;.
  \label{frakBqz-2}
\end{align}
Indeed, the term in braces expands with $dx(\sigma_i(q))=dx(q)$ into
\begin{align*}
  \{~\}_j&=
  \frac{\omega_{0,|I_j|+1}(I_j,q)}{dx(q)(y(z)-y(q))}
+  \frac{\omega_{0,|I_j|+1}(I_j,\sigma_i(q))}{dx(\sigma_i(q))(y(z)-y(\sigma_i(q)))}
  \\
  &-\sum_{I_j'\uplus I_j''=I_j}
  \frac{\omega_{0,|I_j'|+1}(I'_j,q)}{dx(q)(y(z)-y(q))}
 \frac{\omega_{0,|I_j''|+1}(I''_j,\sigma_i(q))}{dx(\sigma_i(q))(y(z)-y(\sigma_i(q)))}\;.
\end{align*}
A $p$-fold product is then of the form
\begin{align*}
\sum_{I_1\uplus ...\uplus I_p=I''} \prod_{j=1}^p \{~\}_j
&=\sum_{n+n_2+\tilde{n}=p}
\frac{(-1)^{n_2}(n+n_2+\tilde{n})!}{n!n_2!\tilde{n}!}
\sum_{I_1\uplus....\uplus I_{n+n_2}
\uplus  I'_1\uplus....\uplus I'_{\tilde{n}+n_2}=I''}\times
\\
&\times \prod_{j=1}^{n+n_2}
  \frac{\omega_{0,|I_j|+1}(I_j,q)}{dx(q)(y(z)-y(q))}
\prod_{k=1}^{\tilde{n}+n_2}
\frac{\omega_{0,|I_k'|+1}(I_k',\sigma_i(q))}{dx(\sigma_i(q))(y(z)-y(\sigma_i(q)))}
\;.
\end{align*}
We change the summation variables to $n+n_2=s_0$, $\tilde{n}+n_2=s-s_0$
and first sum over $n_2\in [0..\min(s_0,s-s_0)]$ and then over $s,s_0$. Because of
\[
\sum_{n_2=0}^{\min(s_0,s-s_0)}
 \frac{(-1)^{n_2}(s-n_2)!}{(s_0-n_2)!n_2!(s-s_0-n_2)!}
 =\sum_{n_2=0}^{\min(s_0,s-s_0)}(-1)^{n_2} \binom{s-n_2}{s_0} \binom{s_0}{n_2}
 =1
\]
(see e.g.\ \cite[Vol.~4, eq.\ (10.13)]{Gould}) we obtain the same
expression as (\ref{frakBqz-1}), which proves (\ref{frakBqz-2}).

With these preparations we complete the final step:
\begin{proposition}
\label{prop:tr-kernel}
  For all $|I|\geq 2$ one has  
\begin{align*}
  &
\Res\displaylimits_{q\to \beta_i}  
\frac{\omega_{0,|I|+1}(I,q)  + \mathsf{S}\omega_{0,|I|+1}(I,q)}{w-q}dw=0\;.
\end{align*}  
Equivalently, the meromorphic differentials $\omega_{0,m+1}$ satisfy
the topological recursion 
\begin{align}
&\mathcal{P}_z^i \omega_{0,|I|+1}(I,z)
\label{tr-formula}
\\
&=
\Res\displaylimits_{q\to \beta_i} \frac{\frac{1}{2} (\frac{dz}{z-q}-\frac{dz}{z-\sigma_i(q)})
  }{dx(\sigma_i(q))(y(q)-y(\sigma_i(q)))}
  \sum_{I_1\uplus I_2=I} \omega_{0,|I_1|+1}(I_1,q)
  \omega_{0,|I_2|+1}(I_2,\sigma_i(q))\;.
  \nonumber
\end{align}
\begin{proof} 
  By induction on $|I|\geq 2$ using (\ref{res-omSom}) together with
  (\ref{frakBqz-2}). For $|I|=2$ we necessarily have $|I_0|=|I'|=1$
  and $I''=\emptyset$. This implies $\tilde{\mathfrak{B}}(I'';q,z)=1$ and
  $\mathsf{S}\omega_{0,2}(I_1,q)=0$.  Since $\omega_{0,2}(I_1,q)$ is
  regular at $q=\beta_i$, the integrand of the rhs of
  (\ref{res-omSom}) has vanishing residue.  Assume the proposition is
  true for $|I|\leq \ell$. Because of $|I_0|\geq 1$, any
  $I',I'',I_1,...,I_p$ on the rhs of (\ref{res-omSom}) and in
  (\ref{frakBqz-2}) is of length strictly $<\ell$. Then by induction
  hypothesis, the linear loop equation Proposition~\ref{prop-linloop} and
  the regularity of $\frac{y(\sigma_i(q))-y(q)}{dx(q)}$ at
  $q\to \beta_i$, the whole integrand on the rhs of (\ref{res-omSom})
  is regular at $q=\beta_i$,  i.e.\ its residue equals $0$.

  This shows $\mathcal{P}^i_z\omega_{0,|I|+1}(I,z) =-
\Res\displaylimits_{q\to \beta_i}  
\frac{\mathsf{S}\omega_{0,|I|+1}(I,q)}{z-q}dz
$.
With Remark~\ref{rem:id} we have
$\Res\displaylimits_{q\to \beta_i}  
\frac{\mathsf{S}\omega_{0,|I|+1}(I,q)}{z-q}dz
=\Res\displaylimits_{q\to \beta_i}  
(\frac{1}{2}\frac{\mathsf{S}\omega_{0,|I|+1}(I,q)}{z-q}
+\frac{1}{2}\frac{\mathsf{S}\omega_{0,|I|+1}(I,\sigma_i(q))}{z-\sigma_i(q)}
)dz$. Now (\ref{tr-formula}) follows from 
$\mathsf{S}\omega_{0,|I|+1}(I,\sigma_i(q))
=-\mathsf{S}\omega_{0,|I|+1}(I,q)$.
\end{proof}
\end{proposition}

\subsection{Symmetry of the involution identity II:
  \texorpdfstring{$q\to \beta_i$}{q->\beta} and
  \texorpdfstring{$q\to \iota \beta_i$}{q->\iota\beta}}
\label{sec:iota-symm-II}

Recall that the investigation of $\omega_{0,|I|+1}(I,q)$ for $q$ near
a ramification point $\beta_i$ of $x$ in Sections~\ref{sec:linloopeq}
and \ref{sec:tr-kernel} started from the $\iota$-reflection
(\ref{eq:flip-om-refl}) of the involution identity (\ref{eq:flip-om}).
It thus remains to prove that the obtained solution is consistent with
the original equation (\ref{eq:flip-om}). This means we have to show
\begin{align}
& 
\Res\displaylimits_{q\to \beta_i}  \frac{\omega_{0,|I|+1}(I,q)dw}{w-q}
\label{eq:symm2}
\\
&=\Res\displaylimits_{q\to \beta_i}  \frac{dy(q)dw}{w-q}
\sum_{s=2}^{|I|} \sum_{I_1\uplus ...\uplus I_s=I}
\frac{1}{s} \Res\displaylimits_{z\to q}  \Big(
\frac{dx(z)}{(y(q)-y(z))^{s}}  \prod_{j=1}^s 
\frac{\omega_{0,|I_j|+1}(I_j, z)}{dx(z)} 
\Big)\;.
\nonumber
\end{align}
This is the same as 
\begin{align*}
0
&=-\Res\displaylimits_{q\to \beta_i}  \frac{dy(q)dw}{w-q}
\sum_{s=1}^{|I|} \sum_{I_1\uplus ...\uplus I_s=I}
\frac{1}{s} \Res\displaylimits_{z\to q}  \Big(
\frac{dx(z)}{(y(q)-y(z))^{s}}  \prod_{j=1}^s 
\frac{\omega_{0,|I_j|+1}(I_j, z)}{dx(z)} 
\Big)
\\
&=\Res\displaylimits_{q\to \beta_i}  \frac{dy(q)dw}{w-q}
\sum_{s=1}^{|I|} \sum_{I_1\uplus ...\uplus I_s=I}
\frac{1}{s} \Res\displaylimits_{z\to \beta_i}  \Big(
\frac{dx(z)}{(y(q)-y(z))^{s}}  \prod_{j=1}^s 
\frac{\omega_{0,|I_j|+1}(I_j, z)}{dx(z)} 
\Big)\;,
\end{align*}
where (\ref{commute-res}) has been used. We expand $\frac{1}{(y(q)-y(z))^{s}}$
about $y(z)=y(\beta_i)$ and then order into powers of $y(\beta_i)$.
Hence (\ref{eq:symm2}) holds iff
\begin{align*}
0
&=\Res\displaylimits_{q\to \beta_i}  \frac{dy(q)dw}{w-q}
\sum_{p=1}^\infty \frac{1}{p (y(q)-y(\beta_i))^{p}}
\\
&\times
\sum_{s=1}^{\min(|I|,p)} \sum_{I_1\uplus ...\uplus I_s=I}
\binom{p}{s} 
\Res\displaylimits_{z\to \beta_i}  \Big(
dx(z) (y(z)-y(\beta_i))^{p-s}  \prod_{j=1}^s 
\frac{\omega_{0,|I_j|+1}(I_j, z)}{dx(z)} 
\Big)
\\
&=\Res\displaylimits_{q\to \beta_i}  \frac{dy(q)dw}{w-q}
\sum_{p=1}^\infty \sum_{k=0}^p \frac{1}{p (y(q)-y(\beta_i))^{p}}
\binom{p}{k} (-1)^k (y(\beta_i))^k
\\
&\times
\sum_{s=1}^{\min(|I|,p-k)} \sum_{I_1\uplus ...\uplus I_s=I}
\binom{p-k}{s} 
\Res\displaylimits_{z\to \beta_i}  \Big(
dx(z) (y(z))^{p-k-s}  \prod_{j=1}^s \frac{\omega_{0,|I_j|+1}(I_j, z)}{dx(z)} 
\Big)\;.
\end{align*}
We prove that the last line vanishes identically for any $n=p-k$:
\begin{proposition}
\label{prop-new}
For any family $\omega_{0,|I|+1}(I,z)$ of $1$-forms in $z$
which satisfy the linear and quadratic loop
equations\footnote{In our situation, the linear loop equations are
proved in Proposition~\ref{prop-linloop} and the quadratic loop equations 
are equivalent to Proposition~\ref{prop:tr-kernel}.}
\cite{Borot:2013lpa, Borot:2015hna} one has, for any $n\geq 1$,
\begin{align}\label{new}
  \sum_{s=1}^{|I|} \sum_{I_1\uplus...\uplus I_s=I}
  \Res\displaylimits_{z\to \beta_i}\Bigg[\binom{n}{s}y(z)^{n-s}
  dx(z)\prod_{j=1}^s\frac{\omega_{0,|I_j|+1}(I_j,z)}{dx(z)}\Bigg]=0\;.
\end{align}
In particular, \eqref{eq:symm2} holds under these assumptions.
\begin{proof} In Remark~\ref{rem:loopinsertion} below
  we indicate that the assertion would be an immediate consequence of
  existence of a loop insertion operator. Because we did not prove
  that such an operator exists in our case we  give a direct
  combinatorial proof based on a technical Lemma~\ref{combi}.

We associate to the complex numbers $y,\bar{y},w,\bar{w}$ in
Lemma \ref{combi} the functions (and forms in $u_k$)
$y\mapsto y(z)$, $\bar{y}\mapsto
y(\sigma_{i}(z))$, $w\mapsto
\sum_{\emptyset \neq I'\subset I} t_{I'}
\frac{\omega_{0,|I'|+1}(I',z)}{dx(z)}$, $\bar{w}\mapsto
\sum_{\emptyset \neq I'\subset I} t_{I'}
\frac{\omega_{0,|I'|+1}(I',\sigma_{i}(z))}{dx(z)}$.
We keep only those terms which give rise to an admissible 
partition of $I$ (restrict to admissible products of 
$t_{I_k}$, then set $t_{I_k}\mapsto 1$). Lemma~ \ref{combi}
together with $e_2:=y\bar{w}+\bar{y}w+w\bar{w}
=y(w+\bar{w})+(\bar{y}-y)(w+\frac{w\bar{w}}{\bar{y}-y})$ 
gives
\begin{align}
  &  \sum_{k=0}^{n-1}\sum_{I_1\uplus ...\uplus I_{n-k}=I}\binom{n}{k}
  \bigg[y(z)^k \prod_{j=1}^{n-k}\frac{\omega_{0,|I_j|+1}(I_j,z)}{dx(z)}
  \nonumber
 \\*[-2ex]
&\hspace*{5cm}  
+y(\sigma_{i}(z))^k \prod_{j=1}^{n-k}
\frac{\omega_{0,|I_j|+1}(I_j,\sigma_i(z))}{dx(z)}\bigg]dx(z)
  \label{zsig}
  \\
  &=
\!\!\!\!\sum_{(n_1,n_2,n_3,n_4)\in \mathcal{D}_n}\!\!\!\!\!\!
(-1)^{n_4}n\frac{\prod_{k=1}^{n_3+n_4-1}(n_1+k)(n_2+k)}{n_3!n_4!(n_3+n_4-1)!}
y(z)^{n_1}y(\sigma_i(z))^{n_2}dx(z)
\nonumber
\\
&\times
\sum_{I_1\uplus I_2\uplus  ...\uplus I_{n_3+n_4}=I}
\prod_{j=1}^{n_3}\frac{(\omega_{0,|I_j|+1}(I_j,z)
  +\omega_{0,|I_j|+1}(I_j,\sigma_i(z)))}{dx(z)}
\nonumber
\\
&\times\prod_{j=n_3+1}^{n_3+n_4}\bigg[y(z)
\frac{\omega_{0,|I_j|+1}(I_j,\sigma_i(z))+\omega_{0,|I_j|+1}(I_j,z)}{dx(z)}
\nonumber
\\
&\hspace*{2cm}
+\frac{y(\sigma_i(z))-y(z)}{dx(z)}
\Big(\omega_{0,|I_j|+1}(I_j,z) +\mathsf{S}\omega_{0,|I_j|+1}(I_j,z) \Big)
\bigg],
\nonumber
\end{align}
where $\mathsf{S}\omega$ was defined in (\ref{calS2}) and
$\mathcal{D}_n$ is a set of tuples specified in (\ref{Dn}).  The
linear loop equation Proposition~\ref{prop-linloop} and
Remark~\ref{rem:id} imply that
$\frac{\omega_{0,|I_j|+1}(I_j,\sigma_i(z))+\omega_{0,|I_j|+1}(I_j,z)}{dx(z)}$
is holomorphic at $z=\beta_i$. Holomorphicity of the last line at $z=\beta_i$
follows from Proposition~\ref{prop:tr-kernel}. Thus, \eqref{zsig} is
regular at $z=\beta_i$.  The projection to admissible partitions of
$I$ guarantees that contributions to \eqref{zsig} with $n-k>|I|$ are
automatically zero.

We finish the proof of the proposition with the fact (\ref{Id}) that 
for any meromorphic $1$-form $\omega(q)$ the residue
does not change under the Galois involution, 
\begin{align*}
  \Res\displaylimits_{q\to \beta_i}(\omega (q)+\omega(\sigma_i(q)))
  =2  \Res\displaylimits_{q\to \beta_i} \omega (q)\;.
\end{align*}
Since the residue of \eqref{zsig} vanishes\footnote{Note that
  \eqref{Id} only states equality of the residue.  The integrand of
  \eqref{new} has, in general, higher order poles, but no residue.},
this implies the assertion~(\ref{new}).
\end{proof}
\end{proposition}

\begin{remark}\label{rem:loopinsertion}
  For topological recursion, the existence of a loop insertion
  operator $D_w$ is proved \cite{Eynard:2007kz}.  It is unclear
  whether the same holds for blobbed topological recursion in general
  as well. However, assuming that a loop insertion operator $D_w$
  exists\footnote{$D_w$ acts as a derivation and satisfies
    $D_w(dx)=0$, $D_w(y(z)dx(z))dx(w)=\omega_{0,2}(w,z)$ and
    $D_w(\omega_{0,|I|+1}(I,z))dx(w)=\omega_{0,|I|+2}(I,w,z)$.}
  for any blobbed topological recursion, we could
  prove \eqref{new} by induction in $|I|\mapsto |I|+1$ with the
  following consideration:
\begin{align*}
  0&=D_w\sum_{s=1}^{|I|} \sum_{I_1\uplus...\uplus I_s=I}
  \Res\displaylimits_{z\to \beta_i}
  \bigg[\binom{n}{s}y(z)^{n-s}dx(z)\prod_{j=1}^s
  \frac{\omega_{0,|I_j|+1}(I_j,z)}{dx(z)}\bigg]dx(w)
  \\
  &=\sum_{s=1}^{|I|} \sum_{I_1\uplus...\uplus I_s=I}  \Res\displaylimits_{z\to \beta_i}
  \bigg[\binom{n}{s} (n-s)y(z)^{n-s-1}dx(z)\frac{\omega_{0,2}(w,z)}{dx(z)}
  \prod_{j=1}^s\frac{\omega_{0,|I_j|+1}(I_j,z)}{dx(z)}\bigg]
  \\
  &+\sum_{s=1}^{|I|} \sum_{I_1\uplus...\uplus I_s=I}\Res\displaylimits_{z\to \beta_i}
  \bigg[\binom{n}{s} y(z)^{n-s}dx(z)
  \\
  &\qquad \qquad \qquad \qquad  \times
  \sum_{\ell=1}^s\frac{\omega_{0,|I_\ell|+2}(I_\ell,w,z)}{dx(z)}
  \prod_{j=1,j\neq \ell}^s
  \frac{\omega_{0,|I_j|+1}(I_j,z)}{dx(z)}\bigg]
  \\
  &=\sum_{s=1}^{|I|+1} \sum_{I_1\uplus...\uplus I_s=I\uplus w}
  \Res\displaylimits_{z\to \beta_i}\bigg[\binom{n}{s}y(z)^{n-s}dx(z)\
  \prod_{j=1}^s\frac{\omega_{0,|I_j|+1}(I_j,z)}{dx(z)}\bigg]\;.
\end{align*}
We have used that the sum $\sum_{I_1\uplus...\uplus I_s=I\uplus w}$
should be symmetric such that all terms with the form $\omega_{0,2}(w,z)$
coming from the second line get a symmetry factor $\frac{1}{s+1}$
so that  $\binom{n}{s}\frac{n-s}{s+1}=\binom{n}{s+1}$.
\end{remark}

\subsection{Finishing the proof of Theorem~\ref{thm:flip}}

We can now assemble the pieces into a proof of
Theorem~\ref{thm:flip}. In a first step we assume that the rhs of
(\ref{eq:flip-om}) and of (\ref{eq:flip-om-refl}) are the same. By
induction these rhs have poles in the points
$q\in \{\beta_i,\iota \beta_i,u_k,\iota u_k\}$. Then conditions (b),(c),(d)
imply that $\omega_{0,|I|+1}(I,q)$ is meromorphic on $\hat{\mathbb{C}}$ with
poles only in $q\in \{\beta_i,\iota u_k\}$. Therefore,
\begin{align}
  &\omega_{0,m+1}(u_1,..,u_m,z)
  \label{liouville}
  \\
  &-\sum_{i=1}^r
  \Res\displaylimits_{q\to \beta_i}
  \frac{\omega_{0,m+1}(u_1,..,u_m,q)dz}{z-q}
  -\sum_{k=1}^m
  \Res\displaylimits_{q\to \iota u_k}
  \frac{\omega_{0,m+1}(u_1,..,u_m,q)dz}{z-q}
    \nonumber
\end{align}
is a holomorphic 1-form on the Riemann sphere $\hat{\mathbb{C}}\ni
z$, hence identically zero. Inserting the
residues from Proposition~\ref{prop:om-hol2} and
Proposition~\ref{prop:tr-kernel} represents $\omega_{0,|I|+1}(I,z)$ as
(\ref{sol:omega})+(\ref{eq:kernel}).

It remains to prove that the difference between the rhs
of (\ref{eq:flip-om}) and (\ref{eq:flip-om-refl}) is a holomorphic form on
$\hat{\mathbb{C}}\ni q$, hence zero.
By induction it can have poles at most in
$q\in \{\beta_i,\iota \beta_i,u_k,\iota u_k\}$. In
Section~\ref{sec:iota-symm-I} we have shown that the difference is
holomorphic at every $q=u_k$ and $q=\iota u_k$, and in
Section~\ref{sec:iota-symm-II} it
is shown that the difference is
holomorphic at every $q=\beta_i$ and $q=\iota \beta_i$.
This completes the proof of
Theorem~\ref{thm:flip}.
\hspace*{\fill} $\square$%

\subsection{The sum over all preimages}

Let $\omega^{\text{TR}}_{g,n+1}$ be the differential forms generated
by topological recursion only \cite{Eynard:2007kz}.  It is well-known
that for any $\omega^{\text{TR}}_{g,n+1}$, except $(g,n)=(0,0)$, the
following identity holds:
\begin{theorem}[\cite{Eynard:2007kz}]\label{thm:eynard}
  Let $I=\{u_1,...,u_n\}$ and $\omega^{\text{TR}}_{g,n+1}$ be the
  differential forms generated by topological recursion, where
  $\omega^{\text{TR}}_{0,2}$ is the Bergman kernel and
  $\omega^{\text{TR}}_{0,1}=y\,dx$. Let further be $\hat{z}^k$,
  $k=1,...,d$ the preimages with $x(z)=x(\hat{z}^k)$ such that
  $z\neq \hat{z}^k$ and $\hat{z}^0\equiv z$. Then, the sum of
  $\omega^{\text{TR}}_{g,n+1}$ over all preimages vanishes, except for
  the Bergman kernel,
\begin{align*}
  \sum_{k=0}^d\frac{\omega^{\text{TR}}_{g,n+1}(I,\hat{z}^k)}{dx(z^k)}
  =\frac{\delta_{g,0}\delta_{n,1}dx(u_1)}{(x(z)-x(u_1))^2}.
\end{align*}
\end{theorem}
For $\omega^{\text{TR}}_{0,2}$ the Theorem can be proved directly, and
for any other $\omega^{\text{TR}}_{g,n+1}$ it follows from the
structure of the recursive kernel
\begin{align*}
	K_i(z,q)=   \frac{\frac{1}{2} (\frac{dz}{z-q}-\frac{dz}{z-\sigma_i(q)})
	}{dx(\sigma_i(q))(y(q)-y(\sigma_i(q)))}\;,
\end{align*}
since 
\begin{align*}
  \sum_{k=0}^d\frac{\frac{1}{\hat{z}^k-q}
    -\frac{1}{\hat{z}^k-\sigma_i(q)}}{x'(\hat{z}^k)}
  =\frac{1}{x(z)-x(q)}-\frac{1}{x(z)-x(\sigma_{i}(q))}=0.
\end{align*}
Consequently, it is natural to ask whether a similar identity holds
for the preimage sum of $\omega_{0,n+1}$ defined by \eqref{eq:flip-om}
together with \eqref{om02}. Applying Theorem \ref{thm:flip}, we get:
\begin{proposition}
Let $I=\{u_1,...,u_n\}$. For $n>0$ the sum over all preimages is 
\begin{align*}
  \sum_{k=0}^d\frac{\omega_{0,n+1}(I,\hat{z}^k)}{dx(z^k)}
  &=\frac{\delta_{n,1}dx(u_1)}{(x(z)-x(u_1))^2}
  +\frac{\delta_{n,1}dy(u_1)}{(x(z)+y(u_1))^2}
  \\
  &+\sum_{j=1}^n d_{u_j}\bigg[\sum_{s=1}^{|I|-1}(-1)^{s+1}
  \sum_{I_1\uplus ...\uplus I_s=I{\setminus} u_j}\frac{
    \prod_{i=1}^s\frac{\omega_{0,|I_1|+1}(I_i,u_j)}{
      (x(z)+y(u_j))dx(u_j)}}{x(z)+y(u_j)}\bigg]\;,
\end{align*}
where $\hat{z}^0\equiv z$.
\begin{proof}
First, look at $\omega_{0,2}$ from the second line of \eqref{om02}. Dividing it by $dx(\hat{z}^k)$ and summing over $k$ yields
\begin{align}\label{om02pre}
\sum_{k=0}^d\frac{\omega_{0,2}(u,\hat{z}^k)}{dx(z^k)}
&=-d_u\sum_{k=0}^d\Big(
\frac{1}{2} \frac{1}{x'(\hat{z}^k)(u-\hat{z}^k)}+
\frac{1}{2} \frac{\iota'(\hat{z}^k)}{x'(\hat{z}^k)(\iota u-\iota \hat{z}^k)}
\\
&\qquad\qquad  -\frac{1}{2} \frac{\iota'(\hat{z}^k)}{
  x'(\hat{z}^k)(u-\iota \hat{z}^k)}
-  \frac{1}{2} \frac{1}{x'(\hat{z}^k)(\iota u- \hat{z}^k)}\Big).
\nonumber
\end{align}
Now, use the fact that $\chi_k=\iota \hat{z}^k$ are preimages
of $\chi=\iota z$ under the map $y$, i.e.\ $y(\chi)=y(\chi_k)$.
Furthermore, if a point $z$ does not coincide with one of its
preimages $\hat{z}^k$, it will generically also not 
coincide under the global involution, $\chi\neq \chi_k$. Together with
$\frac{\iota'(z)}{x'(z)}=-\frac{1}{y'(\iota z)}$,
        \eqref{om02pre} breaks down to
\begin{align*}
  \sum_{k=0}^d\frac{\omega_{0,2}(u,\hat{z}^k)}{dx(z^k)}
  &=\frac{1}{2}d_u\Big(\frac{1}{x(z){-}x(u)}
  -\frac{1}{y(\iota z){-}y(\iota z)}+\frac{1}{y(\iota z){-}y(u)}
  -\frac{1}{x(z){-}x(\iota u)}\Big)
  \\
&=d_u\Big(\frac{1}{x(z)-x(u)}-\frac{1}{x(z)+y(u)}\Big)\;.
\end{align*}
For $\omega_{0,n}$ with $n>2$, Theorem \ref{thm:flip} proves by the
same consideration as for topological recursion in Theorem
\ref{thm:eynard} that the poles at the ramification points $\beta_i$
do not contribute. For the remaining part, we use the equivalence
given by Proposition \ref{prop:om-hol2} to Lemma
\ref{lem:om-hol1}. Interchanging the integral and the sum over all
preimages in Lemma \ref{lem:om-hol1} gives
\begin{align*}
 &\sum_{k=0}^d\frac{\omega_{0,|I|+1}(I,\hat{z}^k)}{dx(z^k)}
 \\
 &=-\sum_{u_j\in I}d_{u_j} \bigg[\sum_{s=1}^{|I|-1}
 \sum_{I_1\uplus ...\uplus I_s=I{\setminus} u_j}\sum_{k=0}^d
\frac{1}{s!}\frac{\partial^s
  \big(\frac{\iota'(\hat{z}^k)}{x'(\hat{z}^k)(\iota \hat{z}^k-u_j)}\big)}{
  \partial (y(u_j))^s}
\prod_{i=1}^s \frac{\omega_{0,|I_i|+1}(I_i,u_j)}{dx(u_j)} \bigg]
\\
&=-\sum_{u_j\in I}d_{u_j} \bigg[\sum_{s=1}^{|I|-1}
\sum_{I_1\uplus ...\uplus I_s=I{\setminus} u_j}
\frac{1}{s!}\frac{\partial^s
\big(\frac{1}{x(z)+y(u_j)}\big)}{\partial (y(u_j))^s}
\prod_{i=1}^s \frac{\omega_{0,|I_i|+1}(I_i,u_j)}{dx(u_j)} \bigg]\;.
\end{align*}
Carrying out the derivative with respect to $y(u_j)$ yields the assertion.
\end{proof}
\end{proposition}

\subsection{A particular symmetry under the involution
  \texorpdfstring{$\iota$}{\iota}}
\label{sec:prod-nabla-sym}

In the second part we prove that the planar sector (genus $0$) of the
quartic Kontsevich model is completely governed by the involution
identity (\ref{eq:flip-om}). In a decisive step of the proof we will
need an intriguing symmetry resulting from (\ref{eq:flip-om}) alone:
\begin{align}
  0&=\Res\displaylimits_{z\to q} \Big[
  \sum_{s=1}^{|I|} \frac{1}{s}\sum_{I_1\uplus ... \uplus I_s=I}
  \Big( \frac{dx(z)
  \prod_{j=1}^s
\frac{\omega_{0,|I_j|+1}(I_j,z)}{dx(z)}
}{(x(z)-x(q))(y(q)-y(z))^s}
\Big)\Big]
\nonumber
\\
&
+\Res\displaylimits_{z\to q} \Big[
  \sum_{s=1}^{|I|} \frac{1}{s}\sum_{I_1\uplus ... \uplus I_s=I}
  \Big( \frac{dx(\iota z)
  \prod_{j=1}^s
\frac{\omega_{0,|I_j|+1}(I_j,\iota z)}{dx(\iota z)}
}{(x(\iota z)-x(\iota q))(y(\iota q)-y(\iota z))^s}
\Big)\Big]\;.
\label{prod-om-nabla-sym}
\end{align}
The residues in (\ref{prod-om-nabla-sym}) can be expressed as limits
of partial derivatives of
$\big( \frac{x(z)-x(q)}{y(z)-y(q)}\big)^s \prod_{j=1}^s
\frac{\omega_{0,|I_j|+1}(I_j,z)}{dx(z)}$ and
$\big( \frac{x(\iota z)-x(\iota q)}{y(\iota z)-y(\iota q)}\big)^s
\prod_{j=1}^s \frac{\omega_{0,|I_j|+1}(I_j,\iota z)}{dx(z)}$ with
respect to $x(z)$ and $x(\iota z)$. Using (\ref{def:nabla-om}) we thus
bring (\ref{prod-om-nabla-sym}) into
an equation that we need:
\begin{align}
0&=\sum_{s=1}^{|I|} \frac{1}{s}\sum_{I_1\uplus ... \uplus I_s=I}
\sum_{n_1+...+n_s=s} \Big(
\prod_{j=1}^s
\nabla^{n_j}\omega_{0,|I_j|+1}(I_j,q)
+\prod_{j=1}^s
  \nabla^{n_j}\omega_{0,|I_j|+1}(I_j,\iota q)\Big)\;.
\label{log-nabla-sym}
\end{align}

The main combinatorial tool to verify (\ref{prod-om-nabla-sym}) is
Corollary~\ref{coro-axby}.  Using Corollay~\ref{coro-axby} we
prove that the integrand in (\ref{prod-om-nabla-sym}) is an exact
$1$-form in $z$:
\begin{proposition}
\label{prop:nabla-sym}
\begin{align}
& \sum_{s=1}^{|I|} \frac{1}{s}\sum_{I_1\uplus ... \uplus I_s=I}
\bigg\{ \frac{dx(z)\prod_{j=1}^s
\frac{\omega_{0,|I_j|+1}(I_j,z)}{dx(z)}
}{(x(z)-x(q))(y(q)-y(z))^s}
-\frac{dy(z)\prod_{j=1}^s
\frac{\omega_{0,|I_j|+1}(I_j,\iota z)}{(-dy(z))}
}{(y(q)-y(z))(x(z)-x(q))^s}
\bigg\}
\nonumber
\\
&=\sum_{s=2}^{|I|}\frac{1}{s!} \sum_{I_1\uplus...\uplus I_s=I} d_z\bigg[
\sum_{r=0}^{s-2}
\Big(\frac{1}{dy(z)}d_z\Big)^{s-2-r} \Big[\frac{1}{y(q)-y(z)}\Big]
\nonumber
\\
&\qquad \times \Big({-}\frac{1}{dy(z)}d_z\Big)^{r}
\Big[\frac{1}{(x(z)-x(q))}
\frac{dx(z)}{dy(z)} \prod_{j=1}^s \frac{\omega_{0,{|I_j|+1}}(I_j,z)}{dx(z)}\Big]
\bigg]\;.
\label{prod-om-nabla-d}
\end{align}
In particular, the residue \eqref{prod-om-nabla-sym} at $z=q$ is zero.
\begin{proof}
  In the first line of (\ref{prod-om-nabla-d}), restricted to
  $s\geq 2$, we write $\frac{1}{(y(q)-y(z))^s}$ as a multiple
  differential and integrate by parts:
\begin{align*}
&\sum_{s=2}^{|I|} \frac{1}{s}\sum_{I_1\uplus ... \uplus I_s=I}
\frac{dx(z)}{(x(z)-x(q))(y(q)-y(z))^s}
\prod_{j=1}^s
\frac{\omega_{0,|I_j|+1}(I_j,z)}{dx(z)}
\\
&=\sum_{s=2}^{|I|}\frac{1}{s!} \sum_{I_1\uplus...\uplus I_s=I} d_z\Big\{
\sum_{r=0}^{s-2}
\Big(\frac{1}{dy(z)}d_z\Big)^{s-2-r} \Big[\frac{1}{(y(q)-y(z))}\Big]
\\*
&\qquad \times \Big({-}\frac{1}{dy(z)}d_z\Big)^{r}
\Big[\frac{1}{(x(z)-x(q))}
\frac{dx(z)}{dy(z)} \prod_{j=1}^s \frac{\omega_{0,{|I_j|+1}}(I_j,z)}{dx(z)}\Big]
\Big\}
\\[-1ex]
&+\sum_{s=2}^{|I|} \frac{1}{s!}
\sum_{I_1\uplus...\uplus I_s=I}
\frac{dy(z)}{(y(q)-y(z))}
\\
&\qquad \times
\Big({-}\frac{1}{dy(z)}d_z\Big)^{s-1}
\Big[\frac{1}{(x(z)-x(q))}
\frac{dx(z)}{dy(z)} \prod_{j=1}^s \frac{\omega_{0,{|I_j|+1}}(I_j,z)}{dx(z)}\Big]\;.
\end{align*}
Hence the assertion is true if the $1$-form in $z$
\begin{align*}
f(I;z,q)&:=  \frac{\omega_{0,|I|+1}(I,z)+\omega_{0,|I|+1}(I,\iota z)}{(x(z)-x(q))}
\\
&+\sum_{s=2}^{|I|} 
\sum_{I_1\uplus...\uplus I_s=I}
\Big\{
-\frac{dy(z)}{s} 
\frac{\prod_{j=1}^s \omega_{0,{|I_j|+1}}(I_j,\iota z)}{
  (x(z)-x(q))^s (-dy(z))^{s}}
\\
&+ 
\frac{dy(z)}{s!}
\Big({-}\frac{1}{dy(z)}d_z\Big)^{s-1}
\Big[\frac{1}{(x(z)-x(q))}
\Big(\frac{dx(z)}{dy(z)}
\prod_{j=1}^s \frac{\omega_{0,{|I_j|+1}}(I_j,z)}{dx(z)}\Big)\Big]\Big\}
\end{align*}
is identically zero.  The last line of $f(I;z,q)$ is of the form of
Corollary~\ref{coro-axby} with
$\frac{1}{dy(z)}d_z=b(y) \partial_x + \partial_y$,
$a(x)=\frac{1}{x -x_q}$, $b(y)=\frac{dx(z)}{dy(z)}$,
$c_j(y)= \frac{\omega_{0,{|I_j|+1}}(I_j,z)}{dx(z)}$ and a constant
$x_q=x(q)$. We thus get
\begin{align}
f(I;z,q)&:=  \frac{\omega_{0,|I|+1}(I,z)+\omega_{0,|I|+1}(I,\iota z)}{(x(z)-x(q))}
\label{fIzq}
\\
&+\sum_{s=2}^{|I|} 
\sum_{I_1\uplus...\uplus I_s=I}
\Big\{
-\frac{dy(z)}{s} 
\frac{\prod_{j=1}^s \omega_{0,{|I_j|+1}}(I_j,\iota z)}{(x(z)-x(q))^s (-dy(z))^{s}}
\nonumber
\\
&+
\frac{dy(z)}{s!}
\sum_{r=1}^s \frac{(r-1)!}{(x(z)-x(q))^r}
\nonumber
\\
&\times
\sum_{\substack{J_1\uplus ... \uplus J_r=\{1,2,...,s\}\\ J_1<...<J_r}}
\prod_{k=1}^r
\Big(-\frac{1}{dy(z)}d_z\Big)^{|J_k|-1} \Big[
\frac{dx(z)}{dy(z)} \prod_{j\in J_k}
\frac{\omega_{0,{|I_j|+1}}(I_j,z)}{dx(z)}\Big]\Big\}\;.
\nonumber
\end{align}
The derivatives in the last line are expressed as a residue:
\begin{align*}
&  \Big(-\frac{1}{dy(z)}d_z\Big)^{|J_k|-1} \Big[
\frac{dx(z)}{dy(z)} \prod_{j\in J_k}
\frac{\omega_{0,{|I_j|+1}}(I_j,z)}{dx(z)}\Big]
\\
&=-(|J_k|-1)!
\Res\displaylimits_{w\to z}
\Big(\frac{dx(w)}{(y(z)-y(w))^{|J_k|}} \prod_{j\in J_k}
\frac{\omega_{0,{|I_j|+1}}(I_j,w)}{dx(w)}\Big)\;.
\end{align*}
The case $r=1$ combines to the involution identity (\ref{eq:flip-om})
and is thus identified as the negative of the first line of
(\ref{fIzq}).  In the remainder we order the partitions of $I$:
\begin{align}
f(I;z,q)&=
\sum_{s=2}^{|I|} 
\sum_{\substack{I_1\uplus...\uplus I_s=I \\ I_1<....<I_s}}
\bigg\{
-\frac{(s-1)!dy(z)\prod_{j=1}^s \omega_{0,{|I_j|+1}}(I_j,\iota z)}{
  (x(z)-x(q))^s (-dy(z))^{s}}
\label{fIzq-1}
\\
&+
dy(z)
\sum_{r=1}^s \frac{(-1)^r (r-1)!}{(x(z)-x(q))^r(dy(z))^r}
\nonumber
\\[-1ex]
&\times \!\!
\sum_{\substack{J_1\uplus ... \uplus J_r=\{1,2,...,s\}\\ J_1<...<J_r}}
\prod_{k=1}^r
\Res\displaylimits_{w\to z}
\Big(\frac{(|J_k|-1)! dy(z)dx(w)}{(y(z)-y(w))^{|J_k|}} \prod_{j\in J_k}
\frac{\omega_{0,{|I_j|+1}}(I_j,w)}{dx(w)}\Big)\bigg\}\,.
\nonumber
\end{align}
We change the order of the summations. The outer summation is a sum
over ordered partitions $I'_1\uplus ... \uplus I'_r$ given by
$I'_k=\cup_{j\in J_k} I_j$, which is combined with an inner summation
over ordered partitions of the individual $I'_k$.  Renaming in the
first line of (\ref{fIzq-1}) $s\mapsto r$ and $I_j\mapsto I'_j$, we
arrive at
\begin{align*}
f(I;z,q)&=
\sum_{r=2}^{|I|} 
\sum_{\substack{I'_1\uplus...\uplus I'_r=I \\ I'_1<....<I'_r}}
\frac{(r-1)!dy(z)}{(x(z)-x(q))^r (-dy(z))^{r}}
\Big\{-\prod_{j=1}^r \omega_{0,{|I'_j|+1}}(I'_j,\iota z)
\\
&+ 
\prod_{k=1}^r
\Res\displaylimits_{w\to z}
\Big(
\sum_{s=1}^{|I_k'|} \sum_{\substack{I_{k1}\uplus ... \uplus I_{ks}=I_k' \\ I_{k1}<...<I_{ks}}}
\frac{(s-1)! dy(z)dx(w)}{(y(z)-y(w))^{s}} \prod_{j=1}^s
\frac{\omega_{0,{|I_{kj}|+1}}(I_{kj},w)}{dx(w)}\Big)\Big\}\;.
\end{align*}
The outcome is zero thanks to (\ref{eq:flip-om}).
\end{proof}
\end{proposition}

\section{The quartic Kontsevich model}

\subsection{Summary of previous results}

Let $H_N$ be the real vector space of self-adjoint
$N\times N$-matrices, $H_N'$ be its dual and $(e_{kl})$ be the
standard matrix basis in the complexification of $H_N$.
We define a measure $d\mu_{E,\lambda}$ on $H_N'$ by
\begin{align}
  d\mu_{E,\lambda}(\Phi)&=\frac{1}{\mathcal{Z}}
  \exp\Big({-}\frac{\lambda N}{4}\mathrm{Tr}(\Phi^4)\Big)
d\mu_{E,0}(\Phi)\;,\quad
\label{measure}
\\
\mathcal{Z}&:=\int_{H_N'} \exp\Big({-}\frac{\lambda N}{4}
\mathrm{Tr}(\Phi^4)\Big)
d\mu_{E,0}(\Phi)\;,
\nonumber
\end{align}
where
$d\mu_{E,0}(\Phi)$ is a Gau\ss{}ian measure with covariance
\begin{align}
\Big[\int_{H_N'} d\mu_{E,0}(\Phi)\;
\Phi(e_{jk})\Phi(e_{lm})\Big]_c=\frac{\delta_{jm}\delta_{kl}}{N(E_j+E_l)}
\label{measure-0}
\end{align}
for some $0<E_1<\dots <E_N$.
The trace is understood as $\mathrm{Tr}(\Phi^4)
=\sum_{k,l,m,n=1}^N \Phi(e_{kl})\Phi(e_{lm})\Phi(e_{mn})\Phi(e_{nk})$.
Moments or cumulants of
$d\mu_{E,\lambda}$ are viewed as general or connected correlation
functions in a finite-dimensional approximation of a Euclidean quantum
field theory.

We call the objects resulting from (\ref{measure})+(\ref{measure-0})
the \emph{Quartic
  Kontsevich Model} because of its formal analogy with the Kontsevich
model \cite{Kontsevich:1992ti} in which $\mathrm{Tr}(\Phi^4)$ in
(\ref{measure}) is replaced by $\mathrm{Tr}(\Phi^3)$. The Gau\ss{}ian
measure $d\mu_{E,0}(\Phi)$ 
is the same (\ref{measure-0}). Kontsevich proved in
\cite{Kontsevich:1992ti} that (\ref{measure}) with
$\mathrm{Tr}(\Phi^3)$-term, viewed as function of the $E_k$, is the
generating function for intersection numbers of tautological
characteristic classes on the moduli space
$\overline{\mathcal{M}}_{g,n}$ of stable complex curves.

Derivatives of the Fourier transform
$\mathcal{Z}(M):=\int_{H_N'}
d\mu_{E,\lambda}(\Phi)\;e^{\mathrm{i}\Phi(M)}$ with respect to matrix
entries $M_{kl}$ and parameters $E_k$ of the free theory give rise to
\emph{Dyson-Schwinger equations} between the cumulants
\begin{align}
\langle e_{k_1l_1}... e_{k_nl_n}\rangle_c 
&= \frac{1}{\mathrm{i}^{n}}
\frac{\partial^n\log \mathcal{Z}(M)}{\partial 
M_{k_1l_1}... \partial M_{k_nl_n}} \Big|_{M=0}\;.
\label{cumulants}
\end{align}
After $1/N$-expansion one obtains a closed non-linear equation
\cite{Grosse:2009pa} for the $1/N$-leading part $G^{(0)}_{|kl|}$ of
the $2$-point function
$N\langle e_{kl}e_{lk}\rangle_c =\sum_{g=0}^\infty N^{-2g}
G^{(g)}_{|kl|}$ and a hierarchy of affine equations
\cite{Grosse:2012uv, Hock:2020rje} for all other functions.
The non-linear equation
for $G^{(0)}_{|kl|}$ was solved in a special case in
\cite{Panzer:2018tvy} and then in \cite{Grosse:2019jnv} in full
generality. The solution introduces a ramified covering
$R:\hat{\mathbb{C}}\to \hat{\mathbb{C}}$ of the Riemann sphere
$\hat{\mathbb{C}}=\mathbb{C}\cup \{\infty\}$ given by
(see \cite{Schurmann:2019mzu-v3})
\begin{align}
  R(z)=z-\frac{\lambda}{N} \sum_{k=1}^d \frac{\varrho_k}{\varepsilon_k+z}\;.
  \label{eq:R}
\end{align}
Here $(\varepsilon_k,\varrho_k)$ are
implicitly defined as solution of the system $R(\varepsilon_l)=e_l$,
$\varrho_l R'(\varepsilon_l)=r_l$ when assuming that $(E_1,...,E_N)$
consists of $d$ pairwise different values $e_1,...,e_d$ which arise
with multiplicities $r_1,...,r_d$. The planar $2$-point function is
then given by
$G^{(0)}_{|kl|}=\mathcal{G}^{(0)}(\varepsilon_k,\varepsilon_l)$ where
$\mathcal{G}^{(0)}$ is the rational function
\begin{align}
\mathcal{G}^{(0)}(z,w)&=\frac{\displaystyle
  1 -\frac{\lambda}{N} \sum_{k=1}^d \frac{r_k
  \prod_{j=1}^d \frac{
R(w){-}R({-}\widehat{\varepsilon_k}^j)}{ R(w)-R(\varepsilon_j)}
}{
(R(z)-R(\varepsilon_k))(R(\varepsilon_k)-R({-}w))}
}{R(w)-R(-z)}
\label{Gzw-final}
\end{align}
with poles located at $z+w=0$ and
$z,w\in \{\widehat{\varepsilon_k}^j\}$ for $k,j\in\{1,...,d\}$.  Here
$v\in\{z,\hat{z}^1,\dots,\hat{z}^d\}$ is the set of solutions of
$R(v)=R(z)$. One has $\mathcal{G}^{(0)}(z,w)=\mathcal{G}^{(0)}(w,z)$.

In \cite{Branahl:2020yru} we identified an algorithm which constructs
recursively, starting from (\ref{Gzw-final}), any cumulant
(\ref{cumulants}) of the measure
(\ref{measure})+(\ref{measure-0}). Its core is a coupled system of
loop equations \cite[Prop.~5.3, Prop.~5.6, Cor.~5.9]{Branahl:2020yru}
for three families of functions $\Omega^{(g)}_{m}(u_1,....,u_m)$,
$\mathcal{T}^{(g)}(u_1,...,u_m\|z,w|)$ and
$\mathcal{T}^{(g)}(u_1,...,u_m\|z|w|)$ with
$\mathcal{T}^{(0)}(\emptyset \|z,w|)=\mathcal{G}^{(0)}(z,w)$ and
$\mathcal{T}^{(0)}(\emptyset \|z|w|)$ determined in
\cite{Schurmann:2019mzu-v3}.  Of particular importance are the functions
$\Omega^{(g)}_{m}(u_1,....,u_m)$ which arise from complexification of
derivatives $\Omega^{(g)}_{a_1,\dots,a_n}$ of the partially summed
two-point function:
\begin{align}
  \sum_{g=0}^\infty N^{1-2g-n} \Omega^{(g)}_{a_1,\dots,a_n}
  :=\frac{\delta_{n,2}}{N(E_{a_1}-E_{a_2})^2}
  +\frac{\partial^{n-1}}{\partial
  E_{a_2}\cdots\partial E_{a_{n}}} \sum_{k=1}^N \langle
e_{a_1k}e_{ka_1}\rangle_{c}\;.
\end{align}
In \cite{Branahl:2020uxs} it is shown that the $\Omega^{(g)}_{a_1,\dots,a_n}$ 
are distinguished polynomials of the cumulants (\ref{cumulants}).

The system of equations established in \cite{Branahl:2020yru} permits
to determine $\Omega^{(g)}_{m}(u_1,...,u_m)$ without prior knowledge
of $\langle e_{a_1k}e_{ka_1}\rangle_{c}$. The solution of this system
for $\Omega^{(0)}_{2}$, $\Omega^{(0)}_{3}$, $\Omega^{(0)}_{4}$ and
$\Omega^{(1)}_{1}$ in \cite{Branahl:2020yru} gave strong support for
the conjecture that the meromorphic forms
$\omega_{g,n}(z_1,...,z_n):=\Omega^{(g)}_{n}(z_1,...,z_n)
dR(z_1)\cdots dR(z_n)$ obey blobbed topological recursion
\cite{Borot:2015hna} for the spectral curve
$(x:\hat{\mathbb{C}}\to \hat{\mathbb{C}},
\omega_{0,1}=ydx,\omega_{0,2})$ with
\begin{align}
  x(z)=R(z)\;,\qquad
  y(z)=-R(-z)\;,\qquad
  \omega_{0,2}(u,z)=\frac{du\,dz}{(u-z)^2}+\frac{du\,dz}{(u+z)^2}\;.
\end{align}

In the remainder of this paper we prove this conjecture for $\omega_{0,n}$.
More precisely, we prove that the solution of the system of equations given in
\cite{Branahl:2020yru} is identical to the solution of the involution identity
(\ref{eq:flip-om}) given in Theorem~\ref{thm:flip} for $\iota z=-z$
and $x=R$ as in (\ref{eq:R}). In particular, the part of 
$\omega_{0,n}$ with poles at ramification points of $x=R$ obeys exactly the
universal formula of topological recursion \cite{Eynard:2007kz}, and the other part with
poles along opposite diagonals $z_i+z_j=0$ is described by a residue formula of
very similar type.

\subsection{Loop equations}

The loop equations derived in \cite{Branahl:2020yru} imply that
$\omega_{0,m+1}(u_1,...,u_m,z)$ is an exact 1-form in every variable
$u_1,...,u_k$. We set
\begin{align}
  \omega_{0,m+1}(u_1,...,u_m,z)=d_{u_1}\cdots d_{u_m}
  \varpi_{0,m+1}(u_1,...,u_m;z)\;.
\end{align}
The $\varpi_{0,m+1}(u_1,...,u_m;z)$ are $1$-forms in $z$, they relate
via
$\varpi_{0,m+1}(u_1,...,u_m;z)= \lambda^{1-m}
\mathcal{W}^{(0)}_{m+1}(u_1,...,u_m,z)dR(z)$ to functions introduced
in \cite{Branahl:2020yru}. The loop equations derived in
\cite[Appendix E]{Branahl:2020yru} translate as follows into equations
between $\varpi_{0,m+1}$ and two classes of
auxiliary functions:
\begin{proposition} \label{prop:om-t}
  The loop equations of the quartic Kontsevich model have in lowest
  degree the solution 
$\varpi_{0,2}(u;z)=-\frac{dz}{(u-z)}-\frac{dz}{(u+z)}$ 
and can be turned for $I=\{u_1,...,u_m\}$ with $m\geq 2$ into
\cite[Prop.~E.1, eqs.~(E.4)+(E.5)]{Branahl:2020yru}
\begin{align}
\varpi_{0,|I|+1}(I;z)
&=\Res\displaylimits_{\substack{q\to -u_{1,...,m} \\ q\to \beta_{1,...2d}}}
\frac{dz}{z-q}
\bigg[
\sum_{I_1\uplus I_2=I}\varpi_{0,|I_1|+1}(I_1;q) \mathfrak{v}_{0,|I_2|}(I_2\|q)\bigg]
\nonumber
\\*
&+
\sum_{k=1}^m\frac{\mathfrak{v}_{0,|I|-1}(I{\setminus}u_k \|u_k)dz}{z+u_k}\;,
\label{om0-fraku}
\\
\text{where}\qquad
  \mathfrak{v}_{0,|I|}(I\|q)
&=\sum_{\substack{I_1\uplus I_2=I\\
 \text{possibly }   I_2= \emptyset}}
\sum_{j=1}^d \frac{\varpi_{0,|I_1|+1}(I_1;-\hat{q}^j)
\tilde{\mathfrak{t}}_{0,|I_2|}(I_2\|{-}\hat{q}^j,q|)}{dR(\hat{q}^j) (R(-q)-R(-\hat{q}^j))}
\nonumber
\\*
&-\sum_{k=1}^m 
\frac{\tilde{\mathfrak{t}}_{0,|I|-1}(I\backslash u_k\|u_k,q|)}{
  (R(u_k)-R(-q))(R(-u_k)-R(q))}
\label{frakuIq}
\end{align}
and $\tilde{\mathfrak{t}}_{0,0}(\emptyset\|z,q|)=1$ and for $|I|\geq 1$
\begin{align}
\tilde{\mathfrak{t}}_{0,|I|}(I\|z,q|)
&=
-\sum_{\substack{I_1\uplus I_2=I\\
\text{possibly }    I_2= \emptyset}}
\sum_{l=1}^d \frac{\varpi_{0,|I_1|+1}(I_1;{-}\hat{q}^l)
  \tilde{\mathfrak{t}}_{0,|I_2|}(I_2\|{-}\hat{q}^l,q|)}{
  dR(\hat{q}^l) (R(z)-R(-\hat{q}^l))}
\label{uIzq}
\\
&+\sum_{k=1}^m 
\frac{\tilde{\mathfrak{t}}_{0,|I|-1}(I\backslash u_k\|u_k,q|)}{(R(z)-R(u_k))(
R(q)-R(-u_k))}
\nonumber
\\
&-
\sum_{\substack{I_1\uplus I_2=I\\ \text{possibly }    I_2= \emptyset}}
\frac{\varpi_{0,|I_1|+1}(I_1;z)
\tilde{\mathfrak{t}}_{0,|I_2|}(I_2\|z,q|)}{dR(z)(R(q)-R(-z))}\;.
\nonumber
\end{align}
In \eqref{om0-fraku}, $\beta_1,...,\beta_{2d}$ are the ramification points of
the ramified cover $R$ given in \eqref{eq:R}. By
$\hat{q}^1,...,\hat{q}^d$ we denote the other preimages of $q$ under $R$,
i.e.\ $R(\hat{q}^j)=R(q)$. Generically they are
pairwise different and different from $q$. 
\end{proposition}
\noindent
Note that conditions (a),(c),(d) of Theorem~\ref{thm:flip} are 
automatically satisfied by (\ref{om0-fraku}).
Compared with \cite{Branahl:2020yru} we have set
  $\tilde{\mathfrak{t}}_{0,|I|}(I\|z,w|)=
\frac{\mathcal{U}^{(0)}(I\|z,w|)}{\lambda^{|I|} \mathcal{G}^{(0)}(z,w)}
$ and $\mathfrak{v}_{0,|I|}(I\|z)=-\lambda^{1-|I|}
\mathfrak{U}^{(0)}(I\|z)$.

The function $\tilde{\mathfrak{t}}_{0,|I|}(I\|z,q|)$ is regular at
every $z=-\hat{q}^j$. To see this we write in the last line of
(\ref{uIzq}) the denominator as $R(q)-R(-z)=R(-(-\hat{q}^j))-R(-z)$ and
insert the Taylor expansion (\ref{nabla-om-taylor}) (for $\varpi$) and
the usual Taylor expansion of
$\tilde{\mathfrak{t}}_{0,|I|}(I_2\|z,q|)$:
\begin{align*}
&
\frac{\varpi_{0,|I_1|+1}(I_1;z)
  \tilde{\mathfrak{t}}_{0,|I_2|}(I_2\|z,q|)}{dR(z)(R(q)-R(-z))}
\\
&=
\sum_{n,p=0}^\infty \frac{(-1)^n}{p!}
(R(z)-R(-\hat{q}^j))^{n+p-1} \nabla^n\varpi_{0,|I_1|+1}(I_1;-\hat{q}^j)
\frac{\partial^p \tilde{\mathfrak{t}}_{0,|I_2|}(I_2\|z,q|)}{\partial (R(z))^p}
\Big|_{z=-\hat{q}^j}\;.
\end{align*}
Inserted back into 
(\ref{uIzq}), the case $p=n=0$ cancels the term $l=j$ of the first line of 
(\ref{uIzq}) when taking $\nabla^0\varpi_{0,|I_1|+1}(I_1;-\hat{q}^j)
=\frac{\varpi_{0,|I_1|+1}(I_1;-\hat{q}^j)}{-dR(\hat{q}^j)}$ into account.
Hence, all partial derivatives of $\tilde{\mathfrak{t}}_{0,|I|}(I\|z,q|)$
are regular at $z=-\hat{q}^j$:
\begin{align}
&  \frac{(-1)^n}{n!}
\frac{\partial^n \tilde{\mathfrak{t}}_{0,|I|}(I\|z,q|)}{
  \partial (R(z))^n} \Big|_{z=-\hat{q}^j}
\label{uIzq-d}
  \\
  &=
-\sum_{\substack{I_1\uplus I_2=I\\
\text{possibly }    I_2= \emptyset}}
\sum_{\substack{l=1\\ l\neq j}}^d \frac{\varpi_{0,|I_1|+1}(I_1;{-}\hat{q}^l)
  \tilde{\mathfrak{t}}_{0,|I_2|}(I_2\|{-}\hat{q}^l,q|)}{dR(\hat{q}^l)
  (R(-\hat{q}^j)-R(-\hat{q}^l))^{n+1}}
\nonumber
\\
&+\sum_{k=1}^m \frac{\tilde{\mathfrak{t}}_{0,|I|-1}
  (I\backslash u_k\|u_k,q|)}{(R(-\hat{q}^j)-R(u_k))^{n+1}(R(q)-R(-u_k))}
\nonumber
\\
&+
\sum_{\substack{I_1\uplus I_2=I \\\text{possibly }    I_2= \emptyset}}
\sum_{p=0}^{n+1} 
\nabla^{n-p+1}\varpi_{0,|I_1|+1}(I_1;-\hat{q}^j)
\frac{(-1)^{p}}{p!}
\frac{\partial^p \tilde{\mathfrak{t}}_{0,|I_2|}(I_2\|z,q|)}{
  \partial (R(z))^p} \Big|_{z=-\hat{q}^j}\;.
\nonumber
\end{align}
Formulae (\ref{uIzq}) for $z\mapsto u_k$ and (\ref{uIzq-d}) provide a
system of equations whose resolution provides
$\tilde{\mathfrak{t}}_{0,|I|}(I\|{-}\hat{q}^l,q|)$ and
$\tilde{\mathfrak{t}}_{0,|I|-1}(I\backslash u_k\|u_k,q|)$ as
polynomials in $\nabla\varpi$ with coefficients in rational functions
of $R$.  Inserted into (\ref{om0-fraku}) we recursively express
$\varpi_{0,|I|+1}(I;z)$ in terms of $\nabla^n \varpi_{0,|I'|+1}(I';z)$
for $|I'|< |I|$. We find it convenient to develop a graphical
description for this resolution. With these tools we can establish:
\begin{theorem}
  \label{thm:sol-dse}
  Starting from $\varpi_{0,2}(u;z)=-\frac{dz}{(u-z)}-\frac{dz}{(u+z)}$,
  the system of equations  \eqref{om0-fraku}, \eqref{frakuIq},
  \eqref{uIzq} for $z\mapsto u_k$ and \eqref{uIzq-d} has the solution
\begin{align}
\varpi_{0,|I|+1}(I;z)
 &= \sum_{i=1}^r
\Res\displaylimits_{q\to \beta_i}K_i(z,q)
  \sum_{I_1\uplus I_2=I} \varpi_{0,|I_1|+1}(I_1;q)\varpi_{0,|I_2|+1}(I_2;\sigma_i(q))
\nonumber  \\
  &-\sum_{k=1}^m 
\Res\displaylimits_{q\to - u_k}
\sum_{I_1\uplus I_2=I}
\tilde{K}(z,q,u_k)
\varpi_{0,|I_1|+1}(I_1;q)
\varpi_{0,|I_2|+1}(I_2;q)\big)\;, 
\label{sol:omegaR}
\\
\text{where} \qquad 
K_i(z,q)&:=   \frac{\frac{1}{2} (\frac{dz}{z-q}-\frac{dz}{z-\sigma_i(q)})
}{dR(\sigma_i(q))(R(-\sigma_i(q))-R(-q))}\;,\qquad
\nonumber
\\
\tilde{K}(z,q,u)&:=
\frac{\frac{1}{2}\big(\frac{dz}{z-q}-\frac{dz}{z+ u}\big)}{
  dR(q)(R(u)-R(-q))}\;.
\nonumber
\end{align}
Hence, $\omega_{0,m+1}(u_1,...,u_m,z)=
d_{u_1}\cdots d_{u_m}\varpi_{0,m+1}(u_1,...,u_m;z)$
coincides with the solution of equation~\eqref{eq:flip-om}
for $x(z)=R(z)$ and $\iota (z)=-z$
given in Theorem~\ref{thm:flip}.
\end{theorem}

\subsection{Graphical description}

We introduce in Table~\ref{tab1} weighted functions, vertices and
edges. These are connected to chains which provide a graphical
description for the terms
$\frac{(-1)^p\partial^p}{\partial (R(z))^p}
\tilde{\mathfrak{t}}_{0,|I|}(I\|z,q|)\big|_{z={-}\hat{q}^j}$ and
$\tilde{\mathfrak{t}}_{0,|I|}(I\|u_k,q|)$ and its constituents.
\begin{table}[p]
\begin{tabular}{@{}|p{0.5cm}|c|p{6cm}|p{4cm}|}\hline
\# &   function & weight & remark 
  \\ \hline
  f1 & \begin{tikzpicture}
    \useasboundingbox (0,-.1) rectangle (2,.5);
    \draw[thick, densely dotted] (0,0) -- (1,0);
    \draw[gray,fill=gray!20] (1.2,0) circle (2ex) ;
    \node (l1) at (0,.3) {$j$};
    \node (r) at (1.2,0) {$I$};
  \end{tikzpicture}
  & $\tilde{\mathfrak{t}}_{0,|I|}(I\|{-}\hat{q}^j,q)$ 
  & $|I|\geq 0$ \par equals $1$ for $I=\emptyset$  
  \\
 & \begin{tikzpicture}
    \useasboundingbox (0,-.1) rectangle (2,.5);
    \draw[thick, dashed] (0,0) -- (1,0);
    \draw[gray,fill=gray!20] (1.2,0) circle (2ex) ;
    \node (l1) at (0,.3) {$j$};
    \node (r) at (1.2,0) {$I$};
  \end{tikzpicture}
  & $\tilde{\mathfrak{t}}_{0,|I|}(I\|{-}\hat{q}^j,q)$ 
                         & $I\neq \emptyset$ \par case $p=0$ of f2
\\  \hline
  f2$^p$ & \begin{tikzpicture}
    \useasboundingbox (0,-.1) rectangle (2,.6);
    \draw[thick, dashed] (0,0) -- (0.8,0);
    \draw[gray,fill=gray!20] (1.2,0) circle (1.8ex) ;
    \centerarc[](1.2,0)(30:330:2.6ex) ;
    \node (l1) at (0,.3) {$j$};
    \node (r) at (1.2,0) {$I$};
    \node (f) at (1.7,0) {\footnotesize$p$};
\end{tikzpicture}
                & $\dfrac{(-1)^p \partial^p}{\partial (R(z))^p}
                  \tilde{\mathfrak{t}}_{0,|I|}(I\|z,q)\big|_{z=-\hat{q}^j}$
  & $I\neq \emptyset$
  \\[2ex] \hline
  f3 & \begin{tikzpicture}
    \useasboundingbox (0,-.1) rectangle (2,.5);
    \draw[thick, densely dotted] (0,0) -- (1,0);
    \draw[fill=gray!20] (1,-.35) rectangle (1.7,.35) ;
    \node (l1) at (0,-.3) {$u$};
    \node (r) at (1.3,0) {$I$};
  \end{tikzpicture}
  & $\tilde{\mathfrak{t}}_{0,|I|}(I\|u,q)$ 
  & $|I|\geq 0$ \par equals $1$ for $I=\emptyset$  
  \\
 & \begin{tikzpicture}
    \useasboundingbox (0,-.1) rectangle (2,.5);
    \draw[thick, dashed] (0,0) -- (1,0);
    \draw[fill=gray!20] (1,-.35) rectangle (1.7,.35) ;
    \node (l1) at (0,-.3) {$u$};
    \node (r) at (1.3,0) {$I$};
  \end{tikzpicture}
  & $\tilde{\mathfrak{t}}_{0,|I|}(I\|u,q)$ 
                         & $I\neq \emptyset$ 
\\[2ex]  \hline\hline
\# &   vertex & weight & remark 
  \\ \hline
  v0 & \begin{tikzpicture}[baseline]
    \draw (1,0) circle (4pt);
    \node (l1) at (1,.35) {$0$};
    \node (l2) at (1,-.35) {$I$};
  \end{tikzpicture} &  $-\varpi_{0,|I|+1}(I;q)$ &
initial vertex 
  \\ \hline
  v1 & \begin{tikzpicture}[baseline]
    \draw[fill=black] (1,0) circle (4pt);
    \node (l1) at (1,.35) {$j$};
    \node (l2) at (1,-.35) {$I$};
  \end{tikzpicture} &  $\varpi_{0,|I|+1}(I;-\hat{q}^j)$ &
                      follows edges e1$^p$,e2,e6
  \\ \hline
  v2 & \begin{tikzpicture}[baseline]
    \useasboundingbox (0,-.5) rectangle (.5,.5);
    \draw (0,-0.05) rectangle (.3,.25);
    \node (l2) at (0.05,-.2) {$u$};
  \end{tikzpicture} &  $\dfrac{1}{R(q)-R(-u)}$ & follows edges
                                                       e3$^p$,e4
  \\ \hline
  v3 & \begin{tikzpicture}
    \useasboundingbox (0,-.3) rectangle (.5,.4);
    \draw (0,0) node[star,draw]  {\begin{picture}(0,0)
        \put(-3,-2.4){\mbox{\small$u$}}\end{picture}} ;
    \node (l2) at (0,-.45) {$I$};
  \end{tikzpicture} &
                      $\dfrac{\varpi_{0,|I|+1}(I;u)}{dR(u) (R(-u)-R(q))}$ &
                   follows edges e5 \rule[-5mm]{0mm}{4mm}
  \\ \hline  \hline
\# &   edge  & weight & remark 
  \\ \hline
e1$^p$ &  \begin{tikzpicture}[baseline]
    \useasboundingbox (0,-.7) rectangle (1,1);
    \draw[-{To[length=10]}] (0,0) --  (1,0);
    \node (n) at (0.5,.5) {\footnotesize$p$};
    \node (l1) at (0,.3) {$j$};
    \node (l2) at (1,.3) {$l$};
  \end{tikzpicture}&
                     $\displaystyle 
                     \frac{1}{(R(-\hat{q}^j)-R(-\hat{q}^l))^{p+1}
                     (-dR(\hat{q}^l))}$
                         &
                           \parbox[c]{4cm}{follows vertices v0,v1 \par
                           requires $l\neq j$ \par $\hat{q}^0{\equiv} q$,
                           no tip for $p=0$}
  \\ \hline
e2 &  \begin{tikzpicture}
    \useasboundingbox (0,-.1) rectangle (1,.6);
    \draw (0,0) --  (1,0);
    \node (l1) at (0,-.3) {$u$};
    \node (l2) at (1,.3) {$l$};
  \end{tikzpicture}&
   $\displaystyle 
  \frac{1}{(R(u)-R(-\hat{q}^l))(-dR(\hat{q}^l))}$   &
  follows vertices v2,v3 \rule[-5mm]{0mm}{4mm}
  \\ \hline
e3$^p$ &  \begin{tikzpicture}
    \useasboundingbox (0,.1) rectangle (1,.75);
    \draw[-{To[length=10]},decorate, decoration={snake}] (0,0) --  (1,0);
    \node (n) at (0.5,.5) {\footnotesize$p$};
    \node (l1) at (0,.3) {$j$};
    \node (l2) at (1,-.3) {$u$};
  \end{tikzpicture}&
                     $
  \!\!\dfrac{1}{(R({-}\hat{q}^j){-}R(u))^{p+1}}$   &
  follows vertices v0,v1\par no tip for $p=0$
  \\ \hline
e4 &  \begin{tikzpicture}
    \useasboundingbox (0,-.1) rectangle (1,.55);
    \draw[decorate, decoration={snake}] (0,0) --  (1,0);
    \node (l1) at (0,-.3) {$v$};
    \node (l2) at (1,-.3) {$u$};
  \end{tikzpicture}&
                     $
  \dfrac{1}{R(v)-R(u)}$   &
  follows vertices v2,v3 \par requires $u\neq v$
  \\ \hline
e5 &  \begin{tikzpicture}
    \useasboundingbox (0,-.2) rectangle (1,.25);
    \draw[decorate, decoration={zigzag}] (0,0) --  (1,0);
    \node (l1) at (0,-.3) {$u$};
    \node (l2) at (1,-.3) {$u$};
  \end{tikzpicture}&
                    $1$   &
  follows vertices v2,v3
  \\[1ex] \hline
e6$^n$ &  \begin{tikzpicture}
    \useasboundingbox (0,.2) rectangle (1,.65);
    \draw[->,cap=rect,double distance=3pt] (0,0) -- (1,0);
    \node (n) at (0.5,.5) {\footnotesize$n$};
    \node (l1) at (0,.3) {$j$};
    \node (l2) at (1,.3) {$j$};
  \end{tikzpicture}&
   $\nabla^n$ &
  follows vertices v1\par applies to next vertex
  \\ \hline
\end{tabular}
\caption{Graphical rules for building blocks of chains\label{tab1}}
\end{table}
We agree that arrow tips with label $p=0$ are not shown. Also the
surrounding circle segment indicating the $n$-th derivative with
respect to $R(z)$ is not shown for $n=0$.

Equation~(\ref{uIzq-d}) has for $|I|\geq 1$ the following graphical
description (we keep the order of the last three lines of
(\ref{uIzq-d})):
\begin{align}
\begin{tikzpicture}
      \useasboundingbox (0,-.1) rectangle (2.1,.3);
      \draw[thick, dashed] (0,0) -- (0.8,0);
    \node (l) at (0,.35) {$j$};
    \draw[gray,fill=gray!20] (1.2,0) circle (1.8ex) ;
    \node (r) at (1.2,0) {$I$};
    \centerarc[](1.2,0)(30:330:2.67ex) ;
    \node (f) at (1.7,0) {\footnotesize$n$};
  \end{tikzpicture}
  &= \sum_{\substack{l=1 \\ l\neq j}}^d
  \sum_{\substack{I_1\uplus I_2= I\\ \text{possibly } I_2=\emptyset}}
\begin{tikzpicture}
     \useasboundingbox (-.2,-.1) rectangle (3,.3);
     \draw[-{To[length=10]}] (0,0) -- (0.82,0);
    \node (n) at (.5,.5) {\footnotesize$n$};
    \draw[thick,densely dotted] (1,0) -- (2,0);
    \node (l) at (0,.35) {$j$};
    \draw[fill=black] (1,0) circle (4pt) ;
    \draw[gray,fill=gray!20] (2.2,0) circle (2ex) ;
    \node (m1) at (1,-.35) {$I_1$};
    \node (m2) at (1,.35) {$l$};
    \node (r) at (2.2,0) {$I_2$};
  \end{tikzpicture}
\nonumber
\\*[0.2ex]
&+\sum_{k=1}^{|I|} \begin{tikzpicture}
    \useasboundingbox (-0.5,-.1) rectangle (3,.3);
    \draw[-{To[length=10]},decorate, decoration={snake}] (0,0) -- (0.95,0);
    \node (n) at (.5,.5) {\footnotesize$n$};
    \draw[thick,densely dotted] (1.3,0) -- (2,0);
    \node (l) at (0,.35) {$j$};
    \draw (.98,-.12) rectangle (1.22,.12) ;
    \draw[gray,fill=gray!20] (2,-.35) rectangle (2.7,.35) ;
    \node (m1) at (1.1,-.37) {$u_k$};
    \node (r) at (2.35,0) {\scriptsize$I{\setminus}u_k$};
  \end{tikzpicture}
\nonumber
\\*[0.2ex]
&+  \sum_{\substack{I_1\uplus I_2= I\\ \text{possibly } I_2=\emptyset}}
\sum_{s=0}^{n+1} \begin{tikzpicture}
  \useasboundingbox (-0.5,0) rectangle (3.8,.5);
    \draw[->,cap=rect,double distance=3pt] (0,0) -- (0.8,0);
    \node (l) at (0,.35) {$j$};
    \draw[fill=black] (1,0) circle (4pt) ;
    \draw[thick,densely dotted] (1.2,0) -- (1.7,0);
    \draw[gray,fill=gray!20] (2.2,0) circle (1.8ex) ;
    \centerarc[](2.2,0)(25:335:2.6ex) ;
    \node (f) at (3.15,0) {\footnotesize$n{+}1{-}s$};
    \node (m1) at (1,-.45) {$I_1$};
    \node (m2) at (1,.45) {$j$};
    \node (r) at (2.2,0) {$I_2$};
    \node (n) at (.5,.5) {\footnotesize$s$};
  \end{tikzpicture}\;.
\label{calUIqq}
\end{align}

Similarly, equation (\ref{uIzq}) is for $|I|\geq 1$ represented as (we
keep the order of lines)
\begin{align}
\begin{tikzpicture}
  \useasboundingbox (1,-.2) rectangle (3.2,.3);
    \draw[thick, dashed] (1.3,0) -- (2,0);
    \draw[fill=gray!20] (2,-.35) rectangle (2.7,.35) ;
    \node (m1) at (1.3,-.25) {$u$};
    \node (r) at (2.35,0) {$I$};
  \end{tikzpicture}
  &=
\sum_{\substack{I_1\uplus I_2=I\\
    \text{possibly }I_2= \emptyset}}\sum_{j=1}^d
\begin{tikzpicture}
    \useasboundingbox (-.5,-.1) rectangle (3,.7);
    \draw (0,0) -- (1,0);
    \draw[thick, densely dotted] (1,0) -- (2,0);
    \node (l) at (0,-.35) {$u$};
    \draw[gray,fill=black] (1,0) circle (4pt) ;
    \draw[gray,fill=gray!20] (2.2,0) circle (2ex) ;
    \node (m1) at (1,-.35) {$I_1$};
    \node (m2) at (1,.35) {$j$};
    \node (r) at (2.2,0) {$I_2$};
  \end{tikzpicture}
\nonumber
\\
&+\sum_{k=1}^{|I|} 
\begin{tikzpicture}
    \useasboundingbox (-0.5,-.1) rectangle (3.2,.4);
    \draw[decorate, decoration={snake}] (0,0) -- (0.95,0);
    \draw[thick, densely dotted] (1.3,0) -- (2,0);
    \node (l) at (0,-.35) {$u$};
    \draw (.98,-.12) rectangle (1.22,.12) ;
    \draw[gray,fill=gray!20] (2,-.35) rectangle (2.7,.35) ;
    \node (m1) at (1.1,-.37) {$u_k$};
    \node (r) at (2.35,0) {\scriptsize$I{\setminus}u_k$};
  \end{tikzpicture}
\nonumber
\\
&+ \sum_{\substack{ I_1\uplus I_2=I\\
    \text{possibly } I_2=\emptyset}}
\begin{tikzpicture}
    \useasboundingbox (-0.5,-.1) rectangle (2.8,.4);
    \draw[decorate, decoration={zigzag}] (0,0) -- (0.9,0);
    \draw[thick, densely dotted] (1.35,0) -- (2,0);
    \node (l) at (0,-.35) {$u$};
    \draw (1.1,0) node[star,draw]  {\begin{picture}(1,1)
        \put(-3.5,-2){\mbox{\small$u$}}\end{picture}} ;
    \draw[gray,fill=gray!20] (2,-.35) rectangle (2.7,.35) ;
    \node (m1) at (1.1,-.45) {$I_1$};
    \node (r) at (2.35,0) {$I_2$};
  \end{tikzpicture}\;.
\end{align}

The integrand of the first line of (\ref{om0-fraku}) is now
iteratively obtained by distinguishing in
$\tilde{\mathfrak{t}}_{0,|I|}(I\|z,q|)$ the cases $I=\emptyset$ from
$I\neq \emptyset$.  We describe this iteration graphically.  The
integrand of the first line of (\ref{om0-fraku}) is the sum of weights
of chains made of initial vertex v0, subsequent vertices v1, v2, v3
and edges in between.  A vertex v3 can follow v2 or another v3,
whereas v1,v2 can be placed anywhere. The edge to choose is governed
by the type of vertices at both ends.  One multiplies the weights
given in Table~\ref{tab1} and sums for each order of vertices over
partitions of $I$ into subsets $I_1,...,I_s,u_{k_1},....,u_{k_r}$ at
the vertices, over the v1-labels $j,l,\dots$ (from $1$ to $d$, but
excluding the preceding label) and over the possible exponents $n,p$
of the edges e1$^p$,e3$^p$ and e6$^n$. These exponents are not
arbitrary; we discuss later their pattern.  

\subsection{Examples}

We write the first iteration in full details:
\begin{align}
&
\sum_{I_1\uplus I_2=I}
\varpi_{0,|I_1|+1}(I_1;q) \mathfrak{v}_{0,|I_2|}(I_2\|q)
\label{omu-it1}
\\
&= \sum_{I_1\uplus I_2=I}\sum_{j=1}^d
\begin{tikzpicture}
    \useasboundingbox (-.3,-0.1) rectangle (1.3,.4);
  \draw (0,0) circle (4pt) ;
  \node (l1) at (0,0.35) {$0$};
  \node (l2) at (0,-0.4) {$I_1$};
    \draw (0.2,0) -- (1,0);
    \draw[fill=black] (1,0) circle (4pt) ;
    \node (m1) at (1,-.4) {$I_2$};
    \node (m2) at (1,.35) {$j$};
 \end{tikzpicture} 
+\sum_{k=1}^{|I|} \sum_{I_1\uplus u_k=I} \begin{tikzpicture}
    \useasboundingbox (-.3,-0.1) rectangle (1.5,.4);
    \draw (0,0) circle (4pt) ;
    \node (l1) at (0,0.35) {$0$};
    \node (l2) at (0,-0.4) {$I_1$};
    \draw[decorate, decoration={snake}] (0.2,0) -- (0.95,0);
    \draw (.98,-.12) rectangle (1.22,.12) ;
    \node (m1) at (1.1,-.37) {$u_k$};
  \end{tikzpicture}
  \nonumber
  \\
  &+\sum_{I_1\uplus I_2\uplus I_3=I}\sum_{j=1}^d
  \begin{tikzpicture}
    \useasboundingbox (-.3,-0.1) rectangle (2.8,.4);
  \draw (0,0) circle (4pt) ;
  \node (l1) at (0,0.35) {$0$};
  \node (l2) at (0,-0.4) {$I_1$};
    \draw (0.2,0) -- (1,0);
    \draw[thick, dashed] (1,0) -- (2,0);
    \draw[fill=black] (1,0) circle (4pt) ;
    \draw[gray,fill=gray!20] (2.2,0) circle (2ex) ;
    \node (m1) at (1,-.4) {$I_2$};
    \node (m2) at (1,.35) {$j$};
    \node (r) at (2.2,0) {$I_3$};
 \end{tikzpicture}
+\sum_{k=1}^{|I|} \sum_{I_1\uplus u_k\uplus I_2=I} \begin{tikzpicture}
    \useasboundingbox (-.3,-0.1) rectangle (2.8,.4);
    \draw (0,0) circle (4pt) ;
    \node (l1) at (0,0.35) {$0$};
    \node (l2) at (0,-0.4) {$I_1$};
    \draw[decorate, decoration={snake}] (0.2,0) -- (0.95,0);
    \draw[thick, dashed] (1.3,0) -- (2,0);
    \draw (.98,-.12) rectangle (1.22,.12) ;
    \draw[fill=gray!20] (2,-.35) rectangle (2.7,.35) ;
    \node (m1) at (1.1,-.37) {$u_k$};
    \node (r) at (2.35,0) {$I_2$};
  \end{tikzpicture}\;.
  \nonumber
\end{align}
The necessary sum over partitions of $I$ and over ranges of labels $j$
are obvious from the vertex labels.
We therefore employ from now on a simplified notation were these summations 
are omitted. This means that instead of (\ref{omu-it1}) we simply write
\begin{align*}
&
\sum_{I_1\uplus I_2=I}
\varpi_{0,|I_1|+1}(I_1;q) \mathfrak{v}_{0,|I_2|}(I_2\|q)
\\
&= 
\begin{tikzpicture}
    \useasboundingbox (-.3,-0.1) rectangle (1.3,.4);
  \draw (0,0) circle (4pt) ;
  \node (l1) at (0,0.35) {$0$};
  \node (l2) at (0,-0.4) {$I_1$};
    \draw (0.2,0) -- (1,0);
    \draw[fill=black] (1,0) circle (4pt) ;
    \node (m1) at (1,-.4) {$I_2$};
    \node (m2) at (1,.35) {$j$};
 \end{tikzpicture} 
+\begin{tikzpicture}
    \useasboundingbox (-.3,-0.1) rectangle (1.5,.4);
    \draw (0,0) circle (4pt) ;
    \node (l1) at (0,0.35) {$0$};
    \node (l2) at (0,-0.4) {$I_1$};
    \draw[decorate, decoration={snake}] (0.2,0) -- (0.95,0);
    \draw (.98,-.12) rectangle (1.22,.12) ;
    \node (m1) at (1.1,-.37) {$u_k$};
  \end{tikzpicture}
+
  \begin{tikzpicture}
    \useasboundingbox (-.3,-0.1) rectangle (2.8,.4);
  \draw (0,0) circle (4pt) ;
  \node (l1) at (0,0.35) {$0$};
  \node (l2) at (0,-0.4) {$I_1$};
    \draw (0.2,0) -- (1,0);
    \draw[thick, dashed] (1,0) -- (2,0);
    \draw[fill=black] (1,0) circle (4pt) ;
    \draw[gray,fill=gray!20] (2.2,0) circle (2ex) ;
    \node (m1) at (1,-.4) {$I_2$};
    \node (m2) at (1,.35) {$j$};
    \node (r) at (2.2,0) {$I_3$};
 \end{tikzpicture}
 +
 \begin{tikzpicture}
    \useasboundingbox (-.3,-0.1) rectangle (2.8,.4);
    \draw (0,0) circle (4pt) ;
    \node (l1) at (0,0.35) {$0$};
    \node (l2) at (0,-0.4) {$I_1$};
    \draw[decorate, decoration={snake}] (0.2,0) -- (0.95,0);
    \draw[thick, dashed] (1.3,0) -- (2,0);
    \draw (.98,-.12) rectangle (1.22,.12) ;
    \draw[fill=gray!20] (2,-.35) rectangle (2.7,.35) ;
    \node (m1) at (1.1,-.37) {$u_k$};
    \node (r) at (2.35,0) {$I_2$};
  \end{tikzpicture}\;.\rule[-4mm]{0pt}{4mm}
\end{align*}
For $|I|=2$ only the first two chains contribute. 
The next iteration reads in simplified notation
\begin{align*}
&
\sum_{I_1\uplus I_2=I}
\varpi_{0,|I_1|+1}(I_1;q) \mathfrak{v}_{0,|I_2|}(I_2\|q)
\\
&= 
\begin{tikzpicture}
    \useasboundingbox (-.3,-0.1) rectangle (1.3,.4);
  \draw (0,0) circle (4pt) ;
  \node (l1) at (0,0.35) {$0$};
  \node (l2) at (0,-0.4) {$I_1$};
    \draw (0.2,0) -- (1,0);
    \draw[fill=black] (1,0) circle (4pt) ;
    \node (m1) at (1,-.4) {$I_2$};
    \node (m2) at (1,.35) {$j$};
 \end{tikzpicture} 
 + \begin{tikzpicture}
    \useasboundingbox (-.3,-0.1) rectangle (1.5,.4);
    \draw (0,0) circle (4pt) ;
    \node (l1) at (0,0.35) {$0$};
    \node (l2) at (0,-0.4) {$I_1$};
    \draw[decorate, decoration={snake}] (0.2,0) -- (0.95,0);
    \draw (.98,-.12) rectangle (1.22,.12) ;
    \node (m1) at (1.1,-.37) {$u_k$};
  \end{tikzpicture}
  \\[2ex]
  &+\begin{tikzpicture}
    \useasboundingbox (-.3,-0.1) rectangle (2.4,.4);
  \draw (0,0) circle (4pt) ;
  \node (l1) at (0,0.35) {$0$};
  \node (l2) at (0,-0.4) {$I_1$};
  \draw[fill=black] (1,0) circle (4pt) ;
      \draw (0.2,0) -- (1,0);
    \node (m1) at (1,-.4) {$I_2$};
    \node (m2) at (1,.35) {$j_1$};
    \draw[fill=black] (2,0) circle (4pt) ;
    \draw (1.2,0) -- (2,0);
    \node (m21) at (2,-.4) {$I_3$};
    \node (m22) at (2,.35) {$j_2$};
  \end{tikzpicture}
+\begin{tikzpicture}
    \useasboundingbox (-.3,-0.1) rectangle (2.4,.4);
  \draw (0,0) circle (4pt) ;
  \node (l1) at (0,0.35) {$0$};
  \node (l2) at (0,-0.4) {$I_1$};
    \draw (0.2,0) -- (1,0);
    \draw[fill=black] (1,0) circle (4pt) ;
    \node (m1) at (1,-.4) {$I_2$};
    \node (m2) at (1,.35) {$j$};
    \draw[decorate, decoration={snake}] (1.2,0) -- (1.95,0);
    \draw (1.98,-.12) rectangle (2.22,.12) ;
    \node (m21) at (2.1,-.37) {$u_k$};
  \end{tikzpicture}
+\begin{tikzpicture}
    \useasboundingbox (-.3,-0.1) rectangle (3.8,.4);
  \draw (0,0) circle (4pt) ;
  \node (l1) at (0,0.35) {$0$};
  \node (l2) at (0,-0.4) {$I_1$};
  \draw[fill=black] (1,0) circle (4pt) ;
      \draw (0.2,0) -- (1,0);
    \node (m1) at (1,-.4) {$I_2$};
    \node (m2) at (1,.35) {$j$};
    \draw[fill=black] (2,0) circle (4pt) ;
    \draw[->,cap=rect,double distance=3pt]  (1.2,0) -- (1.8,0); 
    \node (m21) at (2,-.4) {$I_3$};
    \node (m22) at (2,.35) {$j$};
    \node (n) at (1.5,.55) {\footnotesize$1$};
  \end{tikzpicture}
\\[2ex]
&
+ \begin{tikzpicture}
    \useasboundingbox (-.3,-0.1) rectangle (2.4,.4);
    \draw (0,0) circle (4pt) ;
    \node (l1) at (0,0.35) {$0$};
    \node (l2) at (0,-0.4) {$I_1$};
    \draw[decorate, decoration={snake}] (0.2,0) -- (0.95,0);
    \draw (.98,-.12) rectangle (1.22,.12);
    \node (m1) at (1.1,-.37) {$u_k$};
    \draw[fill=black] (2,0) circle (4pt) ;
    \draw (1.25,0) -- (2,0);
    \node (m21) at (2,-.4) {$I_2$};
    \node (m22) at (2,.35) {$j$};
  \end{tikzpicture}
+ \begin{tikzpicture}
    \useasboundingbox (-.3,-0.1) rectangle (2.4,.4);
    \draw (0,0) circle (4pt) ;
    \node (l1) at (0,0.35) {$0$};
    \node (l2) at (0,-0.4) {$I_1$};
    \draw[decorate, decoration={snake}] (0.2,0) -- (0.95,0);
    \draw (.98,-.12) rectangle (1.22,.12) ;
    \node (m1) at (1.1,-.37) {$u_k$};
    \draw[decorate, decoration={snake}] (1.2,0) -- (1.95,0);
    \draw (1.98,-.12) rectangle (2.22,.12) ;
    \node (m1) at (2.1,-.37) {$u_l$};
  \end{tikzpicture}
  + \begin{tikzpicture}
    \useasboundingbox (-.3,-0.1) rectangle (2.4,.4);
    \draw (0,0) circle (4pt) ;
    \node (l1) at (0,0.35) {$0$};
    \node (l2) at (0,-0.4) {$I_1$};
    \draw[decorate, decoration={snake}] (0.2,0) -- (0.95,0);
    \draw (.98,-.12) rectangle (1.22,.12) ;
    \node (m1) at (1.1,-.37) {$u_k$};
    \draw[decorate, decoration={zigzag}] (1.2,0) -- (1.9,0);
    \draw (2.1,0) node[star,draw]  {\begin{picture}(1,1)
        \put(-5,-2){\mbox{\small$u_k$}}\end{picture}} ;
    \node (m1) at (2.1,-.45) {$I_2$};
  \end{tikzpicture}
 \\[3ex]
  &+\begin{tikzpicture}
    \useasboundingbox (-.3,-0.1) rectangle (3.8,.4);
  \draw (0,0) circle (4pt) ;
  \node (l1) at (0,0.35) {$0$};
  \node (l2) at (0,-0.4) {$I_1$};
  \draw[fill=black] (1,0) circle (4pt) ;
      \draw (0.2,0) -- (1,0);
    \node (m1) at (1,-.4) {$I_2$};
    \node (m2) at (1,.35) {$j_1$};
    \draw[fill=black] (2,0) circle (4pt) ;
    \draw (1.2,0) -- (2,0);
    \node (m21) at (2,-.4) {$I_3$};
    \node (m22) at (2,.35) {$j_2$};
    \draw[thick, dashed] (2,0) -- (3,0);
    \draw[fill=gray!20] (3.2,0) circle (2ex) ;
    \node (r) at (3.2,0) {$I_4$};
  \end{tikzpicture}
+\begin{tikzpicture}
    \useasboundingbox (-.3,-0.1) rectangle (3.8,.4);
  \draw (0,0) circle (4pt) ;
  \node (l1) at (0,0.35) {$0$};
  \node (l2) at (0,-0.4) {$I_1$};
    \draw (0.2,0) -- (1,0);
    \draw[fill=black] (1,0) circle (4pt) ;
    \node (m1) at (1,-.4) {$I_2$};
    \node (m2) at (1,.35) {$j$};
    \draw[decorate, decoration={snake}] (1.2,0) -- (1.95,0);
    \draw (1.98,-.12) rectangle (2.22,.12) ;
    \node (m21) at (2.1,-.37) {$u_k$};
    \draw[thick, dashed] (2.3,0) -- (3,0);
    \draw[fill=gray!20] (3,-.35) rectangle (3.7,.35) ;
    \node (r) at (3.35,0) {$I_3$};
  \end{tikzpicture}
+
  \begin{tikzpicture}
    \useasboundingbox (-.3,-0.1) rectangle (3.8,.4);
  \draw (0,0) circle (4pt) ;
  \node (l1) at (0,0.35) {$0$};
  \node (l2) at (0,-0.4) {$I_1$};
  \draw[fill=black] (1,0) circle (4pt) ;
      \draw (0.2,0) -- (1,0);
    \node (m1) at (1,-.4) {$I_2$};
    \node (m2) at (1,.35) {$j$};
    \draw[fill=black] (2,0) circle (4pt) ;
    \draw[->,cap=rect,double distance=3pt]  (1.2,0) -- (1.8,0); 
    \node (m21) at (2,-.4) {$I_3$};
    \node (m22) at (2,.35) {$j$};
    \draw[thick, dashed] (2,0) -- (3,0);
    \draw[fill=gray!20] (3.2,0) circle (2ex) ;
    \node (r) at (3.2,0) {$I_4$};
    \node (n) at (1.5,.55) {\footnotesize$1$};
  \end{tikzpicture}
 \\[2.5ex]
& +\begin{tikzpicture}
    \useasboundingbox (-.3,-0.1) rectangle (3.8,.4);
  \draw (0,0) circle (4pt) ;
  \node (l1) at (0,0.35) {$0$};
  \node (l2) at (0,-0.4) {$I_1$};
  \draw[fill=black] (1,0) circle (4pt) ;
      \draw (0.2,0) -- (1,0);
    \node (m1) at (1,-.4) {$I_2$};
    \node (m2) at (1,.35) {$j$};
    \draw[fill=black] (2,0) circle (4pt) ;
    \draw[cap=rect,double distance=3pt]  (1.2,0) -- (1.8,0); 
    \node (m21) at (2,-.4) {$I_3$};
    \node (m22) at (2,.35) {$j$};
    \draw[thick, dashed] (2,0) -- (2.8,0);
    \draw[fill=gray!20] (3.2,0) circle (1.8ex) ;
    \centerarc[](3.2,0)(30:330:2.5ex) ;
    \node (r) at (3.2,0) {$I_4$};
    \node (n) at (3.65,0) {\footnotesize$1$};    
  \end{tikzpicture}
\\[2ex]
&
+ \begin{tikzpicture}
    \useasboundingbox (-.3,-0.1) rectangle (3.8,.4);
    \draw (0,0) circle (4pt) ;
    \node (l1) at (0,0.35) {$0$};
    \node (l2) at (0,-0.4) {$I_1$};
    \draw[decorate, decoration={snake}] (0.2,0) -- (0.95,0);
    \draw (.98,-.12) rectangle (1.22,.12);
    \node (m1) at (1.1,-.37) {$u_k$};
    \draw[fill=black] (2,0) circle (4pt) ;
    \draw (1.25,0) -- (2,0);
    \node (m21) at (2,-.4) {$I_2$};
    \node (m22) at (2,.35) {$j$};
    \draw[thick, dashed] (2,0) -- (3,0);
    \draw[fill=gray!20] (3.2,0) circle (2ex) ;
    \node (r) at (3.2,0) {$I_3$};
  \end{tikzpicture}
+ \begin{tikzpicture}
    \useasboundingbox (-.3,-0.1) rectangle (3.8,.4);
    \draw (0,0) circle (4pt) ;
    \node (l1) at (0,0.35) {$0$};
    \node (l2) at (0,-0.4) {$I_1$};
    \draw[decorate, decoration={snake}] (0.2,0) -- (0.95,0);
    \draw (.98,-.12) rectangle (1.22,.12) ;
    \node (m1) at (1.1,-.37) {$u_k$};
    \draw[decorate, decoration={snake}] (1.2,0) -- (1.95,0);
    \draw (1.98,-.12) rectangle (2.22,.12) ;
    \node (m1) at (2.1,-.37) {$u_l$};
    \draw[thick, dashed] (2.3,0) -- (3,0);
    \draw[fill=gray!20] (3,-.35) rectangle (3.7,.35) ;
    \node (r) at (3.35,0) {$I_2$};
  \end{tikzpicture}
  + \begin{tikzpicture}
    \useasboundingbox (-.3,-0.1) rectangle (3.8,.4);
    \draw (0,0) circle (4pt) ;
    \node (l1) at (0,0.35) {$0$};
    \node (l2) at (0,-0.4) {$I_1$};
    \draw[decorate, decoration={snake}] (0.2,0) -- (0.95,0);
    \draw (.98,-.12) rectangle (1.22,.12) ;
    \node (m1) at (1.1,-.37) {$u_k$};
    \draw[decorate, decoration={zigzag}] (1.2,0) -- (1.9,0);
    \draw (2.1,0) node[star,draw]  {\begin{picture}(1,1)
        \put(-5,-2){\mbox{\small$u_k$}}\end{picture}} ;
    \node (m1) at (2.1,-.45) {$I_2$};
    \draw[thick, dashed] (2.35,0) -- (3,0);
    \draw[fill=gray!20] (2.9,-.35) rectangle (3.6,.35) ;
    \node (r) at (3.35,0) {$I_3$};
  \end{tikzpicture}\;. \rule[-4mm]{0pt}{4mm}
\end{align*}  
For $|I|=3$ only the first three lines of the rhs are relevant.
We give another iteration, but stop it at $|I|=4$:
\begin{align*}
&
\sum_{I_1\uplus I_2=I}
\varpi_{0,|I_1|+1}(I_1;q) \mathfrak{v}_{0,|I_2|}(I_2\|q)
\\
&= 
\begin{tikzpicture}
    \useasboundingbox (-.3,-0.1) rectangle (1.3,.4);
  \draw (0,0) circle (4pt) ;
  \node (l1) at (0,0.35) {$0$};
  \node (l2) at (0,-0.4) {$I_1$};
    \draw (0.2,0) -- (1,0);
    \draw[fill=black] (1,0) circle (4pt) ;
    \node (m1) at (1,-.4) {$I_2$};
    \node (m2) at (1,.35) {$j$};
 \end{tikzpicture} 
 + \begin{tikzpicture}
    \useasboundingbox (-.3,-0.1) rectangle (1.5,.4);
    \draw (0,0) circle (4pt) ;
    \node (l1) at (0,0.35) {$0$};
    \node (l2) at (0,-0.4) {$I_1$};
    \draw[decorate, decoration={snake}] (0.2,0) -- (0.95,0);
    \draw (.98,-.12) rectangle (1.22,.12) ;
    \node (m1) at (1.1,-.37) {$u_k$};
  \end{tikzpicture}
  \\[2ex]
  &+\begin{tikzpicture}
    \useasboundingbox (-.3,-0.1) rectangle (2.4,.4);
  \draw (0,0) circle (4pt) ;
  \node (l1) at (0,0.35) {$0$};
  \node (l2) at (0,-0.4) {$I_1$};
  \draw[fill=black] (1,0) circle (4pt) ;
      \draw (0.2,0) -- (1,0);
    \node (m1) at (1,-.4) {$I_2$};
    \node (m2) at (1,.35) {$j_1$};
    \draw[fill=black] (2,0) circle (4pt) ;
    \draw (1.2,0) -- (2,0);
    \node (m21) at (2,-.4) {$I_3$};
    \node (m22) at (2,.35) {$j_2$};
  \end{tikzpicture}
+\begin{tikzpicture}
    \useasboundingbox (-.3,-0.1) rectangle (2.4,.4);
  \draw (0,0) circle (4pt) ;
  \node (l1) at (0,0.35) {$0$};
  \node (l2) at (0,-0.4) {$I_1$};
    \draw (0.2,0) -- (1,0);
    \draw[fill=black] (1,0) circle (4pt) ;
    \node (m1) at (1,-.4) {$I_2$};
    \node (m2) at (1,.35) {$j$};
    \draw[decorate, decoration={snake}] (1.2,0) -- (1.95,0);
    \draw (1.98,-.12) rectangle (2.22,.12) ;
    \node (m21) at (2.1,-.37) {$u_k$};
  \end{tikzpicture}
+\begin{tikzpicture}
    \useasboundingbox (-.3,-0.1) rectangle (3.8,.4);
  \draw (0,0) circle (4pt) ;
  \node (l1) at (0,0.35) {$0$};
  \node (l2) at (0,-0.4) {$I_1$};
  \draw[fill=black] (1,0) circle (4pt) ;
      \draw (0.2,0) -- (1,0);
    \node (m1) at (1,-.4) {$I_2$};
    \node (m2) at (1,.35) {$j$};
    \draw[fill=black] (2,0) circle (4pt) ;
    \draw[->,cap=rect,double distance=3pt]  (1.2,0) -- (1.8,0); 
    \node (m21) at (2,-.4) {$I_3$};
    \node (m22) at (2,.35) {$j$};
    \node (n) at (1.5,.55) {\footnotesize$1$};
  \end{tikzpicture}
\\[2ex]
&
+ \begin{tikzpicture}
    \useasboundingbox (-.3,-0.1) rectangle (2.4,.4);
    \draw (0,0) circle (4pt) ;
    \node (l1) at (0,0.35) {$0$};
    \node (l2) at (0,-0.4) {$I_1$};
    \draw[decorate, decoration={snake}] (0.2,0) -- (0.95,0);
    \draw (.98,-.12) rectangle (1.22,.12);
    \node (m1) at (1.1,-.37) {$u_k$};
    \draw[fill=black] (2,0) circle (4pt) ;
    \draw (1.25,0) -- (2,0);
    \node (m21) at (2,-.4) {$I_2$};
    \node (m22) at (2,.35) {$j$};
  \end{tikzpicture}
+ \begin{tikzpicture}
    \useasboundingbox (-.3,-0.1) rectangle (2.4,.4);
    \draw (0,0) circle (4pt) ;
    \node (l1) at (0,0.35) {$0$};
    \node (l2) at (0,-0.4) {$I_1$};
    \draw[decorate, decoration={snake}] (0.2,0) -- (0.95,0);
    \draw (.98,-.12) rectangle (1.22,.12) ;
    \node (m1) at (1.1,-.37) {$u_k$};
    \draw[decorate, decoration={snake}] (1.2,0) -- (1.95,0);
    \draw (1.98,-.12) rectangle (2.22,.12) ;
    \node (m1) at (2.1,-.37) {$u_l$};
  \end{tikzpicture}
  + \begin{tikzpicture}
    \useasboundingbox (-.3,-0.1) rectangle (2.4,.4);
    \draw (0,0) circle (4pt) ;
    \node (l1) at (0,0.35) {$0$};
    \node (l2) at (0,-0.4) {$I_1$};
    \draw[decorate, decoration={snake}] (0.2,0) -- (0.95,0);
    \draw (.98,-.12) rectangle (1.22,.12) ;
    \node (m1) at (1.1,-.37) {$u_k$};
    \draw[decorate, decoration={zigzag}] (1.2,0) -- (1.9,0);
    \draw (2.1,0) node[star,draw]  {\begin{picture}(1,1)
        \put(-5,-2){\mbox{\small$u_k$}}\end{picture}} ;
    \node (m1) at (2.1,-.45) {$I_2$};
  \end{tikzpicture}
 \\[3ex]
  &+\begin{tikzpicture}
    \useasboundingbox (-.3,-0.1) rectangle (3.4,.4);
  \draw (0,0) circle (4pt) ;
  \node (l1) at (0,0.35) {$0$};
  \node (l2) at (0,-0.4) {$I_1$};
  \draw[fill=black] (1,0) circle (4pt) ;
      \draw (0.2,0) -- (1,0);
    \node (m1) at (1,-.4) {$I_2$};
    \node (m2) at (1,.35) {$j_1$};
    \draw[fill=black] (2,0) circle (4pt) ;
    \draw (1.2,0) -- (2,0);
    \node (m21) at (2,-.4) {$I_3$};
    \node (m22) at (2,.35) {$j_2$};
    \draw[fill=black] (3,0) circle (4pt) ;
    \draw (2.2,0) -- (3,0);
    \node (m31) at (3,-.4) {$I_4$};
    \node (m32) at (3,.35) {$j_3$};
  \end{tikzpicture}
  +\begin{tikzpicture}
    \useasboundingbox (-.3,-0.1) rectangle (3.4,.4);
  \draw (0,0) circle (4pt) ;
  \node (l1) at (0,0.35) {$0$};
  \node (l2) at (0,-0.4) {$I_1$};
  \draw[fill=black] (1,0) circle (4pt) ;
      \draw (0.2,0) -- (1,0);
    \node (m1) at (1,-.4) {$I_2$};
    \node (m2) at (1,.35) {$j_1$};
    \draw[fill=black] (2,0) circle (4pt) ;
    \draw (1.2,0) -- (2,0);
    \node (m21) at (2,-.4) {$I_3$};
    \node (m22) at (2,.35) {$j_2$};
    \draw  (2.9,-.12) rectangle (3.14,.12) ;
    \draw[decorate, decoration={snake}] (2.2,0) -- (2.9,0);
    \node (m31) at (3,-.4) {$u_k$};
  \end{tikzpicture}
  +\begin{tikzpicture}
    \useasboundingbox (-.3,-0.1) rectangle (3.4,.4);
  \draw (0,0) circle (4pt) ;
  \node (l1) at (0,0.35) {$0$};
  \node (l2) at (0,-0.4) {$I_1$};
  \draw[fill=black] (1,0) circle (4pt) ;
      \draw (0.2,0) -- (1,0);
    \node (m1) at (1,-.4) {$I_2$};
    \node (m2) at (1,.35) {$j_1$};
    \draw[fill=black] (2,0) circle (4pt) ;
    \draw (1.2,0) -- (2,0);
    \node (m21) at (2,-.4) {$I_3$};
    \node (m22) at (2,.35) {$j_2$};
    \draw[fill=black] (3,0) circle (4pt) ;
    \draw[->,cap=rect,double distance=3pt] (2.2,0) -- (2.8,0);
    \node (m31) at (3,-.4) {$I_4$};
    \node (m32) at (3,.35) {$j_2$};
    \node (n) at (2.5,.55) {\footnotesize$1$};
  \end{tikzpicture}
  \\[2ex]
  &+\begin{tikzpicture}
    \useasboundingbox (-.3,-0.1) rectangle (3.4,.4);
  \draw (0,0) circle (4pt) ;
  \node (l1) at (0,0.35) {$0$};
  \node (l2) at (0,-0.4) {$I_1$};
    \draw (0.2,0) -- (1,0);
    \draw[fill=black] (1,0) circle (4pt) ;
    \node (m1) at (1,-.4) {$I_2$};
    \node (m2) at (1,.35) {$j_1$};
    \draw[decorate, decoration={snake}] (1.2,0) -- (1.95,0);
    \draw (1.98,-.12) rectangle (2.22,.12) ;
    \node (m21) at (2.1,-.37) {$u_k$};
    \draw[fill=black] (3,0) circle (4pt) ;
    \draw (2.2,0) -- (3,0);
    \node (m31) at (3,-.4) {$I_3$};
    \node (m32) at (3,.35) {$j_2$};
  \end{tikzpicture}
+\begin{tikzpicture}
    \useasboundingbox (-.3,-0.1) rectangle (3.4,.4);
  \draw (0,0) circle (4pt) ;
  \node (l1) at (0,0.35) {$0$};
  \node (l2) at (0,-0.4) {$I_1$};
    \draw (0.2,0) -- (1,0);
    \draw[fill=black] (1,0) circle (4pt) ;
    \node (m1) at (1,-.4) {$I_2$};
    \node (m2) at (1,.35) {$j$};
    \draw[decorate, decoration={snake}] (1.2,0) -- (1.95,0);
    \draw (1.98,-.12) rectangle (2.22,.12) ;
    \node (m21) at (2.1,-.37) {$u_k$};
    \draw  (2.9,-.12) rectangle (3.14,.12) ;
    \draw[decorate, decoration={snake}] (2.2,0) -- (2.9,0);
    \node (m31) at (3,-.4) {$u_l$};
  \end{tikzpicture}
+\begin{tikzpicture}
    \useasboundingbox (-.3,-0.1) rectangle (3.4,.4);
  \draw (0,0) circle (4pt) ;
  \node (l1) at (0,0.35) {$0$};
  \node (l2) at (0,-0.4) {$I_1$};
    \draw (0.2,0) -- (1,0);
    \draw[fill=black] (1,0) circle (4pt) ;
    \node (m1) at (1,-.4) {$I_2$};
    \node (m2) at (1,.35) {$j$};
    \draw[decorate, decoration={snake}] (1.2,0) -- (1.95,0);
    \draw (1.98,-.12) rectangle (2.22,.12) ;
    \node (m21) at (2.1,-.37) {$u_k$};
    \draw[decorate, decoration={zigzag}] (2.2,0) -- (2.9,0);
    \draw (3.1,0) node[star,draw]  {\begin{picture}(1,1)\put(-5,-2){\mbox{\small$u_k$}}\end{picture}} ;
    \node (m1) at (3.1,-.45) {$I_3$};
  \end{tikzpicture}
 \\[3ex]
  &+\begin{tikzpicture}
    \useasboundingbox (-.3,-0.1) rectangle (3.4,.4);
  \draw (0,0) circle (4pt) ;
  \node (l1) at (0,0.35) {$0$};
  \node (l2) at (0,-0.4) {$I_1$};
  \draw[fill=black] (1,0) circle (4pt) ;
      \draw (0.2,0) -- (1,0);
    \node (m1) at (1,-.4) {$I_2$};
    \node (m2) at (1,.35) {$j_1$};
    \draw[fill=black] (2,0) circle (4pt) ;
    \draw[->,cap=rect,double distance=3pt]  (1.2,0) -- (1.8,0); 
    \node (n) at (1.5,.55) {\footnotesize$1$};
    \node (m21) at (2,-.4) {$I_3$};
    \node (m22) at (2,.35) {$j_1$};
    \draw[fill=black] (3,0) circle (4pt) ;
    \draw (2.2,0) -- (3,0);
    \node (m31) at (3,-.4) {$I_4$};
    \node (m32) at (3,.35) {$j_2$};
  \end{tikzpicture}
+\begin{tikzpicture}
    \useasboundingbox (-.3,-0.1) rectangle (3.4,.4);
  \draw (0,0) circle (4pt) ;
  \node (l1) at (0,0.35) {$0$};
  \node (l2) at (0,-0.4) {$I_1$};
  \draw[fill=black] (1,0) circle (4pt) ;
      \draw (0.2,0) -- (1,0);
    \node (m1) at (1,-.4) {$I_2$};
    \node (m2) at (1,.35) {$j$};
    \draw[fill=black] (2,0) circle (4pt) ;
    \draw[->,cap=rect,double distance=3pt]  (1.2,0) -- (1.8,0); 
    \node (n) at (1.5,.55) {\footnotesize$1$};
    \node (m21) at (2,-.4) {$I_3$};
    \node (m22) at (2,.35) {$j$};
    \draw  (2.9,-.12) rectangle (3.14,.12) ;
    \draw[decorate, decoration={snake}] (2.2,0) -- (2.9,0);
    \node (m31) at (3,-.4) {$u_k$};
  \end{tikzpicture}
+\begin{tikzpicture}
    \useasboundingbox (-.3,-0.1) rectangle (3.4,.4);
  \draw (0,0) circle (4pt) ;
  \node (l1) at (0,0.35) {$0$};
  \node (l2) at (0,-0.4) {$I_1$};
  \draw[fill=black] (1,0) circle (4pt) ;
      \draw (0.2,0) -- (1,0);
    \node (m1) at (1,-.4) {$I_2$};
    \node (m2) at (1,.35) {$j$};
    \draw[fill=black] (2,0) circle (4pt) ;
    \draw[->,cap=rect,double distance=3pt]  (1.2,0) -- (1.8,0); 
    \node (n) at (1.5,.55) {\footnotesize$1$};
    \node (m21) at (2,-.4) {$I_3$};
    \node (m22) at (2,.35) {$j$};
    \draw[->,cap=rect,double distance=3pt]  (2.2,0) -- (2.8,0); 
    \node (n1) at (2.5,.55) {\footnotesize$1$};
    \draw[fill=black] (3,0) circle (4pt) ;
   \node (m21) at (3,-.4) {$I_4$};
    \node (m22) at (3,.35) {$j$};
  \end{tikzpicture}
 \\[3ex]
  &+\begin{tikzpicture}
    \useasboundingbox (-.3,-0.1) rectangle (3.4,.4);
  \draw (0,0) circle (4pt) ;
  \node (l1) at (0,0.35) {$0$};
  \node (l2) at (0,-0.4) {$I_1$};
  \draw[fill=black] (1,0) circle (4pt) ;
      \draw (0.2,0) -- (1,0);
    \node (m1) at (1,-.4) {$I_2$};
    \node (m2) at (1,.35) {$j_1$};
    \draw[fill=black] (2,0) circle (4pt) ;
    \draw[cap=rect,double distance=3pt]  (1.2,0) -- (1.8,0); 
    \node (m21) at (2,-.4) {$I_3$};
    \node (m22) at (2,.35) {$j_1$};
    \draw[fill=black] (3,0) circle (4pt) ;
    \draw[-{To[length=10]}] (2.2,0) -- (3,0);
    \node (n) at (2.5,.55) {\footnotesize$1$};
    \node (m31) at (3,-.4) {$I_4$};
    \node (m32) at (3,.35) {$j_2$};
  \end{tikzpicture}
+\begin{tikzpicture}
    \useasboundingbox (-.3,-0.1) rectangle (3.4,.4);
  \draw (0,0) circle (4pt) ;
  \node (l1) at (0,0.35) {$0$};
  \node (l2) at (0,-0.4) {$I_1$};
  \draw[fill=black] (1,0) circle (4pt) ;
      \draw (0.2,0) -- (1,0);
    \node (m1) at (1,-.4) {$I_2$};
    \node (m2) at (1,.35) {$j$};
    \draw[fill=black] (2,0) circle (4pt) ;
    \draw[cap=rect,double distance=3pt]  (1.2,0) -- (1.8,0); 
    \node (m21) at (2,-.4) {$I_3$};
    \node (m22) at (2,.35) {$j$};
    \draw  (2.9,-.12) rectangle (3.14,.12) ;
    \draw[-{To[length=10]},decorate, decoration={snake}] (2.2,0) -- (2.9,0);
    \node (n) at (2.5,.55) {\footnotesize$1$};
    \node (m31) at (3,-.4) {$u_k$};
  \end{tikzpicture}
+\begin{tikzpicture}
    \useasboundingbox (-.3,-0.1) rectangle (3.4,.4);
  \draw (0,0) circle (4pt) ;
  \node (l1) at (0,0.35) {$0$};
  \node (l2) at (0,-0.4) {$I_1$};
  \draw[fill=black] (1,0) circle (4pt) ;
      \draw (0.2,0) -- (1,0);
    \node (m1) at (1,-.4) {$I_2$};
    \node (m2) at (1,.35) {$j$};
    \draw[fill=black] (2,0) circle (4pt) ;
    \draw[cap=rect,double distance=3pt]  (1.2,0) -- (1.8,0); 
    \node (m21) at (2,-.4) {$I_3$};
    \node (m22) at (2,.35) {$j$};
    \draw[->,cap=rect,double distance=3pt]  (2.2,0) -- (2.8,0); 
    \node (n1) at (2.5,.55) {\footnotesize$2$};
    \draw[fill=black] (3,0) circle (4pt) ;
   \node (m21) at (3,-.4) {$I_4$};
    \node (m22) at (3,.35) {$j$};
  \end{tikzpicture}
\\[3ex]  
&
+ \begin{tikzpicture}
    \useasboundingbox (-.3,-0.1) rectangle (3.4,.4);
    \draw (0,0) circle (4pt) ;
    \node (l1) at (0,0.35) {$0$};
    \node (l2) at (0,-0.4) {$I_1$};
    \draw[decorate, decoration={snake}] (0.2,0) -- (0.95,0);
    \draw (.98,-.12) rectangle (1.22,.12);
    \node (m1) at (1.1,-.37) {$u_k$};
    \draw[fill=black] (2,0) circle (4pt) ;
    \draw (1.25,0) -- (2,0);
    \node (m21) at (2,-.4) {$I_2$};
    \node (m22) at (2,.35) {$j_1$};
    \draw[fill=black] (3,0) circle (4pt) ;
    \draw (2.2,0) -- (3,0);
    \node (m31) at (3,-.4) {$I_3$};
    \node (m32) at (3,.35) {$j_2$};
  \end{tikzpicture}
+ \begin{tikzpicture}
    \useasboundingbox (-.3,-0.1) rectangle (3.4,.4);
    \draw (0,0) circle (4pt) ;
    \node (l1) at (0,0.35) {$0$};
    \node (l2) at (0,-0.4) {$I_1$};
    \draw[decorate, decoration={snake}] (0.2,0) -- (0.95,0);
    \draw (.98,-.12) rectangle (1.22,.12);
    \node (m1) at (1.1,-.37) {$u_k$};
    \draw[fill=black] (2,0) circle (4pt) ;
    \draw (1.25,0) -- (2,0);
    \node (m21) at (2,-.4) {$I_2$};
    \node (m22) at (2,.35) {$j$};
    \draw (2.9,-.12) rectangle (3.14,.12) ;
    \draw[decorate, decoration={snake}] (2.2,0) -- (2.9,0);
    \node (m31) at (3,-.4) {$u_l$};
  \end{tikzpicture}
+ \begin{tikzpicture}
    \useasboundingbox (-.3,-0.1) rectangle (3.4,.4);
    \draw (0,0) circle (4pt) ;
    \node (l1) at (0,0.35) {$0$};
    \node (l2) at (0,-0.4) {$I_1$};
    \draw[decorate, decoration={snake}] (0.2,0) -- (0.95,0);
    \draw (.98,-.12) rectangle (1.22,.12);
    \node (m1) at (1.1,-.37) {$u_k$};
    \draw[fill=black] (2,0) circle (4pt) ;
    \draw (1.25,0) -- (2,0);
    \node (m21) at (2,-.4) {$I_2$};
    \node (m22) at (2,.35) {$j$};
    \draw[->,cap=rect,double distance=3pt]  (2.2,0) -- (2.8,0); 
    \node (n1) at (2.5,.55) {\footnotesize$1$};
    \draw[fill=black] (3,0) circle (4pt) ;
   \node (m31) at (3,-.4) {$I_3$};
   \node (m32) at (3,.35) {$j$};
  \end{tikzpicture}
  \\[2.5ex]
  &+ \begin{tikzpicture}
    \useasboundingbox (-.3,-0.1) rectangle (3.4,.4);
    \draw (0,0) circle (4pt) ;
    \node (l1) at (0,0.35) {$0$};
    \node (l2) at (0,-0.4) {$I_1$};
    \draw[decorate, decoration={snake}] (0.2,0) -- (0.95,0);
    \draw (.98,-.12) rectangle (1.22,.12) ;
    \node (m1) at (1.1,-.37) {$u_k$};
    \draw[decorate, decoration={snake}] (1.2,0) -- (1.95,0);
    \draw (1.98,-.12) rectangle (2.22,.12) ;
    \node (m1) at (2.1,-.37) {$u_l$};
    \draw[fill=black] (3,0) circle (4pt) ;
    \draw (2.2,0) -- (3,0);
    \node (m31) at (3,-.4) {$I_2$};
    \node (m32) at (3,.35) {$j$};
  \end{tikzpicture}
+ \begin{tikzpicture}
    \useasboundingbox (-.3,-0.1) rectangle (3.4,.4);
    \draw (0,0) circle (4pt) ;
    \node (l1) at (0,0.35) {$0$};
    \node (l2) at (0,-0.4) {$I_1$};
    \draw[decorate, decoration={snake}] (0.2,0) -- (0.95,0);
    \draw (.98,-.12) rectangle (1.22,.12) ;
    \node (m1) at (1.1,-.37) {$u_k$};
    \draw[decorate, decoration={snake}] (1.2,0) -- (1.95,0);
    \draw (1.98,-.12) rectangle (2.22,.12) ;
    \node (m1) at (2.1,-.37) {$u_l$};
    \draw (2.9,-.12) rectangle (3.14,.12) ;
    \draw[decorate, decoration={snake}] (2.2,0) -- (2.9,0);
    \node (m31) at (3,-.4) {$u_m$};
  \end{tikzpicture}
+ \begin{tikzpicture}
    \useasboundingbox (-.3,-0.1) rectangle (3.4,.4);
    \draw (0,0) circle (4pt) ;
    \node (l1) at (0,0.35) {$0$};
    \node (l2) at (0,-0.4) {$I_1$};
    \draw[decorate, decoration={snake}] (0.2,0) -- (0.95,0);
    \draw (.98,-.12) rectangle (1.22,.12) ;
    \node (m1) at (1.1,-.37) {$u_k$};
    \draw[decorate, decoration={snake}] (1.2,0) -- (1.95,0);
    \draw (1.98,-.12) rectangle (2.22,.12) ;
    \node (m1) at (2.1,-.37) {$u_l$};
\draw[decorate, decoration={zigzag}] (2.2,0) -- (2.85,0);
   \draw (3.1,0) node[star,draw]  {\begin{picture}(1,1)\put(-5,-2){\mbox{\small$u_l$}}\end{picture}} ;
   \node (m1) at (3.1,-.45) {$I_2$};
  \end{tikzpicture}
  \\[2ex]
  &+ \begin{tikzpicture}
    \useasboundingbox (-.3,-0.1) rectangle (3.4,.4);
    \draw (0,0) circle (4pt) ;
    \node (l1) at (0,0.35) {$0$};
    \node (l2) at (0,-0.4) {$I_1$};
    \draw[decorate, decoration={snake}] (0.2,0) -- (0.95,0);
    \draw (.98,-.12) rectangle (1.22,.12) ;
    \node (m1) at (1.1,-.37) {$u_k$};
    \draw[decorate, decoration={zigzag}] (1.2,0) -- (1.9,0);
    \draw (2.1,0) node[star,draw]  {\begin{picture}(1,1)\put(-5,-2){\mbox{\small$u_k$}}\end{picture}} ;
    \node (m1) at (2.1,-.45) {$I_2$};
    \draw[fill=black] (3,0) circle (4pt) ;
    \draw (2.35,0) -- (3,0);
    \node (m31) at (3,-.4) {$I_3$};
    \node (m32) at (3,.35) {$j$};
  \end{tikzpicture}
+ \begin{tikzpicture}
    \useasboundingbox (-.3,-0.1) rectangle (3.4,.4);
    \draw (0,0) circle (4pt) ;
    \node (l1) at (0,0.35) {$0$};
    \node (l2) at (0,-0.4) {$I_1$};
    \draw[decorate, decoration={snake}] (0.2,0) -- (0.95,0);
    \draw (.98,-.12) rectangle (1.22,.12) ;
    \node (m1) at (1.1,-.37) {$u_k$};
    \draw[decorate, decoration={zigzag}] (1.2,0) -- (1.9,0);
    \draw (2.1,0) node[star,draw]  {\begin{picture}(1,1)\put(-5,-2){\mbox{\small$u_k$}}\end{picture}} ;
    \node (m1) at (2.1,-.45) {$I_2$};
    \draw[decorate, decoration={snake}] (2.35,0) -- (2.9,0);
    \draw (2.9,-.12) rectangle (3.14,.12) ;
    \node (m31) at (3,-.4) {$u_l$};
  \end{tikzpicture}
+ \begin{tikzpicture}
    \useasboundingbox (-.3,-0.1) rectangle (3.4,.4);
    \draw (0,0) circle (4pt) ;
    \node (l1) at (0,0.35) {$0$};
    \node (l2) at (0,-0.4) {$I_1$};
    \draw[decorate, decoration={snake}] (0.2,0) -- (0.95,0);
    \draw (.98,-.12) rectangle (1.22,.12) ;
    \node (m1) at (1.1,-.37) {$u_k$};
    \draw[decorate, decoration={zigzag}] (1.2,0) -- (1.9,0);
    \draw (2.1,0) node[star,draw]  {\begin{picture}(1,1)\put(-5,-2){\mbox{\small$u_k$}}\end{picture}} ;
    \node (m2) at (2.1,-.45) {$I_2$};
    \draw[decorate, decoration={zigzag}] (2.35,0) -- (2.9,0);
    \draw (3.1,0) node[star,draw]  {\begin{picture}(1,1)\put(-5,-2){\mbox{\small$u_k$}}\end{picture}} ;
    \node (m3) at (3.1,-.45) {$I_3$};
  \end{tikzpicture}\rule[-5mm]{0pt}{4mm}
  \\
  &+\text{chains with f1, f2$^p$,f3}\;.
\end{align*}

\subsection{Cancellations between chains}
\label{sec:cancellations}

The following tuples will occur in the subsequent combinatorial description:
\begin{definition}[\cite{DeJong}]
\label{def:Catalantuple}
A Catalan tuple $\tilde{\underline{n}}=(n_{0},\ldots,n_{k})$ of length
$k\in\mathbb{N}$ is a tuple of integers $n_{j}\geq 0$ for
$j=0,\ldots,k$, such that
\begin{equation*}
	\sum_{j=0}^{k}n_{j}=k\qquad \text{ and }\qquad
	\sum_{j=0}^{l}n_{j}>l\quad\text{for }l=0,\ldots,k-1\;.
\end{equation*}
The set of Catalan tuples of length
$|\tilde{\underline{n}}|:=k$ is denoted by $\mathcal{C}_{k}$.
\end{definition}  
\noindent
The cardinality of $\mathcal{C}_{k}$ is the $k^\text{th}$
Catalan number \href{https://oeis.org/A000108}{OEIS A000108}.

Now, it will be convenient to collect subchains of consecutive vertices v1
with the same upper label $j$:
\begin{definition}
A \underline{\textup{v1}-block} is a subchain
\[
 \begin{tikzpicture}
  \useasboundingbox (-.3,-0.1) rectangle (1,.7);
  \draw (0,0) node[regular polygon,regular polygon sides=6,draw,fill=black] {};
  \node (j) at (0,.5) {\footnotesize$j,(n_1,n_2,...,n_s)$};
  \node (I) at (0,-.4) {\footnotesize$(I_0,I_1,...,I_s)$};
\end{tikzpicture}
    =
  \begin{tikzpicture}
  \useasboundingbox (-.3,-0.1) rectangle (4.4,.7);
  \draw[fill=black] (0,0) circle (4pt) ;
  \node (l1) at (0,0.35) {$j$};
  \node (l2) at (0,-0.4) {$I_0$};
  \draw[->,cap=rect,double distance=3pt]  (0.2,0) -- (0.8,0);
   \node (n0) at (0.5,.55) {\footnotesize$n_1$};
  \draw[fill=black] (1,0) circle (4pt) ;
    \node (m1) at (1,-.4) {$I_1$};
    \node (m2) at (1,.35) {$j$};
    \draw[fill=black] (2,0) circle (4pt) ;
    \draw[->,cap=rect,double distance=3pt]  (1.2,0) -- (1.8,0); 
    \node (n) at (1.5,.55) {\footnotesize$n_2$};
    \node (m21) at (2,-.4) {$I_2$};
    \node (m22) at (2,.35) {$j$};
    \node (z) at (2.7,0) {$\dots$};
    \draw[->,double distance=3pt]  (3.2,0) -- (3.85,0); 
    \node (n1) at (3.5,.55) {\footnotesize$n_s$};
    \draw[fill=black] (4,0) circle (4pt) ;
   \node (m21) at (4,-.4) {$I_s$};
    \node (m22) at (4,.35) {$j$};
  \end{tikzpicture}\rule[-4mm]{0mm}{4mm}
\]
of vertices \textup{v1} of the same label $j$,
connected by edges \textup{e6}$^{n_i}$.  We call $j$ the
\underline{label}, $\underline{n}=(n_1,...,n_s)$ the
\underline{degree} and $\underline{I}=(I_0,I_1,...,I_s)$ the
\underline{partition distribution} of the block.  Moreover, we let
$s=s(\underline{n})$ be the \underline{size},
$|\underline{n}|=n_1+...+n_s$ be the \underline{length} and
$s(\underline{n})-|\underline{n}|$ be the \underline{deficit} of the
\textup{v1}-block.  We also regard vertices \textup{v1} as
\textup{v1}-blocks of size or length $0$.
\end{definition}
\noindent
A v1-block can terminate a chain iff the deficit is $0$.  A v1-block
can be followed by an edge e1$^p$ or e3$^p$; the label $p$ of such an
edge is then given by the deficit $p=s(\underline{n})-|\underline{n}|$
of the v1-block before it. Since a v1-block is formed by repeatedly
attaching a function f2$^p$ labelled $p\geq 0$, the condition on the
deficit must hold at all intermediate steps. This amounts to a
condition $\sum_{i=1}^r n_i \leq r$ on any partial sum. For blocks of
total deficit $0$ (those which terminate a chain or are followed by
edges e1$^0$ or e3$^0$, this is equivalent to the opposite condition
$\sum_{i=0}^r n_{s-i} > r$ for $0\leq r \leq s-1$ and
$\sum_{i=0}^{s} n_{s-i} = s$ when prepending $n_0=0$. This means that
the reversely ordered tuple
$\underline{\tilde{n}}:=(n_s,n_{s-1}, ...,n_1,0)$ is a \emph{Catalan
  tuple}. We consider the subset of chains
which differ only in the degrees $\underline{n}$ of a v1-block of size
$s$, but otherwise have identically labelled vertices. In this subset
any degree $\underline{n}$ of the v1-block compatible with the deficit
condition is produced, and precisely once.
\begin{definition}
A  \underline{\textup{v2}-block} is a subchain
\[
 \begin{tikzpicture}
  \useasboundingbox (-.3,-0.1) rectangle (1,.3);
  \draw (0,0) node[regular polygon,regular polygon sides=5,draw] {};
  \node (I) at (0,-.4) {\footnotesize$(I_1,...,I_s);u$};
\end{tikzpicture}
    =
\begin{tikzpicture}
    \useasboundingbox (-.3,-0.1) rectangle (4.4,.4);
    \draw (-.12,-.12) rectangle (.12,.12) ;
    \node (m1) at (0,-.37) {$u$};
    \draw[decorate, decoration={zigzag}] (0.15,0) -- (0.75,0);    
    \draw (1,0) node[star,draw]  {\begin{picture}(1,1)
        \put(-3,-2){\mbox{\small$u$}}\end{picture}} ;
    \node (l1) at (1,-0.45) {$I_1$};
    \draw[decorate, decoration={zigzag}] (1.25,0) -- (1.75,0);
    \draw (2,0) node[star,draw]  {\begin{picture}(1,1)
        \put(-3,-2){\mbox{\small$u$}}\end{picture}} ;
    \node (l2) at (2,-0.45) {$I_2$};
    \node (z) at (2.7,0) {$\dots$};
    \draw[decorate, decoration={zigzag}] (3.25,0) -- (3.75,0);
    \draw (4,0) node[star,draw]  {\begin{picture}(1,1)
        \put(-3,-2){\mbox{\small$u$}}\end{picture}} ;
    \node (ms) at (4,-.45) {$I_s$};
  \end{tikzpicture}\rule[-5mm]{0mm}{4mm}
\]
starting with a vertex \textup{v2} of lower label $u$ and several consecutive
vertices \textup{v3} with the same inner label $u$, connected by edges
\textup{e5}.  We let $u$ be the \underline{label}, $s$ be the
\underline{size} and $(I_1,...,I_s)$ be the \underline{partition
  distribution} of a \textup{v2}-block.  A \textup{v2}-block of size
$0$ is identified with a vertex \textup{v2} with lower label $u$.
\end{definition}
\noindent
If several v2-blocks arise in a chain then its labels $u_k,u_l,\dots$
are necessarily different.

We will prove that, after taking cancellations into account, also the
labels $j_1,j_2,\dots$ of v1-blocks in the surviving chains are
pairwise different. These cancellations start with chains of 4
vertices:
\begin{align*}
\begin{tikzpicture}
    \useasboundingbox (-.3,-0.1) rectangle (3.4,.4);
  \draw (0,0) circle (4pt) ;
  \node (l1) at (0,0.35) {$0$};
  \node (l2) at (0,-0.4) {$I_1$};
  \draw[fill=black] (1,0) circle (4pt) ;
      \draw (0.2,0) -- (1,0);
    \node (m1) at (1,-.4) {$I_2$};
    \node (m2) at (1,.35) {$j_1$};
    \draw[fill=black] (2,0) circle (4pt) ;
    \draw (1.2,0) -- (2,0);
    \node (m21) at (2,-.4) {$I_3$};
    \node (m22) at (2,.35) {$j_2$};
    \draw[fill=black] (3,0) circle (4pt) ;
    \draw (2.2,0) -- (3,0);
    \node (m31) at (3,-.4) {$I_4$};
    \node (m32) at (3,.35) {$j_1$};
  \end{tikzpicture}
  +
\begin{tikzpicture}
    \useasboundingbox (-.3,-0.1) rectangle (3.4,.4);
  \draw (0,0) circle (4pt) ;
  \node (l1) at (0,0.35) {$0$};
  \node (l2) at (0,-0.4) {$I_1$};
  \draw[fill=black] (1,0) circle (4pt) ;
      \draw (0.15,0) -- (1,0);
    \node (m1) at (1,-.4) {$I_2$};
    \node (m2) at (1,.35) {$j_1$};
    \draw[fill=black] (2,0) circle (4pt) ;
    \draw[cap=rect,double distance=3pt]  (1.2,0) -- (1.8,0); 
    \node (m21) at (2,-.4) {$I_3$};
    \node (m22) at (2,.35) {$j_1$};
    \draw[fill=black] (3,0) circle (4pt) ;
    \draw[-{To[length=10]}] (2.15,0) -- (3,0);
    \node (n) at (2.5,.55) {\footnotesize$1$};
    \node (m31) at (3,-.4) {$I_4$};
    \node (m32) at (3,.35) {$j_2$};
  \end{tikzpicture}
&=0\;,
\nonumber
\\[2ex]
\begin{tikzpicture}
    \useasboundingbox (-.3,-0.1) rectangle (3.4,.4);
  \draw (0,0) circle (4pt) ;
  \node (l1) at (0,0.35) {$0$};
  \node (l2) at (0,-0.4) {$I_1$};
    \draw (0.15,0) -- (1,0);
    \draw[fill=black] (1,0) circle (4pt) ;
    \node (m1) at (1,-.4) {$I_2$};
    \node (m2) at (1,.35) {$j$};
    \draw[decorate, decoration={snake}] (1.15,0) -- (1.85,0);
    \draw (1.88,-.12) rectangle (2.12,.12) ;
    \node (m21) at (2,-.37) {$u_k$};
    \draw[fill=black] (3,0) circle (4pt) ;
    \draw (2.15,0) -- (3,0);
    \node (m31) at (3,-.4) {$I_3$};
    \node (m32) at (3,.35) {$j$};
  \end{tikzpicture}
  +  \begin{tikzpicture}
    \useasboundingbox (-.3,-0.1) rectangle (3.4,.4);
  \draw (0,0) circle (4pt) ;
  \node (l1) at (0,0.35) {$0$};
  \node (l2) at (0,-0.4) {$I_1$};
  \draw[fill=black] (1,0) circle (4pt) ;
      \draw (0.15,0) -- (1,0);
    \node (m1) at (1,-.4) {$I_2$};
    \node (m2) at (1,.35) {$j$};
    \draw[fill=black] (2,0) circle (4pt) ;
    \draw[cap=rect,double distance=3pt]  (1.2,0) -- (1.8,0); 
    \node (m21) at (2,-.4) {$I_3$};
    \node (m22) at (2,.35) {$j$};
    \draw  (2.9,-.12) rectangle (3.14,.12) ;
    \draw[-{To[length=10]},decorate, decoration={snake}] (2.15,0) -- (2.9,0);
    \node (n) at (2.5,.55) {\footnotesize$1$};
    \node (m31) at (3,-.4) {$u_k$};
  \end{tikzpicture}\rule[-5mm]{0mm}{4mm}
  &=0\;,
\end{align*}
which follows from the weights in Table~\ref{tab1} and with
$\nabla^0\varpi_{0,|I_3|+1}(I_3;-\hat{q}^j)
=\frac{\varpi_{0,|I_3|+1}(I_3;-\hat{q}^j)}{(-dR(\hat{q}^j))}$.
These identities reduce the set of graphs to a much simpler subset:
\begin{lemma} \label{lemma-cancel-noCatalan}
  Let $M$ be the set of
  chains generated by the loop equations for
  $\sum_{I_1\uplus I_2=I} \varpi_{0,|I_1|+1}(I_1;q) \mathfrak{v}_{0,|I_2|}(I_2\|q)$.
  Then cancellations between weights
  remove all chains with edges \textup{e1}$^p$ and \textup{e3}$^p$
  having a tip labelled $p\geq 1$ and all chains with two or more
  identically labelled \textup{v1}-blocks. The subset of surviving
  graphs is given by the set of chains made of \textup{v2}-blocks and
  of \textup{v1}-blocks which have deficit 0 and pairwise different
  labels, connected by appropriate edges without tip.
\begin{proof}
  Consider a v1-block of label $j$, partition distribution
  $\underline{I}$ and degree $\underline{n}=(n_1,n_2,...,n_s)$ with
  deficit $p=s-n_1-...-n_s\geq 1$. Its reverse degree
  $(n_s,n_{s-1},...,n_1,0)$ cannot be a Catalan tuple  for $p\geq
  1$. This means that either $n_s=0$, or there is a unique
  $2\leq t\leq s$ such that $(n_s,n_{s-1},...,n_t,0)$ is a Catalan
  tuple but $(n_s,n_{s-1},...,n_t,n_{t-1},0)$ is not. This necessarily
  means $n_{t-1}=0$. We define a unique splitting into two v1-blocks
  of degrees $\underline{n}^{\!-}$ and $\underline{n}^{\!+}$:
  \\[\medskipamount]
  \hspace*{\fill}\begin{tabular}{rr@{\,}l} $n_s=0$: & Set
    $\underline{n}^{\!-} $ &$=(n_1,\dots,n_{s-1})$,
    $\underline{I}^{-}=(I_0,\dots,I_{s-1})$,
    $\underline{n}^{\!+}=\emptyset$, $\underline{I}^{+}=(I_s)$.
                   \\[\medskipamount]
                   $n_s\neq 0$: & Set
                                  $\underline{n}^{\!-}$
                           &$=(n_1,\dots,n_{t-2})$,
$\underline{I}^{-}=(I_0,\dots,I_{t-2})$,
                   \\
&         $\underline{n}^{\!+}$& $=(n_t,\dots,n_s)$,
$\underline{I}^{+}=(I_{t-1},I_t,...,I_s)$.
\end{tabular}
\\[\medskipamount]
By construction, $\underline{n}^{\!+}$ has deficit 0 so that it can
terminate a chain or is followed by edges e1$^0$ or e3$^0$. The other
label $\underline{n}^{\!-}$ has deficit $p-1$ and is followed by edges
e1$^{p{-}1}$ or e3$^{p{-}1}$. Conversely, two degrees
$\underline{n}^{\!-}$ of deficit $p-1\geq 0$ and $\underline{n}^{\!+}$
of deficit $0$ can be joined to a unique degree $\underline{n}$ of
deficit $p$.

The weights given in Table~\ref{tab1} together with
$\nabla^{n_{t-1}}\varpi_{0,|I_{t-1}|+1}(I_{t-1};-\hat{q}^j)
= \frac{\varpi_{0,|I_{t-1}|+1}(I_{t-1};-\hat{q}^j)}{(-dR(\hat{q}^j))}$
confirm the following identity:
\usetikzlibrary{patterns}
\begin{align}
0=\begin{tikzpicture}
    \useasboundingbox (-.5,-0.1) rectangle (3.6,.4);
  \draw[pattern=north east lines] (0,0) circle (8pt) ;
  \draw (0.3,0) -- (1.2,0);
  \draw (1.2,0) node[regular polygon,regular polygon sides=6,fill=black,draw] {};
    \node (m1) at (1.2,-.45) {\footnotesize$\underline{I}$};
    \node (m2) at (1.2,.4) {\footnotesize$j,\underline{n}$};
    \draw[-{To[length=10]}] (1.4,0) -- (2.8,0);  
    \node (n) at (2.3,.55) {\footnotesize$p$};
    \draw (3,0) node[regular polygon,regular polygon sides=6,fill=black,draw] {};
    \node (m21) at (3,-.45) {\footnotesize$\underline{I}_1$};
    \node (m22) at (3.2,.4) {\footnotesize$j_1,\underline{n}_1$};
  \end{tikzpicture}
  +
\begin{tikzpicture}
    \useasboundingbox (-.5,-0.1) rectangle (5,.4);
    \draw[pattern=north east lines] (0,0) circle (8pt) ;
    \draw (0.3,0) -- (1.2,0);
    \draw (1.2,0) node[regular polygon,regular polygon sides=6,fill=black,draw] {};
    \node (m1) at (1.2,-.45) {\footnotesize$\underline{I}^{-}$};
    \node (m2) at (1.2,.4) {\footnotesize$j,\underline{n}^{-}$};
    \draw[-{To[length=10]}] (1.4,0) -- (2.8,0);  
    \node (n) at (2.2,.55) {\footnotesize$p{-}1$};
    \draw (3,0) node[regular polygon,regular polygon sides=6,fill=black,draw] {};
    \node (m21) at (3,-.45) {\footnotesize$\underline{I}_1$};
    \node (m22) at (3.2,.4) {\footnotesize$j_1,\underline{n}_1$};
    \draw (3.2,0) -- (4.5,0);
    \draw (4.5,0) node[regular polygon,regular polygon sides=6,fill=black,draw] {};
    \node (m31) at (4.5,-.45) {\footnotesize$\underline{I}^{+}$};
    \node (m32) at (4.5,.4) {\footnotesize$j,\underline{n}^{\!+}$};
  \end{tikzpicture} \rule[-5mm]{0mm}{4mm} 
\;,
\label{chain-12}
\end{align}
where the shaded circle stands for any identical subchain in both chains.
The same cancellation arises if the v1-block labelled $j_1$ is replaced
by a v2-block and e1$^p$ by e3$^p$.

Next for chains which extend by further blocks to the right,
all with labels $\neq j$, we have
\begin{align}
0&=\begin{tikzpicture}
    \useasboundingbox (-.5,-0.1) rectangle (3.8,.4);
  \draw[pattern=north east lines] (0,0) circle (8pt) ;
  \draw (0.3,0) -- (1.2,0);
  \draw (1.2,0) node[regular polygon,regular polygon sides=6,fill=black,draw] {};
    \node (m1) at (1.2,-.45) {\footnotesize$\underline{I}$};
    \node (m2) at (1.2,.4) {\footnotesize$j,\underline{n}$};
    \draw[-{To[length=10]}] (1.4,0) -- (2.8,0);  
    \node (n) at (2.3,.55) {\footnotesize$p$};
    \draw (3,0) node[regular polygon,regular polygon sides=6,fill=black,draw] {};
    \node (m21) at (3,-.45) {\footnotesize$\underline{I}_1$};
    \node (m22) at (3.2,.4) {\footnotesize$j_1,\underline{n}_1$};
    \draw (3,0) -- (4.3,0) ;
      \draw (4.5,0) node[regular polygon,regular polygon sides=6,fill=black,draw] {};
    \node (m31) at (4.5,-.45) {\footnotesize$\underline{I}_3$};
    \node (m32) at (4.5,.4) {\footnotesize$j_2,\underline{n}_2$};
    \node (d) at (5.3,0) {\dots};
      \draw (6,0) node[regular polygon,regular polygon sides=6,fill=black,draw] {};
    \node (m41) at (6,-.45) {\footnotesize$\underline{I}_{r-1}$};
    \node (m42) at (6,.4) {\footnotesize$j_{r-1},\underline{n}_{r-1}$};
    \draw (6,0) -- (7.3,0) ;
    \draw (7.5,0) node[regular polygon,regular polygon sides=6,fill=black,draw] {};
    \node (m41) at (7.5,-.45) {\footnotesize$\underline{I}_{r}$};
    \node (m42) at (7.5,.4) {\footnotesize$j_r,\underline{n}_{r}$};
  \end{tikzpicture}
  \nonumber
  \\[3ex]
  &+ \sum_{l=1}^{r}
\begin{tikzpicture}
    \useasboundingbox (-.3,-0.1) rectangle (11,.4);
  \draw[pattern=north east lines] (0,0) circle (8pt) ;
  \draw (0.3,0) -- (1.2,0);
  \draw (1.2,0) node[regular polygon,regular polygon sides=6,fill=black,draw] {};
    \node (m1) at (1.2,-.45) {\footnotesize$\underline{I}^{-}$};
    \node (m2) at (1.2,.4) {\footnotesize$j,\underline{n}^{\!-}$};
    \draw[-{To[length=10]}] (1.4,0) -- (2.8,0);  
    \node (n) at (2.2,.55) {\footnotesize$p{-}1$};
    \draw (3,0) node[regular polygon,regular polygon sides=6,fill=black,draw] {};
    \node (m21) at (3,-.45) {\footnotesize$\underline{I}_1$};
    \node (m22) at (3.2,.4) {\footnotesize$j_1,\underline{n}_1$};
    \draw (3.2,0) -- (4.5,0);
  \draw (4.5,0) node[regular polygon,regular polygon sides=6,fill=black,draw] {};
    \node (m31) at (4.5,-.45) {\footnotesize$\underline{I}_2$};
    \node (m32) at (4.5,.4) {\footnotesize$j_2,\underline{n}_2$};
    \node (d) at (5.3,0) {\dots};
      \draw (6,0) node[regular polygon,regular polygon sides=6,fill=black,draw] {};
    \node (m41) at (6,-.45) {\footnotesize$\underline{I}_{l}$};
    \node (m42) at (6,.4) {\footnotesize$j_{l},\underline{n}_{l}$};
    \draw (6,0) -- (7.3,0) ;
    \draw (7.5,0) node[regular polygon,regular polygon sides=6,fill=black,draw] {};
    \node (m41) at (7.5,-.45) {\footnotesize$\underline{I}^+$};
    \node (m42) at (7.5,.4) {\footnotesize$j,\underline{n}^{\!+}$};
    \draw (7.5,0) -- (8.8,0) ;
    \draw (9,0) node[regular polygon,regular polygon sides=6,fill=black,draw] {};
    \node (m51) at (9,-.45) {\footnotesize$\underline{I}_{l+1}$};
    \node (m52) at (9,.4) {\footnotesize$j_{l+1},\underline{n}_{l+1}$};
    \node (d) at (9.8,0) {\dots};
    \draw (10.5,0) node[regular polygon,regular polygon sides=6,fill=black,draw] {};
    \node (m61) at (10.5,-.45) {\footnotesize$\underline{I}_{r}$};
    \node (m62) at (10.5,.4) {\footnotesize$j_{r},\underline{n}_{r}$};
  \end{tikzpicture}\rule[-5mm]{0mm}{4mm}  \;.
  \label{chain-12p}
\end{align}
Again the shaded circle stands for any
identical subchain.  The same cancellation arises if any subset of
v1-blocks (other than the one labelled $j$) is replaced by
corresponding v2-blocks.

After these preparations we prove that (\ref{chain-12}) and
(\ref{chain-12p}) provide the claimed reduction in the set of chains
describing
$\sum_{I_1\uplus I_2=I}
\varpi_{0,|I_1|+1}(I_1;q) 
\mathfrak{v}_{0,|I_2|}(I_2\|q)$.
\begin{enumerate}  
\item[(1)] We start with the type of chains indicated by the left
  graph in (\ref{chain-12}), with $p\geq 1$. Since the splitting of
  $\underline{n}$ into $\underline{n}^{\!-},\underline{n}^{\!+}$ is
  unique, it cancels against a unique chain indicated on the right of
  (\ref{chain-12}). Conversely, for any chain $K$ terminating in
  a triple consisting of two v1-blocks of the same label $j$ and any
  other block in between, there is a unique chain indicated on the
  left of (\ref{chain-12}) against which $K$ cancels. As result
  we remove all chains with a single block after the last e1$^p$ or
  e3$^p$ edge (with $p\geq 1$) and all those chains which terminate in
  a triple of blocks in which two v1-blocks are equally labelled.

\item[(2)] We pass to (\ref{chain-12p}) for $r=2$. The chain in the
  first line is only present for $j_2\neq j$ because the case $j_2= j$
  was removed in step (1). According to (\ref{chain-12p}) the chain
  $K$ indicated in the first line cancels against two uniquely
  determined chains terminating in a quadruple of blocks two of which
  are labelled $j$, and conversely. After all we remove all chains
  with two blocks after the last e1$^p$ or e3$^p$ edge (with
  $p\geq 1$) and all those chains terminating in a v1-block labelled
  $j$ which is followed by three more blocks one of them also labelled
  $j$.

\item[($r$)] Continuing in this manner removes all chains with an
  e1$^p$ or e3$^p$ edge with $p\geq 1$ and all chains with two or more
  identically labelled v1-blocks.
\end{enumerate}
We are left with chains in which all blocks have different labels and
are connected by edges e1$^0$,e3$^0$, i.e.\ without tip.
\end{proof}
\end{lemma}

All surviving v1-blocks have degrees of deficit $0$, i.e.\ are
reversals of Catalan tuples.  In the next step we collect all
v1-blocks which have the same union of their partition distribution
(and deficit $0$) to a \emph{\textup{v1}-group}:
\begin{align}
 \begin{tikzpicture}
  \useasboundingbox (-.3,-0.1) rectangle (.3,.7);
  \draw (0,0) node[regular polygon,regular polygon sides=6,draw] {};
  \node (j) at (0,.5) {\footnotesize$j$};
  \node (I) at (0,-.4) {\footnotesize$I$};
\end{tikzpicture}
  &:=\sum_{s=0}^{|I|-1} \sum_{I_0\uplus I_1\uplus ... \uplus I_s=I}
  \sum_{(n_s,....,n_1,0)\in \mathcal{C}_s} 
 \begin{tikzpicture}
  \useasboundingbox (-1,-0.1) rectangle (1,.7);
  \draw (0,0) node[regular polygon,regular polygon sides=6,draw,fill=black] {};
  \node (j) at (0,.5) {\footnotesize$j,(n_1,n_2,...,n_s)$};
  \node (I) at (0,-.4) {\footnotesize$(I_0,I_1,...,I_s)$};
\end{tikzpicture}\;,
\label{coll-v1}
\\
\text{weight}\Big(
 \begin{tikzpicture}
  \useasboundingbox (-.3,-0.1) rectangle (.3,.7);
  \draw (0,0) node[regular polygon,regular polygon sides=6,draw] {};
  \node (j) at (0,.5) {\footnotesize$j$};
  \node (I) at (0,-.4) {\footnotesize$I$};
\end{tikzpicture}
\Big)&=
(-dR(\hat{q}^j)) \sum_{s=0}^{|I|-1} \sum_{I_0\uplus I_1\uplus ... \uplus I_s=I}
  \sum_{(n_s,....,n_1,0)\in \mathcal{C}_s} 
\prod_{i=0}^s \nabla^{n_i}\varpi_{0,|I_i|+1}(I_i;-\hat{q}^j)
\;.
\nonumber
\end{align}
We have used that the leftmost vertex of every v1-block
has weight
$\varpi_{0,|I_0|+1}(I_0;-\hat{q}^j)
=(-dR(\hat{q}^j)) \nabla^{0}\varpi_{0,|I_0|+1}(I_0;-\hat{q}^j)$.

Similarly we collect v2-blocks with the same union of their partition
distribution to a \emph{\textup{v2}-group}:
\begin{align}
 \begin{tikzpicture}
  \useasboundingbox (-.3,-0.1) rectangle (.3,.4);
  \draw (0,0) node[diamond,draw] {};
  \node (I) at (0,-.4) {\footnotesize$I;u$};
\end{tikzpicture}
  &:= \begin{tikzpicture}
  \useasboundingbox (-.3,-0.1) rectangle (.2,.4);
  \draw (-.12,-.12) rectangle (.12,.12);
  \node (I) at (0,-.4) {\footnotesize$u$};
\end{tikzpicture}\delta_{\|I\|,0}
+
\sum_{s=1}^{|I|} \sum_{I_1\uplus ... \uplus I_s=I}
 \begin{tikzpicture}
  \useasboundingbox (-1,-0.1) rectangle (1,.7);
  \draw (0,0) node[regular polygon,regular polygon sides=5,draw] {};
  \node (I) at (0,-.4) {\footnotesize$(I_1,...,I_s);u$};
\end{tikzpicture}\;,
\label{coll-v2}
\\
\text{weight}\Big(
 \begin{tikzpicture}
  \useasboundingbox (-.3,-0.1) rectangle (.3,.4);
  \draw (0,0) node[diamond,draw] {};
  \node (I) at (0,-.4) {\footnotesize$I;u$};
\end{tikzpicture}
\Big)&= -\sum_{s=0}^{|I|}
\sum_{I_1\uplus ... \uplus I_s=I}^{(I\neq \emptyset)}
\frac{1}{(R(-u)-R(q))^{s+1}}
\prod_{i=1}^s \frac{\varpi_{|I_i|+1}(I_i;u)}{dR(u)} \;.
\nonumber
\end{align}
The summation $\sum_{I_1\uplus ... \uplus I_s=I}^{(I\neq \emptyset)}$ is left out for $|I| =\emptyset$. 
For $I\neq \emptyset$ there is no contribution from $s=0$.
We summarise the previous simplifications and collections:
\begin{corollary}
  \label{cor-graph-reduced}
  The integrand $\sum_{I_1\uplus I_2=I}
\varpi_{0,|I_1|+1}(I_1;q) \mathfrak{v}_{0,|I_2|}(I_2\|q)$
in the first line of \eqref{om0-fraku} is the sum of weights of all
  different chains which meet the criteria:
  \begin{itemize}
    \item The leftmost vertex is \textup{v0} with weight $-\varpi_{0,|I_1|+1}(I_1;q)$.
    \item Any other vertex is a \textup{v1}-group with weight \eqref{coll-v1}
      or a \textup{v2}-group with weight \eqref{coll-v2}.
      The labels $j_i$ of the \textup{v1}-groups are pairwise different.
    \item The union of all subsets $I_i$ at the initial vertex, the \textup{v1}-groups
      and the \textup{v2}-groups, together with the labels $u_k$ of the
      \textup{v2}-groups, is $I=\{u_1,...,u_m\}$.

    \item The edges between the groups (and initial vertex) are given
      by \textup{e1}$^0$,\textup{e2},\textup{e3}$^0$,\textup{e4}
      depending on the groups they connect. Their weights are given in
      Table~\ref{tab1}.

    \end{itemize}
\end{corollary}

\subsection{Weight of a v1-group}

\label{sec:weightv1}

Next we prove a simpler formula for the weight (\ref{coll-v1}) of a
v1-group. Its main step is Corollary (\ref{cor:leibniz}), a variant of
Corollary \ref{lemma:multinomialDolega} given in the
appendix. In the
second line of (\ref{coll-v1}) we write the sum over all
partitions as sum over ordered partitions (introduced in the
beginning of Subsection~\ref{sec:tools}) together with a sum over
permutations $\varsigma$.  Inserting the definition (\ref{def:nabla-om}) of
$\nabla \varpi$ we thus have
\begin{align*}
\text{weight}\Big(
 \begin{tikzpicture}
  \useasboundingbox (-.3,-0.1) rectangle (.3,.7);
  \draw (0,0) node[regular polygon,regular polygon sides=6,draw] {};
  \node (j) at (0,.5) {\footnotesize$j$};
  \node (I) at (0,-.4) {\footnotesize$I$};
\end{tikzpicture}
\Big)
&=
(-dR(\hat{q}^j)) \sum_{s=0}^{|I|-1} \sum_{\substack{I_0\uplus I_1\uplus ... \uplus I_s=I \\ I_0<I_1<...<I_s}}
\sum_{\varsigma \in \mathcal{S}_{s+1}} 
\sum_{(n_s,....,n_1,0)\in \mathcal{C}_s} \lim_{z\to -\hat{q}^j} \Big\{
\\*
&\qquad 
  \prod_{i=0}^s \frac{(-1)^{n_i}}{n_i!}
  \frac{\partial}{\partial (R(z))^{n_i}}
  \Big( \frac{R(z)-R(-\hat{q}^j)}{R(q)-R(-z)} 
\frac{\varpi_{0,|I_{\varsigma(i)}|+1}(I_{\varsigma(i)};z)}{dR(z)}\Big)\Big\}\;.
\end{align*}
With Corollary~\ref{cor:leibniz} and the bijection between rooted
plane trees and Catalan tuples we can replace
$\sum_{\varsigma \in \mathcal{S}_{s+1}} \sum_{(n_s,....,n_1,0)\in
  \mathcal{C}_s} \mapsto s! \sum_{n_0+...+n_s=s}$. We reexpress the
result in terms of $\nabla \varpi$ and admit again any order of
partitions of $I$ into $s+1$ subsets:
\begin{align}
&\text{weight}\Big(
 \begin{tikzpicture}
  \useasboundingbox (-.3,-0.1) rectangle (.3,.7);
  \draw (0,0) node[regular polygon,regular polygon sides=6,draw] {};
  \node (j) at (0,.5) {\footnotesize$j$};
  \node (I) at (0,-.4) {\footnotesize$I$};
\end{tikzpicture}
\Big)
\label{weight-v1-a}
\\
&=
(-dR(\hat{q}^j)) \sum_{s=0}^{|I|-1} \sum_{I_0\uplus I_1\uplus ... \uplus I_s=I}
\frac{1}{s+1}
  \sum_{n_0+...+n_s=s}  \prod_{i=0}^s \nabla^{n_i}\varpi_{0,|I_i|+1}(I_i;-\hat{q}^j)\;.
\nonumber
\end{align}

Our aim is to prove Theorem~\ref{thm:sol-dse}, namely that the
solution $\varpi_{0,|I|+1}(I;q)$ of the system \eqref{om0-fraku},
\eqref{frakuIq}, \eqref{uIzq} (for $z\mapsto u_k$) and \eqref{uIzq-d}
is, after applying the exterior differentials $d_{u_k}$ to pass from
$\varpi_{0,|I|+1}$ to $\omega_{0,|I|+1}$, the same as the
solution of (\ref{eq:flip-om}) for $x(z)=R(z)$ and
$\iota(z)=-z$. We prove this theorem by induction. The v1-group is
always a genuine subchain because at least the initial vertex v0 is
excluded. Therefore Theorem~\ref{thm:sol-dse} is the induction
hypothesis for the v1-group, which gives:
\begin{proposition}
\label{prop:weight-v1}
  The \textup{v1}-group has
  $\textup{weight}\Big(
 \begin{tikzpicture}
  \useasboundingbox (-.3,-0.1) rectangle (.3,.7);
  \draw (0,0) node[regular polygon,regular polygon sides=6,draw] {};
  \node (j) at (0,.5) {\footnotesize$j$};
  \node (I) at (0,-.4) {\footnotesize$I$};
\end{tikzpicture}
\Big)
=-\varpi_{0,|I|+1}(I;\hat{q}^j)$.
\begin{proof}
  This follows from Lemma \ref{lem:omega-nabla} for
  $q\mapsto -\hat{q}^j$ when moving the first term
  $\omega_{0,|I|+1}(I,-\hat{q}^j)=dy(-\hat{q}^j) \nabla^0
  \omega_{0,|I|+1}(I,-\hat{q}^j)$ to the rhs. Then
  $d_{u_1}\cdots d_{u_m}$ applied to (\ref{weight-v1-a}) equals
  $-\omega_{0,|I|+1}(I,\hat{q}^j)$ when taking
  Theorem~\ref{thm:sol-dse} as induction hypothesis, for
  $I=\{u_1,...,u_m\}$. Inverting the differentials $d_{u_k}$ gives the
  assertion.
\end{proof}
\end{proposition}

\subsection{Poles of \texorpdfstring{$\varpi_{0,|I|+1}(I;z)$}{\omega(I,z)} at
  \texorpdfstring{$z=\beta_i$}{z=\beta}}

We let
$\mathcal{P}^i_z\omega(z) = \Res\displaylimits_{q\to \beta_i}
\frac{\omega(q)dz}{z-q}$ be the projection of a 1-form $\omega$ to its poles at
$z=\beta_i$. Proposition~\ref{prop:om-t} gives
\begin{align*}
\mathcal{P}^i_z \varpi_{0,|I|+1}(I;z)
&= \mathcal{P}^i_z\Big(
\sum_{I_1\uplus I_2=I} \varpi_{0,|I_1|+1}(I_1;z) \mathfrak{v}_{0,|I_2|}(I_2\|z)\Big)\;.
\end{align*}
\begin{proposition}
  \label{prop:varpi-polar}
  Let $\hat{q}^{j_i}=\sigma_i(q)$ be the preimage of $q$ which
  corresponds to the local Galois involution near $\beta_i$. Then for
  all $I=\{u_1,...,u_m\}$ with $m\geq 2$ one has
\begin{align*}  
\mathcal{P}^i_z\varpi_{0,|I|+1}(I;z)
&=- \mathcal{P}^i_z\Big(
\sum_{I_1\uplus I_2=I}
\frac{\varpi_{0,|I_1|+1}(I_1;z)\varpi_{0,|I_1|+1}(I_1;\hat{z}^{j_i})}{
 dR(\hat{z}^{j_i})(R(-z)-R(-\hat{z}^{j_i}))}  \Big)\;.
\end{align*}
The application of $d_{u_1}\cdots d_{u_k}$ agrees with the restriction of \eqref{sol:omega}
to poles at $z=\beta_i$.
\begin{proof}
In the graphical representation of Corollary \ref{cor-graph-reduced}, the assertion amounts to
\begin{align}
\mathcal{P}_q^i\Big(- \begin{tikzpicture}
        \useasboundingbox (.8,-.1) rectangle (1.2,.55);
    \draw (1,0) circle (4pt);
    \node (l1) at (1,.35) {$0$};
    \node (l2) at (1,-.35) {$I$};
  \end{tikzpicture} \Big)
  =\mathcal{P}_q^i\Big(
   \begin{tikzpicture}
  \useasboundingbox (-.3,-0.1) rectangle (1.3,.7);
  \draw (0,0) circle (4pt) ;
  \node (a1) at (0,0.35) {$0$};
  \node (a2) at (0,-0.4) {$I_{1}$};  
  \draw (.15,0) -- (.8,0);  
  \draw (1,0) node[regular polygon,regular polygon sides=6,draw] {};
  \node (j) at (1,.5) {\footnotesize$j_i$};
  \node (I) at (1,-.4) {\footnotesize$I_2$};
\end{tikzpicture}\rule[-4mm]{0mm}{4mm}\Big)\;.
\label{polar-assertion}
\end{align}
The rhs is one of the chains contributing to the residue at $q=\beta_i$ in the first line
of (\ref{om0-fraku}).  We have to prove that all other chains described in
Corollary \ref{cor-graph-reduced} sum up to expressions regular at
$q=\beta_i$.

We prove this regularity by induction on the length of chains (with
v1/v2-groups as vertices). By $\bar{j}$ we denote a label different
from $j_i$.  There are two remaining chains of length $2$, namely $
\begin{tikzpicture}
  \useasboundingbox (-.3,-0.1) rectangle (1.3,.7);
  \draw (0,0) circle (4pt) ;
  \node (a1) at (0,0.35) {$0$};
  \node (a2) at (0,-0.4) {$I_{1}$};  
  \draw (.15,0) -- (.8,0);  
  \draw (1,0) node[regular polygon,regular polygon sides=6,draw] {};
  \node (j) at (1,.5) {\footnotesize$\bar{j}$};
  \node (I) at (1,-.4) {\footnotesize$I_2$};
\end{tikzpicture}
+
   \begin{tikzpicture}
  \useasboundingbox (-.3,-0.1) rectangle (1.3,.7);
  \draw (0,0) circle (4pt) ;
  \node (a1) at (0,0.35) {$0$};
  \node (a2) at (0,-0.4) {$I_{1}$};  
  \draw[decorate, decoration={snake}] (0.15,0) -- (0.75,0); 
  \draw (1,0) node[diamond,draw] {};
  \node (I) at (1,-.4) {\footnotesize$I_2;u$};
\end{tikzpicture}\rule[-2.4ex]{0mm}{4mm}$
(in the first chain summation over $\bar{j}\neq j_i$ and over
partitions $I_1\uplus I_2=I$, in the second chain summation over
partitions $I_1\uplus u \uplus I_2=I$). Edges, v1-groups with label
$\bar{j}$ and v2-groups are regular at $q=\beta_i$. The initial vertex
v0 is regular for $|I_1|=1$ so that these chains only contribute to
$\mathcal{P}^i_q$ for $|I_1|\geq 2$.  In that case we can, up to
terms holomorphic at $\beta_i$, replace the initial 
vertex by (\ref{polar-assertion})
for $I\mapsto I_1$, which is true by induction hypothesis. We thus
have
\begin{align}
  &\mathcal{P}_q^i\Big(
\begin{tikzpicture}
  \useasboundingbox (-.3,-0.1) rectangle (1.3,.7);
  \draw (0,0) circle (4pt) ;
  \node (a1) at (0,0.35) {$0$};
  \node (a2) at (0,-0.4) {$I_{1}$};  
  \draw (.15,0) -- (.8,0);  
  \draw (1,0) node[regular polygon,regular polygon sides=6,draw] {};
  \node (j) at (1,.5) {\footnotesize$\bar{j}$};
  \node (I) at (1,-.4) {\footnotesize$I_2$};
\end{tikzpicture}
+
   \begin{tikzpicture}
  \useasboundingbox (-.3,-0.1) rectangle (1.4,.7);
  \draw (0,0) circle (4pt) ;
  \node (a1) at (0,0.35) {$0$};
  \node (a2) at (0,-0.4) {$I_{1}$};  
  \draw[decorate, decoration={snake}] (0.15,0) -- (0.75,0); 
  \draw (1,0) node[diamond,draw] {};
  \node (I) at (1,-.4) {\footnotesize$I_2;u$};
\end{tikzpicture}
  \Big)
  =-\mathcal{P}_q^i\Bigg(  \begin{tikzpicture}
    \useasboundingbox (-.3,-0.1) rectangle (1.8,1);
  \draw (0,0) circle (4pt) ;
  \node (l1) at (0,0.35) {$0$};
  \node (l2) at (0,-0.4) {$I_{-1}$};
  \draw (0.1,-0.1) -- (1.35,-0.6);
  \draw (0.1,0.1) -- (1.35,0.6);
  \draw (1.5,-.7) node[regular polygon,regular polygon sides=6,draw] {};
  \node (m1) at (1.5,-1.1) {\footnotesize$I_2$};
  \node (m2) at (1.5,-.3) {\footnotesize$\bar{j}$};
  \draw (1.5,.7) node[regular polygon,regular polygon sides=6,draw] {};
  \node (m3) at (1.5,.3) {\footnotesize$I_0$};
  \node (m4) at (1.5,1.1) {\footnotesize$j_i$};
 \end{tikzpicture} 
 + \begin{tikzpicture}
    \useasboundingbox (-.3,-0.1) rectangle (1.8,1);
  \draw (0,0) circle (4pt) ;
  \node (l1) at (0,0.35) {$0$};
  \node (l2) at (0,-0.4) {$I_{-1}$};
  \draw[decorate, decoration={snake}] (0.1,-0.1) -- (1.32,-0.6);
  \draw (0.1,0.1) -- (1.32,0.6);
  \draw (1.5,-.7) node[diamond,draw] {};
  \node (m1) at (1.5,-1.1) {\footnotesize$I_2;u$};
  \draw (1.5,.7) node[regular polygon,regular polygon sides=6,draw] {};
  \node (m3) at (1.5,.3) {\footnotesize$I_0$};
  \node (m4) at (1.5,1.1) {\footnotesize$j_i$};
\end{tikzpicture}
\Bigg)
\label{Pq-3}
  \\[4ex]
  &=-\mathcal{P}_q^i\Big(
  \begin{tikzpicture}
    \useasboundingbox (-.3,-0.1) rectangle (2.3,.4);
  \draw (0,0) circle (4pt) ;
  \node (l1) at (0,0.35) {$0$};
  \node (l2) at (0,-0.4) {$I_1$};
    \draw (0.15,0) -- (.8,0);
    \draw (1,0) node[regular polygon,regular polygon sides=6,draw] {};
    \node (m01) at (1,-.4) {\footnotesize$I_2$};
    \node (m02) at (1,.4) {\footnotesize$j_i$};
    \draw (1.2,0) -- (1.8,0);
    \draw (2,0) node[regular polygon,regular polygon sides=6,draw] {};
    \node (m11) at (2,-.4) {\footnotesize$I_3$};
    \node (m12) at (2,.4) {\footnotesize$\bar{j}$};
  \end{tikzpicture}
  +
  \begin{tikzpicture}
    \useasboundingbox (-.3,-0.1) rectangle (2.3,.4);
  \draw (0,0) circle (4pt) ;
  \node (l1) at (0,0.35) {$0$};
  \node (l2) at (0,-0.4) {$I_1$};
    \draw (0.15,0) -- (.8,0);
    \draw (1,0) node[regular polygon,regular polygon sides=6,draw] {};
    \node (m01) at (1,-.4) {\footnotesize$I_2$};
    \node (m02) at (1,.4) {\footnotesize$\bar{j}$};
    \draw (1.2,0) -- (1.8,0);
    \draw (2,0) node[regular polygon,regular polygon sides=6,draw] {};
    \node (m11) at (2,-.4) {\footnotesize$I_3$};
    \node (m12) at (2,.4) {\footnotesize$j_i$};
  \end{tikzpicture}
  +
  \begin{tikzpicture}
    \useasboundingbox (-.3,-0.1) rectangle (2.3,.4);
    \draw (0,0) circle (4pt) ;
    \node (l1) at (0,0.35) {$0$};
    \node (l2) at (0,-0.4) {$I_1$};
     \draw (0.15,0) -- (.8,0);
    \draw (1,0) node[regular polygon,regular polygon sides=6,draw] {};
    \node (m01) at (1,-.4) {\footnotesize$I_2$};
    \node (m02) at (1,.4) {\footnotesize$j_i$};
    \draw[decorate, decoration={snake}] (1.2,0) -- (1.75,0);
    \draw (2,0) node[diamond,draw] {};
    \node (m1) at (2,-.4) {\footnotesize$I_3;u$};
  \end{tikzpicture}
 + \begin{tikzpicture}
    \useasboundingbox (-.3,-0.1) rectangle (2.3,.4);
    \draw (0,0) circle (4pt) ;
    \node (l1) at (0,0.35) {$0$};
    \node (l2) at (0,-0.4) {$I_1$};
    \draw[decorate, decoration={snake}] (0.15,0) -- (.75,0);
    \draw (1,0) node[diamond,draw] {};
    \node (m1) at (1,-.4) {\footnotesize$I_2;u$};
    \draw (1.2,0) -- (1.8,0);
    \draw (2,0) node[regular polygon,regular polygon sides=6,draw] {};
    \node (m01) at (2,-.4) {\footnotesize$I_3$};
    \node (m02) at (2,.4) {\footnotesize$j_i$};
  \end{tikzpicture}\rule[-2.4ex]{0mm}{4mm}
  \Big)\;.
  \nonumber
\end{align}
This identity removes all chains of length 3 with a v1-group labelled $j_i$
at any position. There remain only the chains of length 3 without
v1-group labelled $j_i$. For $|I_1|=1$ these are
holomorphic at $q=\beta_i$ and can be discarded in the projection
$\mathcal{P}_q^i$. The only poles come from initial v0-vertices
with $|I_1|\geq 2$ multiplied by regular expressions. We can thus use 
(\ref{polar-assertion}) for $I\mapsto I_1$ again and
express by the same mechanism as (\ref{Pq-3})
the survived length-3 chains as $-\mathcal{P}_q^i$ of
all length-4 chains which have a v1-group labelled $j_i$
at any position. These cancel in the graphical representation.
Since the  v1-group labelled $j_i$ can occur only once by
Lemma~\ref{lemma-cancel-noCatalan}, 
only the length-4 chains without
any v1-group labelled $j_i$ survive the cancellation.

We repeat this procedure up to chains of length $|I|$. The surviving
ones have an initial v0-vertex and otherwise v1/v2-groups with other
labels than $j_i$.  Now because the initial v0-vertex necessarily has
$|I_1|=1$, it is also regular at $q=\beta_i$. Therefore, all chains
survived up to this point project with $\mathcal{P}_q^i$ to $0$.
\end{proof}
\end{proposition}

\subsection{Poles of
  \texorpdfstring{$\varpi_ {0,|I|+1}(I;z)$}{\omega(I,z)} at
  \texorpdfstring{$z=-u_k$}{z=-uk}}

We prove in Section~\ref{sec:assumption}:
\begin{assumption}
\label{assump-res0}
  Let $I=\{u_1,...,u_m\}$ with $m\geq 2$. Then for every $k=1,...,m$ one has   
\[
\Res\displaylimits_{q\to  -u_k} \varpi_{0,|I|+1}(I;q)=0\;.
\]
\end{assumption}
\noindent We can thus focus on poles of second or higher order
captured by the projection
\begin{align}
\mathcal{H}^k_z \omega(u_k,z):=
\Res\displaylimits_{q\to -u_k} \Big[\Big(\frac{dz}{z-q}-\frac{dz}{z+u_k}\Big)
\omega(u_k,q)\Big]
\end{align}
for some $1$-form $\omega(u_k,z)$ in $z$ (which may depend on further
variables).  We prove:
\begin{proposition}
  \label{prop:W0-holom}
Let $I=\{u_1,...,u_m\}$ with $m\geq 2$. The projection $\mathcal{H}^k_q$
of $\varpi_{0,|I|+1}(I;q)$
is recursively given by 
  \begin{align}
    \mathcal{H}^k_q\Big( -   \begin{tikzpicture}
  \useasboundingbox (-.3,-0.1) rectangle (.3,.7);
  \draw (0,0) circle (4pt) ;
  \node (a1) at (0,0.35) {$0$};
  \node (a2) at (0,-0.4) {$I$};  
\end{tikzpicture}\Big)
=\mathcal{H}^k_q\bigg(
    \begin{tikzpicture}
  \useasboundingbox (-.3,-0.1) rectangle (1.5,.7);
  \draw (0,0) circle (4pt) ;
  \node (a1) at (0,0.35) {$0$};
  \node (a2) at (0,-0.4) {$I_{1}$};  
  \draw[decorate, decoration={snake}] (0.15,0) -- (0.75,0); 
  \draw (1,0) node[diamond,draw] {};
  \node (I) at (1,-.4) {\footnotesize$I_2;u_k$};
\end{tikzpicture}\bigg)
\label{Hkq-graph}
\end{align}
in the graphical description or explicitly by
\begin{align}
  &\mathcal{H}^k_q \varpi_{0,|I|+1}(I;q)
  \label{eq-Hkq}
  \\
  &=\sum_{s=0}^{|I|-2}
  \sum_{I_0\uplus...\uplus I_s=I{\setminus} u_k } \!\!\!\!\!\!
  \mathcal{H}^k_q\Big(
  \frac{\varpi_{0,|I_0|+1}(I_0;q)}{(R(-q)-R(u_k))(R(-u_k)-R(q))^{s+1}}
  \prod_{i=1}^s\frac{\varpi_{0,|I_i|+1}(I_i;u_k)}{dR(u_k)}
  \Big)\;.
  \nonumber
\end{align}
Remark. \normalfont Under the Assumption \ref{assump-res0} the
expression \eqref{eq-Hkq} is equal to
$\Res\displaylimits_{q\to -u_k} \frac{dz}{z-q}
\varpi_{0,|I|+1}(I;q)$. Application of $d_{u_1}\cdots d_{u_m}$ thus
coincides with \eqref{om-holqu} at $q\mapsto \iota q=-q$ and
$w\mapsto -z$. This was shown to be equivalent to (\ref{om-hol-u}) and
to \eqref{om-hol-qq}, both for $z\mapsto \iota z=-z$ and
$q\mapsto \iota z=-q$. Together with
Proposition~\ref{prop:varpi-polar} it follows that
$\omega_{0,m+1}(u_1,...,u_m,z):=d_{u_1}\cdots
d_{u_m}\varpi_{0,m+1}(u_1,...,u_m;z)$ agrees with
(\ref{sol:omega})+(\ref{eq:kernel}). Hence Theorem~\ref{thm:sol-dse}
is true if Assumption~\ref{assump-res0} holds.
\begin{proof}
  Since the second line of (\ref{om0-fraku}) only has a first-order pole at
  $z=-u_k$, the projection of (\ref{om0-fraku}) to poles of higher order reads
\begin{align}
  &\mathcal{H}^k_q\varpi_{0,|I|+1}(I;q)
  =\mathcal{H}^k_q\Big[
\sum_{I_1\uplus I_2=I}
\varpi_{|I_1|+1}(I_1;q) \mathfrak{v}_{0,|I_2|}(I_2\|q)\Big]\;.
\label{Hkq-all}
\end{align}
The rhs of (\ref{Hkq-graph}) is one of the chains contributing to the
rhs of (\ref{Hkq-all}). We prove by induction on the chain length
(with v1/v2-groups as vertices) that all other chains sum to
expressions which at $q=-u_k$ are holomorphic or have at most a
first-order pole. By $\bar{u}$ we denote any $u_l\neq u_k$.

At length 2 we have in addition to the rhs of (\ref{Hkq-graph}) the chains
\begin{tikzpicture}
  \useasboundingbox (-.3,-0.1) rectangle (1.3,.4);
  \draw (0,0) circle (4pt) ;
  \node (a1) at (0,0.35) {$0$};
  \node (a2) at (0,-0.4) {$I_{1}$};  
  \draw (.15,0) -- (.8,0);  
  \draw (1,0) node[regular polygon,regular polygon sides=6,draw] {};
  \node (j) at (1,.5) {\footnotesize$j$};
  \node (I) at (1,-.4) {\footnotesize$I_2$};
\end{tikzpicture}\rule[-3mm]{0mm}{4mm}
+
   \begin{tikzpicture}
  \useasboundingbox (-.3,-0.1) rectangle (1.3,.4);
  \draw (0,0) circle (4pt) ;
  \node (a1) at (0,0.35) {$0$};
  \node (a2) at (0,-0.4) {$I_{1}$};  
  \draw[decorate, decoration={snake}] (0.15,0) -- (0.75,0); 
  \draw (1,0) node[diamond,draw] {};
  \node (I) at (1,-.4) {\footnotesize$I_2;\bar{u}$};
\end{tikzpicture}\rule[-4mm]{0mm}{4mm}.
In the case $u_k\in I_2$ these chains are holomorphic at $q=-u_k$ and can
be discarded under $\mathcal{H}^k_q$. Remains $u_k\in I_1$. If
$I_1=\{u_k\}$, then the initial vertex v0 has weight
$-\varpi_{0,2}(u_k;q)=\frac{dq}{u_k-q}+\frac{dq}{u_k+q}$ and thus only
a first-order pole at $q=-u_k$ which does not contribute to
$\mathcal{H}^k_q$. The only contributions are thus from $u_k\in I_1$
with $|I_1|\geq 2$. Here we can use the induction hypothesis
(\ref{Hkq-graph}) for $I\mapsto I_1\mapsto I_2\mapsto I_3$ so that
\begin{align}
  &\mathcal{H}_q^k\Big(
\begin{tikzpicture}
  \useasboundingbox (-.3,-0.1) rectangle (1.3,.7);
  \draw (0,0) circle (4pt) ;
  \node (a1) at (0,0.35) {$0$};
  \node (a2) at (0,-0.4) {$I_{1}$};  
  \draw (.15,0) -- (.8,0);  
  \draw (1,0) node[regular polygon,regular polygon sides=6,draw] {};
  \node (j) at (1,.5) {\footnotesize$j$};
  \node (I) at (1,-.4) {\footnotesize$I_2$};
\end{tikzpicture}
+
   \begin{tikzpicture}
  \useasboundingbox (-.3,-0.1) rectangle (1.4,.7);
  \draw (0,0) circle (4pt) ;
  \node (a1) at (0,0.35) {$0$};
  \node (a2) at (0,-0.4) {$I_{1}$};  
  \draw[decorate, decoration={snake}] (0.15,0) -- (0.75,0); 
  \draw (1,0) node[diamond,draw] {};
  \node (I) at (1,-.4) {\footnotesize$I_2;\bar{u}$};
\end{tikzpicture}
  \Big)
\label{Hk-3}
  \\[2ex]
  &=-\mathcal{H}_q^k\Big(
  \begin{tikzpicture}
    \useasboundingbox (-.3,-0.1) rectangle (2.3,.4);
  \draw (0,0) circle (4pt) ;
  \node (l1) at (0,0.35) {$0$};
  \node (l2) at (0,-0.4) {$I_1$};
  \draw[decorate, decoration={snake}] (0.15,0) -- (0.75,0); 
  \draw (1,0) node[diamond,draw] {};
    \node (m01) at (1,-.4) {\footnotesize$I_2;u_k$};
    \draw (1.25,0) -- (1.8,0);
    \draw (2,0) node[regular polygon,regular polygon sides=6,draw] {};
    \node (m11) at (2,-.4) {\footnotesize$I_3$};
    \node (m12) at (2,.4) {\footnotesize$j$};
  \end{tikzpicture}
  +
  \begin{tikzpicture}
    \useasboundingbox (-.3,-0.1) rectangle (2.3,.4);
  \draw (0,0) circle (4pt) ;
  \node (l1) at (0,0.35) {$0$};
  \node (l2) at (0,-0.4) {$I_1$};
    \draw (0.15,0) -- (.8,0);
    \draw (1,0) node[regular polygon,regular polygon sides=6,draw] {};
    \node (m01) at (1,-.4) {\footnotesize$I_2$};
    \node (m02) at (1,.4) {\footnotesize$j$};
  \draw[decorate, decoration={snake}] (1.2,0) -- (1.75,0); 
  \draw (2,0) node[diamond,draw] {};
  \node (m01) at (2,-.4) {\footnotesize$I_3;u_k$};
  \end{tikzpicture}
  +
  \begin{tikzpicture}
    \useasboundingbox (-.3,-0.1) rectangle (2.3,.4);
    \draw (0,0) circle (4pt) ;
    \node (l1) at (0,0.35) {$0$};
    \node (l2) at (0,-0.4) {$I_1$};
   \draw[decorate, decoration={snake}] (0.15,0) -- (0.75,0); 
   \draw (1,0) node[diamond,draw] {};
   \node (m01) at (1,-.4) {\footnotesize$I_2;u_k$};
   \draw[decorate, decoration={snake}] (1.2,0) -- (1.75,0);
    \draw (2,0) node[diamond,draw] {};
    \node (m1) at (2,-.4) {\footnotesize$I_3;\bar{u}$};
  \end{tikzpicture}
 + \begin{tikzpicture}
    \useasboundingbox (-.3,-0.1) rectangle (2.4,.4);
    \draw (0,0) circle (4pt) ;
    \node (l1) at (0,0.35) {$0$};
    \node (l2) at (0,-0.4) {$I_1$};
    \draw[decorate, decoration={snake}] (0.15,0) -- (.75,0);
    \draw (1,0) node[diamond,draw] {};
    \node (m1) at (1,-.4) {\footnotesize$I_2;\bar{u}$};
   \draw[decorate, decoration={snake}] (1.2,0) -- (1.75,0); 
   \draw (2,0) node[diamond,draw] {};
   \node (m01) at (2,-.4) {\footnotesize$I_3;u_k$};
  \end{tikzpicture}\rule[-2.4ex]{0mm}{4mm}
  \Big)\;.
  \nonumber
\end{align}
This identity removes all chains of length 3 with a v2-group labelled
$u_k$ at any position. The remaining length-3 chains have edges and
v1/v2-groups which are holomorphic at $=-u_k$.  Poles arise only if
$u_k\in I_1$ in the initial vertex, and poles of second and higher
order require $|I_1|\geq 2$. Here the induction hypothesis is
available, so that the same mechanism removes all chains of length 4
with a v2-group labelled $u_k$. We repeat this construction until the
initial vertex necessarily has $|I_1|=1$ and also projects to $0$
under $\mathcal{H}^k_q$. This finishes the proof of (\ref{eq-Hkq}).
\end{proof}
\end{proposition}

As discussed in the remark directly after Proposition~\ref{prop:W0-holom},
the proof of Theorem~\ref{thm:sol-dse} is complete provided that 
Assumption~\ref{assump-res0} holds.

\section{Proof of Assumption~\ref{assump-res0}}

\label{sec:assumption}

\subsection{The residue}

The recursion formula (\ref{om0-fraku}) generates, a priori, also a
first-order pole at $z=-u_k$ with residue
\begin{align}
&\Res\displaylimits_{q\to -u_k} \varpi_{0,|I|+1}(I;q)
\label{ResWnegu}
\\
&=
\Res\displaylimits_{q\to -u_k}\Big[
\sum_{I_1\uplus I_2=I}
\varpi_{0,|I_1|+1}(I_1;q) \mathfrak{v}_{0,|I_2|}(I_2\|q)\Big]
+\mathfrak{v}_{0,|I|-1}(I{\setminus}u_k \|u_k)\;.
\nonumber
\end{align}
Our goal is to prove Assumption~\ref{assump-res0}, i.e.\ that the
residue (\ref{ResWnegu}) vanishes for $|I|\geq 2$. Of particular
importance will be the functions
\begin{align}
  \Delta\omega_{0,|I|+1}(I,z)&=\sum_{s=0}^{I-1}
  \sum_{I_0\uplus I_1\uplus...\uplus I_s=I} \nabla^{s+1}\omega_{0,|I_0|+1}(I_0,z)
  \prod_{j=1}^s\frac{\omega_{0,|I_j|+1}(I_j,\iota z)}{-dy(z)}
  \;,\qquad
  \nonumber
  \\
  \Delta\varpi_{0,|I|+1}(I;z)&=\sum_{s=0}^{I-1}
  \sum_{I_0\uplus I_1\uplus...\uplus I_s=I} \nabla^{s+1}\varpi_{0,|I_0|+1}(I_0;z)
  \prod_{j=1}^s\frac{\varpi_{0,|I_j|+1}(I_j;- z)}{dR(-z)}  \;.
  \label{def:Delta-om}
\end{align} 
These arise as follows:
\begin{lemma}
  Let $I=\{u_1,...,u_m\}$ for $m\geq 2$. Suppose Assumption~\ref{assump-res0} holds for
$\varpi_{0,|I'|+1}$ with $u_k\in I'$ and $2\leq |I'|<|I|$ (there is no condition for $m=2$).
Then
\begin{align}
 \Res\displaylimits_{z\to -u_k} \varpi_{0|I|+1}(I;z)
&= \mathfrak{v}_{0,|I|-1}(I{\setminus} u_k\|u_k)
-\mathfrak{v}_{0,|I|-1}(I{\setminus} u_k\|{-}u_k)
-\Delta \varpi_{0,|I|}(I{\setminus} u_k;{-}u_k)
\nonumber
\\
&+ \sum_{I_1\uplus I_2=I{\setminus} u_k}
\mathfrak{v}_{0,|I_1|}(I_1\|{-}u_k)
\Delta \varpi_{0,|I_2|+1}(I_2;{-}u_k)\;.
\label{ResWnegu1}
\end{align}
\begin{proof}
We evaluate the residue on the rhs of (\ref{ResWnegu}). In the graphical representation we have 
a contribution from the chain (in which $I'=\emptyset$ is allowed;
the weight of the v2-group is given in  (\ref{coll-v2}))
\begin{align}
&\Res\displaylimits_{q\to -u_k}\bigg(  
    \begin{tikzpicture}
  \useasboundingbox (-.3,-0.1) rectangle (1.4,.7);
  \draw (0,0) circle (4pt) ;
  \node (a1) at (0,0.35) {$0$};
  \node (a2) at (0,-0.4) {$I_0$};  
  \draw[decorate, decoration={snake}] (0.15,0) -- (0.75,0); 
  \draw (1,0) node[diamond,draw] {};
  \node (I) at (1,-.4) {\footnotesize$I';u_k$};
\end{tikzpicture}\bigg)_{I_0\uplus I'=I{\setminus} u_k}
\label{Res-v2-l2}
\\
&= \!\!\sum_{s=0}^{|I|-2}\sum_{I_0\uplus...\uplus I_s=I{\setminus} u_k}\!\!
\Res\displaylimits_{q\to -u_k} \!\! \Big(  
\frac{\varpi_{0,|I_0|+1}(I_0;q)}{
  (R({-}q){-}R(u_k))  (R({-}u_k){-}R(q))^{s+1}} \Big)
\prod_{i=1}^s \frac{\varpi_{0,|I_i|+1}(I_i;u_k)}{dR(u_k)}
\nonumber
\\
&= -\sum_{s=0}^{|I|-2}\sum_{I_0\uplus...\uplus I_s=I{\setminus} u_k}
\nabla^{s+1} \varpi_{0,|I_0|+1}(I_0;-u_k)
\prod_{i=1}^s \frac{\varpi_{0,|I_i|+1}(I_i;u_k)}{dR(u_k)}
\nonumber
\\
&\equiv -\Delta\varpi_{0,|I|}(I{\setminus} u_k;-u_k)
\;,
\nonumber
\end{align}
where the definition (\ref{def:nabla-om}) for $\omega\mapsto \varpi$
at $q\mapsto -u_k$ and $z\mapsto q$ has been used. This provides the
third term on the rhs of (\ref{ResWnegu1}). The first term is copied
from (\ref{ResWnegu}).

We investigate the residues at $q=-u_k$ of all other chains. The
remaining chains of length 2 (with v1/v2-groups as vertices) with
residue at $q=-u_k$ are the same as in the first line of (\ref{Hk-3})
with $u_k\in I_1$. Two cases are to distinguish. For $I_1=u_k$ we have
a purely first-order pole and
\[
  \Res\displaylimits_{q\to -u_k}\bigg(
\begin{tikzpicture}
  \useasboundingbox (-.3,-0.1) rectangle (1.3,.7);
  \draw (0,0) circle (4pt) ;
  \node (a1) at (0,0.35) {\footnotesize$0$};
  \node (a2) at (0,-0.4) {\footnotesize$u_k$};  
  \draw (.15,0) -- (.8,0);  
  \draw (1,0) node[regular polygon,regular polygon sides=6,draw] {};
  \node (j) at (1,.5) {\footnotesize$j$};
  \node (I) at (1,-.4) {\footnotesize$I{\setminus} u_k$};
\end{tikzpicture}
+
   \begin{tikzpicture}
  \useasboundingbox (-.3,-0.1) rectangle (2.4,.7);
  \draw (0,0) circle (4pt) ;
  \node (a1) at (0,0.35) {\footnotesize$0$};
  \node (a2) at (0,-0.4) {\footnotesize$u_k$};  
  \draw[decorate, decoration={snake}] (0.15,0) -- (0.75,0); 
  \draw (1,0) node[diamond,draw] {};
  \node (I) at (1.4,-.4) {\footnotesize$I{\setminus} \{u_k,\bar{u}_l\};\bar{u}_l$};
\end{tikzpicture}
\bigg)
=\bigg(
\begin{tikzpicture}
  \useasboundingbox (-.1,-0.1) rectangle (1.3,.7);
  \node (a1) at (0,0.2) {\footnotesize$0$};
  \draw (0,0) -- (.8,0);  
  \draw (1,0) node[regular polygon,regular polygon sides=6,draw] {};
  \node (j) at (1,.5) {\footnotesize$j$};
  \node (I) at (1,-.4) {\footnotesize$I{\setminus} u_k$};
\end{tikzpicture}
+
   \begin{tikzpicture}
  \useasboundingbox (-.1,-0.1) rectangle (2.4,.7);
  \node (a1) at (0,0.2) {\footnotesize$0$};
  \draw[decorate, decoration={snake}] (0,0) -- (0.75,0); 
  \draw (1,0) node[diamond,draw] {};
  \node (I) at (1.3,-.4) {\footnotesize$I{\setminus}\{u_k,\bar{u}_l\};
                          \bar{u}_l$};
\end{tikzpicture}\bigg)_{q\mapsto -u_k}\;.
\]
Amputation of the initial v0-vertex gives the chains contributing to
$-\mathfrak{v}_{0,|I'|}(I'\|q)$. Hence, 
the rhs of the above equation is the restriction of
$-\mathfrak{v}_{0,|I|-1}(I{\setminus}u_k\|{-}u_k)$ to chains of length $1$.
The other case is $u_k\in I_1$ but $|I_1|\geq 2$. Since $|I_1|<|I|$,
Assumption~\ref{assump-res0} holds for 
the initial vertex $-\varpi_{|I_1|+1}(I_1;q)$ whose poles at $q=-u_k$
are thus of purely higher order. They are thus given by 
$-\mathcal{H}^{k}_q \varpi_{0,|I_1|+1}(I_1;q)\big)$, which can be
expressed by (\ref{Hkq-graph}):
\begin{align*}
&  \Res\displaylimits_{q\to -u_k}\bigg(
\begin{tikzpicture}
  \useasboundingbox (-.3,-0.1) rectangle (1.3,.7);
  \draw (0,0) circle (4pt) ;
  \node (a1) at (0,0.35) {\footnotesize$0$};
  \node (a2) at (0,-0.4) {\footnotesize$I_1$};  
  \draw (.15,0) -- (.8,0);  
  \draw (1,0) node[regular polygon,regular polygon sides=6,draw] {};
  \node (j) at (1,.5) {\footnotesize$j$};
  \node (I) at (1,-.4) {\footnotesize$I_2$};
\end{tikzpicture}
+
   \begin{tikzpicture}
  \useasboundingbox (-.3,-0.1) rectangle (1.4,.7);
  \draw (0,0) circle (4pt) ;
  \node (a1) at (0,0.35) {\footnotesize$0$};
  \node (a2) at (0,-0.4) {\footnotesize$I_1$};  
  \draw[decorate, decoration={snake}] (0.15,0) -- (0.75,0); 
  \draw (1,0) node[diamond,draw] {};
  \node (I) at (1,-.4) {\footnotesize$I_2;\bar{u}_l$};
\end{tikzpicture}
\bigg)_{\substack{u_k\in I_1\\|I_1|\geq 2}}
\\
&={-}\!  \Res\displaylimits_{q\to -u_k}\!\!\bigg\{
\mathcal{H}^k_q\bigg(
  \begin{tikzpicture}
    \useasboundingbox (-.25,-0.1) rectangle (1.4,.4);
  \draw (0,0) circle (4pt) ;
  \node (l1) at (0,0.35) {$0$};
  \node (l2) at (0,-0.4) {$I_0$};
  \draw[decorate, decoration={snake}] (0.15,0) -- (0.75,0); 
  \draw (1,0) node[diamond,draw] {};
    \node (m01) at (1,-.4) {\footnotesize$I_1;u_k$};
  \end{tikzpicture}
  \bigg) \cdot
  \bigg(
  \begin{tikzpicture}
  \useasboundingbox (-.1,-0.1) rectangle (1.3,.7);
  \node (a1) at (0,0.35) {\footnotesize$0$};
  \draw (0,0) -- (.8,0);  
  \draw (1,0) node[regular polygon,regular polygon sides=6,draw] {};
  \node (j) at (1,.5) {\footnotesize$j$};
  \node (I) at (1,-.4) {\footnotesize$I_2$};
\end{tikzpicture}
+
   \begin{tikzpicture}
  \useasboundingbox (-.1,-0.1) rectangle (1.4,.7);
  \node (a1) at (0,0.35) {\footnotesize$0$};
  \draw[decorate, decoration={snake}] (0,0) -- (0.75,0); 
  \draw (1,0) node[diamond,draw] {};
  \node (I) at (1,-.4) {\footnotesize$I_2;\bar{u}_l$};
\end{tikzpicture}
\bigg)\bigg\}
\\
&={-}\!  \Res\displaylimits_{q\to -u_k}\!\!\bigg\{
\bigg(
  \begin{tikzpicture}
    \useasboundingbox (-.25,-0.1) rectangle (1.4,.4);
  \draw (0,0) circle (4pt) ;
  \node (l1) at (0,0.35) {$0$};
  \node (l2) at (0,-0.4) {$I_0$};
  \draw[decorate, decoration={snake}] (0.15,0) -- (0.75,0); 
  \draw (1,0) node[diamond,draw] {};
    \node (m01) at (1,-.4) {\footnotesize$I_1;u_k$};
  \end{tikzpicture}
-\frac{dq}{u_k+q}\Res\displaylimits_{q\to -u_k}\!\!\bigg(
  \begin{tikzpicture}
    \useasboundingbox (-.25,-0.1) rectangle (1.4,.4);
  \draw (0,0) circle (4pt) ;
  \node (l1) at (0,0.35) {$0$};
  \node (l2) at (0,-0.4) {$I_0$};
  \draw[decorate, decoration={snake}] (0.15,0) -- (0.75,0); 
  \draw (1,0) node[diamond,draw] {};
    \node (m01) at (1,-.4) {\footnotesize$I_1;u_k$};
  \end{tikzpicture}
\bigg)
  \bigg) \cdot
  \bigg(
  \begin{tikzpicture}
  \useasboundingbox (-.1,-0.1) rectangle (1.3,.7);
  \node (a1) at (0,0.35) {\footnotesize$0$};
  \draw (0,0) -- (.8,0);  
  \draw (1,0) node[regular polygon,regular polygon sides=6,draw] {};
  \node (j) at (1,.5) {\footnotesize$j$};
  \node (I) at (1,-.4) {\footnotesize$I_2$};
\end{tikzpicture}
+
   \begin{tikzpicture}
  \useasboundingbox (-.1,-0.1) rectangle (1.4,.7);
  \node (a1) at (0,0.35) {\footnotesize$0$};
  \draw[decorate, decoration={snake}] (0,0) -- (0.75,0); 
  \draw (1,0) node[diamond,draw] {};
  \node (I) at (1,-.4) {\footnotesize$I_2;\bar{u}_l$};
\end{tikzpicture}
\bigg)\bigg\}
\\
&= \Res\displaylimits_{q\to -u_k}\!\!\bigg(
  \begin{tikzpicture}
    \useasboundingbox (-.25,-0.1) rectangle (1.4,.4);
  \draw (0,0) circle (4pt) ;
  \node (l1) at (0,0.35) {$0$};
  \node (l2) at (0,-0.4) {$I_0$};
  \draw[decorate, decoration={snake}] (0.15,0) -- (0.75,0); 
  \draw (1,0) node[diamond,draw] {};
    \node (m01) at (1,-.4) {\footnotesize$I_1;u_k$};
  \end{tikzpicture}
  \bigg) \cdot
  \bigg(
  \begin{tikzpicture}
  \useasboundingbox (-.1,-0.1) rectangle (1.3,.7);
  \node (a1) at (0,0.35) {\footnotesize$0$};
  \draw (0,0) -- (.8,0);  
  \draw (1,0) node[regular polygon,regular polygon sides=6,draw] {};
  \node (j) at (1,.5) {\footnotesize$j$};
  \node (I) at (1,-.4) {\footnotesize$I_2$};
\end{tikzpicture}
+
   \begin{tikzpicture}
  \useasboundingbox (-.1,-0.1) rectangle (1.4,.7);
  \node (a1) at (0,0.35) {\footnotesize$0$};
  \draw[decorate, decoration={snake}] (0,0) -- (0.75,0); 
  \draw (1,0) node[diamond,draw] {};
  \node (I) at (1,-.4) {\footnotesize$I_2;\bar{u}_l$};
\end{tikzpicture}
\bigg)_{q\mapsto -u_k}
\\
&-\!  \Res\displaylimits_{q\to -u_k}\!\!\bigg(
  \begin{tikzpicture}
    \useasboundingbox (-.25,-0.1) rectangle (2.25,.4);
  \draw (0,0) circle (4pt) ;
  \node (l1) at (0,0.35) {$0$};
  \node (l2) at (0,-0.4) {$I_1$};
  \draw[decorate, decoration={snake}] (0.15,0) -- (0.75,0); 
  \draw (1,0) node[diamond,draw] {};
    \node (m01) at (1,-.4) {\footnotesize$I_2;u_k$};
    \draw (1.25,0) -- (1.8,0);
    \draw (2,0) node[regular polygon,regular polygon sides=6,draw] {};
    \node (m11) at (2,-.4) {\footnotesize$I_3$};
    \node (m12) at (2,.4) {\footnotesize$j$};
  \end{tikzpicture}
  +
  \begin{tikzpicture}
    \useasboundingbox (-.25,-0.1) rectangle (2.25,.4);
  \draw (0,0) circle (4pt) ;
  \node (l1) at (0,0.35) {$0$};
  \node (l2) at (0,-0.4) {$I_1$};
    \draw (0.15,0) -- (.8,0);
    \draw (1,0) node[regular polygon,regular polygon sides=6,draw] {};
    \node (m01) at (1,-.4) {\footnotesize$I_2$};
    \node (m02) at (1,.4) {\footnotesize$j$};
  \draw[decorate, decoration={snake}] (1.2,0) -- (1.75,0); 
  \draw (2,0) node[diamond,draw] {};
  \node (m01) at (2,-.4) {\footnotesize$I_3;u_k$};
  \end{tikzpicture}
  +
  \begin{tikzpicture}
    \useasboundingbox (-.25,-0.1) rectangle (2.25,.4);
    \draw (0,0) circle (4pt) ;
    \node (l1) at (0,0.35) {$0$};
    \node (l2) at (0,-0.4) {$I_1$};
   \draw[decorate, decoration={snake}] (0.15,0) -- (0.75,0); 
   \draw (1,0) node[diamond,draw] {};
   \node (m01) at (1,-.4) {\footnotesize$I_2;u_k$};
   \draw[decorate, decoration={snake}] (1.25,0) -- (1.75,0);
    \draw (2,0) node[diamond,draw] {};
    \node (m1) at (2,-.4) {\footnotesize$I_3;\bar{u}_l$};
  \end{tikzpicture}
 + \begin{tikzpicture}
    \useasboundingbox (-.25,-0.1) rectangle (2.4,.4);
    \draw (0,0) circle (4pt) ;
    \node (l1) at (0,0.35) {$0$};
    \node (l2) at (0,-0.4) {$I_1$};
    \draw[decorate, decoration={snake}] (0.15,0) -- (.75,0);
    \draw (1,0) node[diamond,draw] {};
    \node (m1) at (1,-.4) {\footnotesize$I_2;\bar{u}_l$};
   \draw[decorate, decoration={snake}] (1.25,0) -- (1.75,0); 
   \draw (2,0) node[diamond,draw] {};
   \node (m01) at (2,-.4) {\footnotesize$I_3;u_k$};
  \end{tikzpicture}\rule[-2.4ex]{0mm}{4mm}
\bigg).
\end{align*}
In the step from the second to third line we have used that only the
whole projection $\mathcal{H}^k_q+\frac{dq}{q+u_k}\Res_{q\to -u_k}$ to
the principal part of a Laurent series, but not $\mathcal{H}^k_q$
alone, is the identity operator under the residue.  According to
(\ref{Res-v2-l2}), the second line from below equals
$-\Delta \varpi_{0,|I_0|+|I_1|+1}(I_0\cup I_1;-u_k)$ times the
restriction of $-\mathfrak{v}_{0,|I_3|}(I_3\|{-}u_k)$ to chains of
length $1$, here with $I_0\uplus I_1\uplus I_3=I{\setminus} u_k$.  The
last line removes from the residue all chains of length 3 with a
v2-group labelled $u_k$.

Thus only those length-3 chains for which $u_k\in I_1$ is located at
the initial vertex v0 contribute to the remaining residue. Again the
case $I_1=u_k$ produces the restriction of
$-\mathfrak{v}_{0,|I|-1}(I{\setminus} u_k\|{-}u_k)$ to chains of
length $2$. For $|I_1|\geq 2$ we use Assumption~\ref{assump-res0} that
$-\varpi_{0,|I_1|+1}(I_1;q)$ has at $q=-u_k$ a pole of second or
higher order given by (\ref{Hkq-graph}). The same argument as before
produces on one hand (\ref{Res-v2-l2}) times the restriction of
$-\mathfrak{v}_{0,|I_3|}(I_3;{-}u_k)$ to chains of length $2$, on the
other hand removes from the residue all chains of length 4 with a
v2-group labelled $u_k$.  Continuing this strategy until $I_1=u_k$ is
the only choice shows that the residue of all chains other than
(\ref{Res-v2-l2}) evaluates to
$-\mathfrak{v}_{0,|I|-1}(I{\setminus}u_k\|{-}u_k)$ plus
(\ref{Res-v2-l2}) times $-\mathfrak{v}_{0,|I'|}(I'\|{-}u_k)$, summed
over partitions of $I{\setminus} u_k$.
\end{proof}
\end{lemma}

Assumption~\ref{assump-res0} for $\varpi_{0,|I|+1}$, that the rhs of
(\ref{ResWnegu1}) evaluates to $0$, is thus a condition on
$\varpi_{0,|I'|+1}$ or $\omega_{0,|I'|+1}$ for $|I'|<I$. Here
Theorem~\ref{thm:sol-dse} is the induction hypothesis. Its proof is
complete (following the previous considerations) if
Theorem~\ref{thm:sol-dse} implies
\begin{align}
0  &=  \mathfrak{v}_{0,|I|}(I\|{-}q)
-\mathfrak{v}_{0,|I|}(I\|q)
-\Delta \varpi_{0,|I|+1}(I;q)
+ \!\!\! \sum_{I_1\uplus I_2=I} \!\!\!
\mathfrak{v}_{0,|I_1|}(I_1\|q)\Delta \varpi_{0,|I_2|+1}(I_2;q)\,.
  \label{reflect-U}
\end{align}
We are going to prove that the rhs of (\ref{reflect-U}) is an entire
holomorphic function on $\hat{\mathbb{C}}$, i.e.\ a constant, equal
to its value $0$ for $q\to \infty$. This implies (\ref{reflect-U}).

We start to discuss absence of poles at $q=\pm \beta_i$. 
Recall that (\ref{reflect-U}) equals $\Res_{z\to q} \varpi_{0,|I|+2}(I,-q;z)$.
The projection of (\ref{reflect-U}) to poles at $q=\beta_i$ is thus given by
\begin{align*}
  \Res\displaylimits_{q\to \beta_i}\frac{\eqref{reflect-U} dq}{w-q}
  &=  \Res\displaylimits_{q\to \beta_i} \Res\displaylimits_{z\to q}
  \frac{\varpi_{0,|I|+2}(I,-q;z) dq}{w-q}
  \\
  &= - \Res\displaylimits_{q\to \beta_i} \Res\displaylimits_{z\to \beta_i}
  \frac{\varpi_{0,|I|+2}(I,-q;z) dq}{w-q}
  + \Res\displaylimits_{z\to \beta_i} \Res\displaylimits_{q\to \beta_i}
  \frac{\varpi_{0,|I|+2}(I,-q;z) dq}{w-q}
\end{align*}  
when taking (\ref{commute-res}) into account. The final term gives
zero because non of the chains contributing to
$\varpi_{0,|I|+2}(I,-q;z)$ has a pole at $-q=-\beta_i$.  The other
term is also zero because $\varpi_{0,|I|+2}(I,-q;z)$ has 
due to the kernel $K_i(z,q)$ in Theorem~\ref{thm:sol-dse} at
$z=\beta_i$ poles of purely higher order, without residue. We have
established this fact in Proposition~\ref{prop:varpi-polar} without relying on
Assumption~\ref{assump-res0}. In summary, (\ref{reflect-U}) is regular
at $q=\beta_i$.  We will show in Subsection~\ref{sec:res0-necessary}
that (\ref{reflect-U}) is antisymmetric under $q\mapsto -q$.  This
means that (\ref{reflect-U}) is also regular at $q=-\beta_i$.

The same simple argument cannot be used to prove that
(\ref{reflect-U}) is regular at $q=\pm u_k$ because this would need
Assumption~\ref{assump-res0}. We therefore give in
Subsection~\ref{sec:res0-u} a direct proof which uses the
antisymmetry of (\ref{reflect-U}).

In principle, the functions $\mathfrak{v}_{0,|I|}(I\|{\pm}q)$ may (and do)
have poles at the other preimages
$q=v$ where $v\in \{\pm \widehat{u_k}^j,\pm \widehat{(-u_k)}^j\}$. 
Recalling  that (\ref{reflect-U}) equals $\Res_{z\to q} \varpi_{0,|I|+2}(I,-q;z)$,
the projection of (\ref{reflect-U}) to a pole at such $q=v$ is 
\begin{align*}
  \Res\displaylimits_{q\to v}\frac{\eqref{reflect-U} dq}{w-q}
  &=  \Res\displaylimits_{q\to v} \Res\displaylimits_{z\to q}
  \frac{\varpi_{0,|I|+2}(I,-q;z) dq}{w-q}
  \\
  &= - \Res\displaylimits_{q\to v} \Res\displaylimits_{z\to v}
  \frac{\varpi_{0,|I|+2}(I,-q;z) dq}{w-q}
  + \Res\displaylimits_{z\to v} \Res\displaylimits_{q\to v}
  \frac{\varpi_{0,|I|+2}(I,-q;z) dq}{w-q}\;.
\end{align*}  
The first term in the last line is trivially zero, but the second term
can indeed have a pole at $-q=\widehat{u_k}^j$ coming from the edge e4
in Table~\ref{tab1}. An edge with these labels can only occur once in
a chain so that it is a first-order pole. Its residue
$\Res\displaylimits_{q\to v} \frac{\varpi_{0,|I|+2}(I,-q;z) dq}{w-q}$
is a $1$-form in $z$ from which we take the residue at the same
$z=-\widehat{u_k}^j$. But there are no such poles so that
(\ref{reflect-U}) is regular at any
$q\in \{\pm \widehat{u_k}^j,\pm \widehat{(-u_k)}^j\}$.

\subsection{A necessary condition}

\label{sec:res0-necessary}

Adding (\ref{reflect-U}) and its copy for $q\to -q$ shows that 
necessary for (\ref{reflect-U}) to be true is the identity (we apply $d_{u_1}\cdots d_{u_m}$ to
pass to $\omega$)
\begin{align}
0&=  
\Delta\omega_{0,|I|+1}(I,q)+\Delta\omega_{0,|I|+1}(I,{-}q)
- \sum_{I_1\uplus I_2=I}
\Delta\omega_{0,|I_1|+1}(I_1,q)\Delta\omega_{0,|I_2|+1}(I_2,{-}q)\;.
\label{Deltaom-qmq}
\end{align}
This is true for $I=\{u\}$. Equations of such type can be disentangled by repeated
insertion into itself, which is the same operation as a treatment of formal power series
(which here are in fact polynomials). This shows that  (\ref{Deltaom-qmq}) is equivalent to
\begin{align}
&  -\Big[\log \big(1-\Delta\omega_{0,|.|+1}(\,.\,,q) \big)\Big](I)
  -\Big[\log \big(1-\Delta\omega_{0,|.|+1}(\,.\,,{-}q) \big)\Big](I)=0
  \label{log-flip}
  \\
  &\text{where}\quad
  -\Big[\log \big(1-\Delta\omega_{0,|.|+1}(\,.\,,q) \big)\Big](I)
\equiv \sum_{s=1}^{I} \frac{1}{s} \sum_{I_1\uplus....\uplus I_s=I} 
\prod_{j=1}^s \Delta\omega_{0,|I_j|+1}(I_j,q)\;.
\nonumber
\end{align}
This logarithm can also be represented as follows:
\begin{proposition}
\label{prop:log-Delta-om}
  \begin{align}
&\sum_{r=1}^{|I|} \frac{1}{r} \sum_{I_1\uplus....\uplus I_r=I} 
\prod_{j=1}^r \Delta\omega_{0,|I_j|+1}(I_j,q)
\label{log-Delta-om}
\\
&=\sum_{s=1}^{|I|} \frac{1}{s}\sum_{I_1\uplus ... \uplus I_s=I}
  \sum_{n_1+...+n_s=s} \prod_{j=1}^s
  \nabla^{n_j}\omega_{0,|I_j|+1}(I_j,q)\;.
\nonumber
\end{align}
\begin{proof}
  We write Lemma~\ref{lem:omega-nabla} as
\begin{align}
\frac{\omega_{0,|I|+1}(I,\iota q)}{-dy(q)}
&=\sum_{s=1}^{|I|} \sum_{\substack{I_1\uplus ...\uplus I_s=I \\ I_1<...<I_s}}
(s-1)! \sum_{n_1+...+n_s=s-1}
\prod_{i=1}^s \nabla^{n_i}\omega_{0,|I_i|+1}(I_i,q)
\label{omega-log-nabla}
\end{align}
and insert it into the product in (\ref{def:Delta-om}). This shows
that products of $\Delta\omega_{0,|I_j|+1}(I_j,q)$ expand into
products of $\nabla^{n_j}\omega_{0,|I_j|+1}(I_j,q)$ with the given
condition on the sum of $n_j$.  That the prefactor reduces to
$\frac{1}{s}$ is, however, by no means obvious.  The first step of the
proof is Lemma~\ref{lem:Delta-om} below, which relies on
Corollary~\ref{coro:combin-l} in the Appendix.
Then a discussion given after the proof
of Lemma~\ref{lem:Delta-om} completes the proof. It relies on the same 
Corollary \ref{coro:combin-l}.
\end{proof}
\end{proposition}

\begin{lemma}
\label{lem:Delta-om}
  \begin{align*}
\Delta\omega_{0,|I|+1}(I,q)
&=\nabla^{1}\omega_{0,|I|+1}(I,q)
\\
&+
\sum_{s=2}^{|I|} \sum_{I_1\uplus ... \uplus I_s=I}
  \sum_{n_1+...+n_s=s} \frac{\#(n_j=0)}{s(s-1)}\prod_{j=1}^s
  \nabla^{n_j}\omega_{0,|I_j|+1}(I_j,q)\;.
\end{align*}
\begin{proof}
As discussed before we have a representation
\begin{align}
\Delta \omega_{0,|I|+1}(I,q)
&=\sum_{s=1}^{|I|}
\sum_{\substack{I_1\uplus...\uplus I_s=I\\  I_1<...<I_s}}
\sum_{\substack{n_1+...+n_s=s\\ \max(n_j)\geq 2\text{ if } s>1}}
C^1_{n_1...n_s} 
\prod_{j=1}^s\nabla^{n_j} \omega_{0,|I_j|+1}(I_i,q)
\label{om-Iq-C}
\end{align} 
in which $C^1_{n_1...n_s}$ is symmetric in all its arguments. To
determine $C^1_{n_1...n_s}$ we can consider a subsector of the
$n_j$-summations where $n_1,...,n_p>0$ for some $p$ and
$n_{p+1}=...=n_s=0$. Other sectors are then obtained by symmetry.  We
will count the contributions from (\ref{def:Delta-om}) which
contribute to $C^1_{n_1...n_p0...0}$ for given positive integers
$n_1,...,n_p$ (which are followed by $n_1+...+n_p-p$ zeros). In a first
step we show that the number of these contributions is $C^1_{s0...0}=(s-1)!$
(which is clear) and for $p\geq 2$ given by
\begin{align}
&C^1_{n_1...n_p0...0}
\label{C-ns0}
\\
&= \sum_{\ell=1}^{p-1} \sum_{j=1}^{p} (n_j{-}1)! 
\sum_{\substack{J_1\uplus ... \uplus J_\ell =\{1,...,p\}\setminus \{j\} \\
     J_1<....<J_\ell}}
\frac{(n_1+...+n_p-p)!}{(n_j-\ell-1)!}
\prod_{i=1}^\ell \frac{(|\underline{n}|_{J_i})!}{
  (|\underline{n}|_{J_i}-|J_i|+1)!}\;.
\nonumber
\end{align}

The number $C^1_{n_1...n_p0...0}$ is the sum over all $j$ with $n_j\geq 2$
of specially ordered contributions from 
\begin{align}
\sum_{\substack{\tilde{I}_1\uplus...\uplus \tilde{I}_{n_j}=I\\ I_j=\tilde{I}_j,~ \tilde{I}_1<...<\tilde{I}_{n_j}}}
\!\!\!
\nabla^{n_j} \omega_{0,|I_j|+1}(I_j,q)
(n_j-1)! \prod_{\substack{i=1 \\ i\neq j}}^{n_j}
\frac{\omega_{0,|\tilde{I}_i|+1}(\tilde{I}_i,\iota q)}{-dy(q)}\;.
\label{Xnj-0}
\end{align}
The factors
$\frac{\omega_{0,|\tilde{I}_i|+1}(\tilde{I}_i,\iota q)}{-dy(q)}$ are
expressed via (\ref{omega-log-nabla}), but only contributions
compatible with $n_{p+1}=...=n_s=0$ are retained.  The positive
$n_1,...,n_p$, excluding $n_j$, arise from the part of
(\ref{omega-log-nabla}) in which all factors $\nabla^0\omega$ have a
larger order than any $\nabla^{r}\omega$ with $r>0$.  In particular,
contributions $\nabla^{r}\omega$ with $r>0$ only arise from every of
the first $\ell$ factors
$\frac{\omega_{0,|\tilde{I}_i|+1}(\tilde{I}_i,\iota q)}{-dy(q)}$ in
(\ref{Xnj-0}), for some $\ell$ with $1\leq \ell<p-1$ to sum over.
From the last $n_j-1-\ell$ factors in (\ref{Xnj-0}) we only take the
special term $\nabla^0\omega_{0,|\tilde{I}_i|+1}(\tilde{I}_i,q)$.

An expansion (\ref{omega-log-nabla}) used for the first $\ell$ factors
(\ref{Xnj-0}) contributes to the specially ordered
$C^1_{n_1...n_p0...0}$ whenever $\{1,...,p\}\setminus\{j\}$ is partitioned
into $J_1\uplus ...\uplus J_\ell$ with $\min J_k<\min J_l$ for every
pair $k<l$, where $\min J_k$ is the smallest integer in the set
$J_k$. Then the subset $\tilde{I}_i$ in (\ref{omega-log-nabla})
\begin{itemize}
\item contains $\bigcup_{k\in J_i} I_k$ if $i<j$;
\item contains $\bigcup_{k\in J_{i-1}} I_k$ if $j<i\leq p$.
\end{itemize}
We let $|\underline{n}|_{J_i} =\sum_{k\in J_i} n_k$.  To be an
admissible contribution to (\ref{omega-log-nabla}) the factors
$\nabla^{n_k}\omega$ in the $i$-th block must be supplemented by
$|\underline{n}|_{J_i}-|J_i|+1$ factors $\nabla^0\omega$.

Hence, the number $C^1_{n_1...n_p0...0}$ is given by the sum over $j$
and $\ell$ of
\begin{itemize}
\item a sum over ordered partitions
$\{1,...,p\}\setminus\{j\}= J_1\uplus ...\uplus J_\ell$

\item of a factor $(n_j-1)!$ from (\ref{Xnj-0})

\item  times a factor  $(|\underline{n}|_{J_i})!$ for every $1\leq i\leq \ell$
  which is the factor $(s{-}1)!$ in (\ref{omega-log-nabla})

\item times the number of distributions of the $n_1+...+n_p-p$ factors
  $\nabla^0\omega$ into $\ell+1$ blocks, namely
  \begin {enumerate}
\item    a block of
$(n_j-1-\ell)$ factors where from
$\frac{\omega_{0,|\tilde{I}_i|+1}(\tilde{I}_i,\iota q)}{-dy(q)}$ only
the special term
$\nabla^0\omega_{0,|\tilde{I}_i|+1}(\tilde{I}_i,q)$ is retained;

\item $\ell$ blocks of $|\underline{n}|_{J_i}-|J_i|+1$ factors which
  supplement the $\prod_{k\in J_i}\nabla^{n_k}\omega$ in a
  non-trivially expanded
  $\frac{\omega_{0,|\tilde{I}_i|+1}(\tilde{I}_i,\iota q)}{-dy(q)}$.
\end{enumerate}
There are
$\frac{(n_1+...+n_p-p)!}{(n_j-1-\ell)! \prod_{i=1}^\ell
  (|\underline{n}|_{J_i}-|J_i|+1)!}$ such distributions, which is a
valid multinomial coefficient due to
$\sum_{i=1}^\ell |\underline{n}|_{J_i}{=}n_1{+}...{+}n_p{-}n_j$ and
$\sum_{i=1}^\ell |J_i|{=}p{-}1$.
\end{itemize}
This number is (\ref{C-ns0}). We remark that the restriction to
$n_j\geq 2$ is automatic because $\frac{1}{(n_j-\ell-1)!}$ gives zero
for $n_j=1$.

We write (\ref{C-ns0}) in terms of falling factorials (see
Corollary~\ref{cor:partition}), insert (\ref{eq:partition-falling})
and shift $\ell-1\mapsto r$:
  \begin{align*}
&C^1_{n_1...n_p0...0}
\nonumber
\\[-1ex]
&= (n_1+...+n_p-p)! \sum_{\ell=1}^{p-1}
\sum_{j=1}^{p} (n_j{-}1) (n_j{-}2)^{\underline{\ell-1}}
\sum_{\substack{J_1\uplus ... \uplus J_\ell =\{1,...,p\}\setminus \{j\} \\
     J_1<....<J_\ell}}
\prod_{i=1}^\ell (|\underline{n}|_{J_i})^{\underline{|J_i|-1}}
\nonumber
\\
&= (n_1+...+n_p-p)! \sum_{r=0}^{p-2}
\sum_{j=1}^{p} (n_j{-}1) (n_j{-}2)^{\underline{r}}
\binom{p-2}{r} 
(n_1+...+n_p-n_j)^{\underline{p-2-r}}
\nonumber
\\*
&= (n_1+...+n_p-p)! 
\sum_{j=1}^{p} (n_j{-}1) 
(n_1+...+n_p-2)^{\underline{p-2}}
\nonumber
\\*
&= (n_1+...+n_p-2)!(n_1+...+n_p-p)\equiv (s-2)!(s-p)\;.
\end{align*}
In the fourth line we have used the binomial theorem for
the falling factorial. The final line is obvious.

For a general order of the $n_i$ we thus have 
$C^1_{n_1...n_s}=(s-2)!\#(n_j=0)$, where
$\#(n_j=0)$ is the number of $n_j$ which equal zero. Relaxing the
condition that the $I_j$ are ordered amounts to
an additional factor $\frac{1}{s!}$. This is the assertion.
\end{proof}  
\end{lemma}

Lemma~\ref{lem:Delta-om} is the starting point to evaluate the sum
over $r$ in the first line of (\ref{log-Delta-om}). It is clear that
this sum has a similar expansion as (\ref{om-Iq-C}):
\begin{align}
&-\big[\log(1-\Delta \omega_{0,|.|+1}(\,.\,,q))\big](I)
\label{log-Delta-om1}
\\
&=\sum_{s=1}^{|I|}
\sum_{\substack{I_1\uplus...\uplus I_s=I\\ I_1<...<I_s}}
\sum_{n_1+...+n_s=s} C_{n_1...n_s} 
\prod_{j=1}^s\nabla^{n_j} \omega_{0,|I_j|+1}(I_i,q)
\nonumber
\end{align} 
where $C_{n_1...n_s}$ is symmetric. We first show that for an order
$n_1,...,n_p\geq 1$ and $n_{p+1}=...=n_s=0$ it is given by
\begin{align}
  &  C_{n_1...n_p0...0}
  \label{C-ns-all}
\\
&= \sum_{r=1}^{p}(r-1)!
  \sum_{\substack{J_1\uplus...\uplus J_r=\{1,...,p\}\\
       J_1<....<J_r}}
  (n_1+...+n_p-p)! \prod_{i=1}^r \frac{(|\underline{n}|_{J_i}-2)!
    (|\underline{n}|_{J_i}-|J_i|)}{(|\underline{n}|_{J_i}-|J_i|)!}\;.
  \nonumber
\end{align}
The factor $(r-1)!$ combines the step from any partitions
$I_1\uplus...\uplus I_r=I$ into $r!$ ordered ones with the prefactor
$\frac{1}{r}$ in (\ref{log-Delta-om}). The subset $J_i$ corresponds to
$\Delta \omega_{0,|\tilde{I}_i|+1} (\tilde{I}_i,q)$ with
$\tilde{I}_i =\bigcup_{j\in J_i} I_j \cup
\bigcup_{k=1}^{|\underline{n}|_{J_i}-|J_i|} I'_k$ where the $I'_k$ are
taken from $I_{p+1},...,I_{n_1+...+n_p}$. There are
$\frac{(n_1+...+n_p-p)!}{ \prod_{i=1}^r
  (|\underline{n}|_{J_i}-|J_i|)!}$ different distributions of these
$I_k'$, which explains the corresponding factor above. The numerator
$(|\underline{n}|_{J_i}-2)! (|\underline{n}|_{J_i}-|J_i|)$ is the
weight of $C^1_{\underline{n}_{J_i}0....0}$ found in Lemma~\ref{lem:Delta-om}.

We write (\ref{C-ns-all}) in terms of rising factorials and 
insert (\ref{eq:partition-rising}), where $x_j\mapsto n_j-1$
and $\ell\mapsto r$:
\begin{align*}
  &  C_{n_1...n_p0...0}
\\
&= \sum_{r=1}^{p}(r-1)!
   (n_1+...+n_p-p)! \prod_{i=1}^r 
   \sum_{\substack{J_1\uplus...\uplus J_r=\{1,...,p\}\\
        J_1<....<J_r}}
     (|\underline{n}|_{J_i}-|J_i|)^{\overline{|J_i|-1}}
   \nonumber
 \\
  &=\sum_{r=1}^{p}(r-1)!
  (n_1+...+n_p-p)! \binom{p-1}{r-1}
  (n_1+...+n_p-p)^{\overline{p-r}}
  \nonumber
  \\
  &=
  (n_1+...+n_p-p)! 
  (n_1+...+n_p-p+1)^{\overline{p-2}}=  (n_1+...+n_p-1)! \equiv (s-1)!\;.
\end{align*}
We have used $(r-1)!=1^{\overline{r-1}} $ and then applied
the binomial theorem.

By symmetry we thus have $C_{n_1...n_s}=(s-1)!$ for any partition
$n_1+...+n_s=s$.  Relaxing the condition that the $I_j$ are ordered
has to be corrected with an additional factor $\frac{1}{s!}$.
This completes the proof of Proposition~\ref{prop:log-Delta-om}.

Proposition~\ref{prop:log-Delta-om} together with (\ref{log-flip})
give as necessary condition for (\ref{reflect-U}) to be true
the equality (\ref{log-nabla-sym}) which we have proved in
Section~\ref{sec:prod-nabla-sym}.

\subsection{Absence of poles of (\ref{reflect-U})
  at \texorpdfstring{$q=-u_k$}{q=-u}}

\label{sec:res0-u}

The function $\mathfrak{v}_{0,|I|}(I\|q)$ is holomorphic at $u_k=q$ and
has poles at $u_k=-q$ which exclusively come from v2-groups with label
$u_k$. There are two possibilities. Either this v2-group is the single
vertex of a length-1 chain $-\Big(
    \begin{tikzpicture}
  \useasboundingbox (0,-0.1) rectangle (1.4,.7);
  \node (a1) at (0.2,0.35) {\footnotesize$0$};
  \draw[decorate, decoration={snake}] (0.15,0) -- (0.75,0); 
  \draw (1,0) node[diamond,draw] {};
  \node (I) at (1,-.4) {\footnotesize$I';u_k$};
\end{tikzpicture}\Big)_{I=I'\uplus u_k}$\rule[-4mm]{0pt}{4mm},
or it is part of a chain of larger length. In the second case it can
be collectively taken out of all other chains, the remnant is just
another copy of $-\mathfrak{v}_{0,|I|}(I'\|q)$ of smaller length with
$I'\notin u_k$:
\begin{align}
&\mathfrak{v}_{0,|I|}(I\|q) + \mathcal{O}((q+u_k)^0)
\label{UIq-mq}
\\
&= \sum_{s=0}^{|I|-1} 
\sum_{I_1\uplus...\uplus I_s=I{\setminus} u_k}^{(I\setminus u_k\neq \emptyset)}
\frac{1}{(R(-q)-R(u_k))(R(-u_k)-R(q))^{s+1}}
\prod_{j=1}^s\frac{\varpi_{0,|I_j|+1}(I_j;u_k)}{dR(u_k)}
\nonumber
\\
&-  \sum_{s=0}^{|I|-2} \sum_{I_{-1}\uplus I_1\uplus...\uplus I_s=I\setminus u_k}
\frac{\mathfrak{v}_{0,|I_{-1}|}(I_{-1}\|q)}{(R(-q)-R(u_k))(R(-u_k)-R(q))^{s+1}}
\prod_{j=1}^s\frac{\varpi_{0,|I_j|+1}(I_j;u_k)}{dR(u_k)}\;.
\nonumber
\end{align}
We recall that \mbox{\footnotesize$\displaystyle\sum^{(I'\neq \emptyset)}$} 
indicates that for $|I'|=0$ the sum is omitted, whereas for 
$|I'|>0$ the case $s=0$ is left out.
The following lemma (which we formulate for $q\mapsto -q$)
characterises the polar part of the second line at $q=-u_k$:
\begin{lemma}
\label{lem:v2-prodnabla}
\begin{align}
&  \sum_{s=1}^{|I|} \frac{1}{s}\sum_{I_1\uplus...\uplus I_s=I}
\sum_{n_1+...+n_s=s}  \prod_{j=1}^s \nabla^{n_i}\omega_{0,|I_j|+1}(I_j,q)
\label{v2-prodnabla}
\\  
&=  d_{u_k}\bigg[\sum_{s=0}^{|I|-1} 
\sum_{I_1\uplus...\uplus I_s=I\setminus u_k}^{(I\setminus u_k\neq \emptyset)}
\frac{\prod_{j=1}^s\frac{\omega_{0,|I_j|+1}(I_j,u_k)}{dx(u_k)}
}{(x(q)-x(u_k))(y(q)-y(u_k))^{s+1}}
\bigg]
+\mathcal{O}((q-u_k)^0)\;.
\nonumber
\end{align}
\begin{proof}
The lhs of (\ref{v2-prodnabla}) can be rewritten with (\ref{def:nabla-om}) as
\begin{align*}
  \eqref{v2-prodnabla}_{\text{lhs}}
  &=\sum_{s=1}^{|I|} \frac{1}{s}\sum_{I_1\uplus ... \uplus I_s=I}
\frac{(-1)^s}{s!}
\lim_{z\to q} \frac{\partial^s}{\partial (x(z))^s}
\Big(\Big( \frac{x(z)-x(q)}{y(z)-y(q)}\Big)^s
\prod_{j=1}^s
\frac{\omega_{0,|I_j|+1}(I_j,z)}{dx(z)}\Big)
\\
&=\Res\displaylimits_{z\to q} \Big[
  \sum_{s=1}^{|I|} \frac{1}{s}\sum_{I_1\uplus ... \uplus I_s=I}
\Big( \frac{dx(z)}{(x(z){-}x(q))(y(q){-}y(z))^s}
\prod_{j=1}^s
\frac{\omega_{0,|I_j|+1}(I_j,z)}{dx(z)}\Big)\Big]\,.
\end{align*}
Up to $\mathcal{O}((q-u_k)^0)$-contributions it coincides with its projection
to the polar part in which we change with (\ref{commute-res}) the order of residues:
\begin{align*}
&  \eqref{v2-prodnabla}_{\text{rhs}}+\mathcal{O}((q-u_k)^0)
\\
&=\Res\displaylimits_{w\to u_k} \frac{dw}{q-w}
  \Res\displaylimits_{z\to w} \Big[
  \sum_{s=1}^{|I|} \frac{1}{s}\sum_{I_1\uplus ... \uplus I_s=I}
\Big( \frac{dx(z)\prod_{j=1}^s
\frac{\omega_{0,|I_j|+1}(I_j,z)}{dx(z)}
}{(x(z)-x(w))(y(w)-y(z))^s}\Big)\Big]
\\
&=-\Res\displaylimits_{w\to u_k} \frac{dw}{q-w}
  \Res\displaylimits_{z\to u_k} \Big[
  \sum_{s=0}^{|I|-1} \sum_{I_0\uplus ... \uplus I_s=I{\setminus} u_k}
\Big( \frac{\omega_{0,2}(u_k,z) \prod_{j=1}^s
\frac{\omega_{0,|I_j|+1}(I_j,z)}{dx(z)}
}{(x(z)-x(w))(y(w)-y(z))^{s+1}}\Big)\Big]
\\
&=d_{u_k}\bigg[ \Res\displaylimits_{w\to u_k} \frac{dw}{q-w}
\sum_{s=1}^{|I|-1} \sum_{I_0\uplus ... \uplus I_s=I{\setminus} u_k}
\Big( \frac{\prod_{j=1}^s
\frac{\omega_{0,|I_j|+1}(I_j,u_k)}{dx(u_k)}
}{(x(u_k)-x(w))(y(w)-y(u_k))^{s+1}}\Big)\bigg]
\\
&=d_{u_k}\bigg[ 
\sum_{s=1}^{|I|-1} \sum_{I_0\uplus ... \uplus I_s=I{\setminus} u_k}
\Big( \frac{\prod_{j=1}^s
\frac{\omega_{0,|I_j|+1}(I_j,u_k)}{dx(u_k)}
}{(x(u_k)-x(q))(y(q)-y(u_k))^{s+1}}\Big) + \mathcal{O}((q-u_k)^0)
\bigg]\;.
\end{align*}
We have used that only $\omega_{0,2}$ has a pole at $z=u_k$ which is
given by (\ref{om02}).
\end{proof}  
\end{lemma}
With these preparations we control the polar part of (\ref{reflect-U})
at $q=\pm u_k$:
\begin{proposition}
Holomorphicity of \eqref{reflect-U} at $q=-u_k$ is a consequence of
\eqref{log-nabla-sym} proved in Proposition~\ref{prop:nabla-sym}.
\begin{proof}
  The first term $\mathfrak{v}_{0,|I|}(I\|{-}q)$ in (\ref{reflect-U})
  is holomorphic at $q=-u_k$.  In the sum over $I_1\uplus I_2=I$ we
  distinguish $u_k\in I_1$ from $u_k\in I_2$. The function
  $\mathfrak{v}_{0,|I|}(I\|q)$ is for $u_k\in I$ written as
  (\ref{UIq-mq}) with Lemma~\ref{lem:v2-prodnabla} used for the
  rhs. We thus find
\begin{align*}
  \eqref{reflect-U}
  &= -A(I;q)+ \sum_{\substack{I_1\uplus I_2=I\\ u_k\in I_1}}
  A(I_1;q)\mathfrak{v}_{0,|I_2|}(I_2\|q)
  + \mathcal{O}((q+u_k)^0)\qquad \text{where}
  \\
  A(I;q)&:=\Delta\varpi_{0,|I|+1}(I;q)
  + \sum_{s=1}^{|I|} \frac{1}{s}\sum_{I_1\uplus...\uplus I_s=I}
  \sum_{n_1+...+n_s=s}  \prod_{j=1}^s \nabla^{n_i}\varpi_{0,|I_j|+1}(I_j;-q)  
  \\
  &- \!\!\sum_{\substack{I'\uplus I''=I\\ u_k\in I''}} \!\!\!\!
\Delta\varpi_{0,|I'|+1}(I';q)
  \sum_{s=1}^{|I|} \frac{1}{s}\sum_{I_1\uplus...\uplus I_s=I''}
  \sum_{n_1+...+n_s=s}  \prod_{j=1}^s \nabla^{n_i}\varpi_{0,|I_j|+1}(I_j;-q)\;.
\end{align*}  
Consider the equation
\[
  0=  \Delta\varpi_{0,|I|+1}(I;q)
  + B(I;q)
- \sum_{\substack{I'\uplus I''=I\\ u_k\in I''}} 
\Delta\varpi_{0,|I'|+1}(I';q) B(I'';q)\;.
\]
Its iterative solution is
\begin{align*}
  B(I;q)&=-\Delta\varpi_{0,|I|+1}(I;q)
  -\sum_{s=2}^{|I|} \sum_{\substack{I_1\uplus .. \uplus I_s=I \\ u_k\in I_1}}
  \Delta\varpi_{0,|I_1|+1}(I_1;q)
 \prod_{j=2}^s \Delta\varpi_{0,|I_j|+1}(I_j;q)
 \\[-2ex]
 &=  -\sum_{s=1}^{|I|} \frac{1}{s} \sum_{I_1\uplus .. \uplus I_s=I}
 \prod_{j=1}^s \Delta\varpi_{0,|I_j|+1}(I_j;q)\;.
\end{align*}  
The factor $\frac{1}{s}$ arises by symmetrisation when dropping the
condition $u_k\in I_1$.  The consistency condition
\eqref{log-nabla-sym} of (\ref{reflect-U}) together with
Proposition~\ref{prop:log-Delta-om} thus imply $A(I;q)\equiv 0$, which
gives the assertion.
\end{proof}
\end{proposition}
As result we have proved that (\ref{reflect-U}) does not have any poles
on $\hat{\mathbb{C}}$, it is thus a constant equal to its
value $0$ at $q=\infty$.
This means that Assumption~\ref{assump-res0} is true and the proof of 
Theorem~\ref{thm:sol-dse} complete.

\section{Conclusion and outlook}

We have proved for genus $g=0$ the main conjecture of
\cite{Branahl:2020yru} that meromorphic forms $\omega_{g,n}$ which
naturally appear in the quartic analogue of the Kon\-tsevich model
follow blobbed topological recursion \cite{Borot:2015hna}. This makes
the quartic Kontsevich model part of the growing family of structures
in mathematics and physics governed by topological recursion
\cite{Eynard:2007kz, Eynard:2016yaa}. Other examples include the
combinatorics of the Kontsevich model \cite{Kontsevich:1992ti}, the
one- and two matrix models \cite{Chekhov:2006vd}, Hurwitz theory
\cite{Bouchard:2007hi}, Gromov-Witten theory \cite{Bouchard:2007ys},
Weil-Petersson volumes of moduli spaces of hyperbolic Riemann surfaces
\cite{Mirzakhani:2006fta} and many more.

We consider as most important result of this paper the discovery that 
the quartic Kontsevich model is completely characterised by the behaviour of its
objects  $\omega_{g,n}$ under the global involution $\iota z=-z$.
We showed how a single equation (\ref{eq:flip-om}),
\begin{align}
&  \omega_{0,|I|+1}(I,q)
  +\omega_{0,|I|+1}(I,\iota q)
\tag{\ref{eq:flip-om}}
\\
&=\sum_{s=2}^{|I|} \sum_{I_1\uplus ...\uplus I_s=I}
\frac{1}{s} \Res\displaylimits_{z\to q}  \Big(
\frac{dy(q) dx(z)}{(y(q)-y(z))^{s}}  \prod_{j=1}^s
\frac{\omega_{0,|I_j|+1}(I_j,z)}{dx(z)} 
\Big)
\nonumber
\end{align}
governs the genus-0 case. This equation admits a na\"{\i}ve solution 
\begin{align}
&  \omega_{0,|I|+1}(I,z)
\nonumber
\\
&=\sum_{i=1}^r \sum_{s=2}^{|I|} \sum_{I_1\uplus ...\uplus I_s=I}
\frac{1}{s}\Res\displaylimits_{q\to \beta_i} \frac{dz}{z-q}
\Res\displaylimits_{w\to q}  \Big(
\frac{dy(q) dx(w)}{(y(q)-y(w))^{s}}  \prod_{j=1}^s
\frac{\omega_{0,|I_j|+1}(I_j,w)}{dx(w)} 
\Big)
\nonumber
\\
&+ \sum_{k=1}^{|I|} \sum_{s=2}^{|I|} \sum_{I_1\uplus ...\uplus I_s=I}
\frac{1}{s}\Res\displaylimits_{q\to \iota u_k} \frac{dz}{z-q}
\Res\displaylimits_{w\to q}  \Big(
\frac{dy(q) dx(w)}{(y(q)-y(w))^{s}}  \prod_{j=1}^s
\frac{\omega_{0,|I_j|+1}(I_j,w)}{dx(w)} 
\Big)\;,
\nonumber
\end{align}
which however leaves each of the following points obscure:
\begin{enumerate}
\item Is (\ref{eq:flip-om}) meaningful, i.e.\ is its rhs symmetric
  under $q\mapsto \iota q$?
\item Has (\ref{eq:flip-om}) anything to do with topological recursion?

\item Is there any connection between   (\ref{eq:flip-om})  and the quartic
  Kontsevich model?
\end{enumerate}
To answer the first two of these critical questions we had to prove
that the  na\"{\i}ve solution is equivalent to the solution
(\ref{sol:omega})+(\ref{eq:kernel}) given in Theorem~\ref{thm:flip}.
The generated material also allowed to affirm question (c), where a
difficulty was to show that all poles are of purely higher order.
This property is a consequence of a hidden symmetry
(\ref{prod-om-nabla-sym}) resulting from (\ref{eq:flip-om}) alone. As
result we have established for genus $g=0$ a precise
connection between (b)
and (c), which was conjectured in \cite{Branahl:2020yru}.

But the statement is more general: in \cite{Branahl:2020yru} it is
shown that the poles of $\omega_{g\geq 1,n}(z_1,...,z_n)$ are located,
besides ramification points of $x$ and diagonals $z_k=\iota z_l$, at
the fixed points $z_k=\iota z_k$ of the involution.

The next step in our programme will be to extend (\ref{eq:flip-om}) to
higher genus. This should help to answer the exciting question whether
the intersection numbers \cite{Borot:2015hna} encoded in the quartic
Kontsevich model capture geometric information about a moduli space of
curves equipped with an involution. It would also be interesting to
investigate whether these structures relate to other extensions of
topological recursion. We mention a recent work \cite{Belliard:2021jtj}
(which contains a beautiful introduction to topological
recursion and its ramifications) on $r$-spin intersection numbers to which
the quartic Kontsevich model could be related.

\appendix

\markboth{\hfill\textsc{Maciej Do{\l}\k{e}ga}}{\textsc{{Blobbed topological
      recursion of the quartic Kontsevich model II: Genus=0}\hfill}}

\newcommand{\PPP}{\mathcal{P}}
\newcommand{\T}{\mathcal{T}}
\newcommand{\CC}{\mathbb{C}}

\section{Combinatorial identities involving labeled
  trees (by Maciej Do{\l}\k{e}ga)}

A {\em set partition} of $S$ is a (non-ordered) family of non-empty disjoint
subsets of $S$ (called parts of the partition), whose union is $S$.
In the following, we always assume that $S$ is finite. Denote by
$\PPP(S)$ the set of set partitions of $S$ and for any $\pi \in
\PPP(S)$ and for any $B \in \pi$ denote by $|\pi|$ the number of parts
of $\pi$ and by $|B|$ the number of elements in the part $B$.

A graph $G = (V,E)$ is a \emph{forest} if it has no cycles. If a
forest is additionally connected it is called a \emph{tree}. Denote by
$\T^{\ast}_{V}$ the set of \emph{plane trees}, that is trees embedded
in a plane, with the set of vertices $V$. Denote by $\T_V$ the set of
\emph{labeled trees} with the set of vertices $V$, that is the set of
trees, whose vertices are labeled by distinct numbers
$\{1,\dots,|V|\}$ (or, by isomorphism, any linearly ordered set of the
cardinality $|V|$). Finally, a tree is \emph{rooted} if it has a
distinguished vertex $v_\bullet \in V$ called \emph{the root} and we
denote by $\T^{\ast,\bullet}_{V}$ and $\T^\bullet_V$ the
set of plane rooted trees and of labelled rooted trees, respectively, with
the vertex set $V$. The degree of a vertex $v$ in a tree $T$ is the
number of adjacent vertices to $v$.

The following classical theorem is a multivariate version of the
celebrated Cayley's formula for the number of labeled trees.

\begin{theorem}
  \label{theo:Cayley}
    For any positive integer $n$ and family of indeterminates
    $x_1,\dots,x_n, x_{n+1}$the
  following formulas hold true:
  \begin{equation} 
\label{eq:Cayley}
(x_1+\cdots+x_{n+1})^{n-1} = \sum_{T \in \T_{[n+1]}
  }\prod_{v \in V}x_v^{\deg(v)-1},
\end{equation}
and 
\begin{equation} 
\label{eq:RootedCayley} 
(x_1+\cdots+x_{n+1})^{n} = \sum_{T \in \T_{[n+1]}^\bullet
  }x_{v_\bullet}^{\deg(v_\bullet)}\prod_{v \in
    V\setminus\{v_{\bullet}\}}x_v^{\deg(v)-1},
\end{equation}
\end{theorem}

Cayley proved his formula by computing a certain determinant
\cite{Cayley1897}, and the first bijective proof was given by Pr\"ufer
\cite{Prufer1918}. Since then many different proofs were proposed and
we would like to mention a relatively general method for counting
trees by the use of the matrix-tree theorem, see for instance
\cite{Abdesselam:2004??} for generalisations and applications.
\begin{corollary}
  \label{lemma:multinomialDolega}
    For any positive integer $n$ and family of indeterminates
    $x_1,\dots,x_n, x_{n+1}$ the
  following formula holds true:
  \begin{equation} 
\label{eq:multinomialDolega}
(x_1+\cdots+x_{n+1})^{n} = \sum_{T \in \T_{[n+1]}^{\ast,\bullet}
}\sum_{\varsigma \in \mathcal{S}_{n+1}}\frac{x_{\varsigma(v_\bullet)}^{\deg(v_\bullet)}}{
  \deg(v_\bullet)!}\prod_{v \in
    V\setminus\{v_{\bullet}\}}\frac{x_{\varsigma(v)}^{\deg(v)-1}}{(\deg(v)-1)!}\;.
\end{equation}
\end{corollary}

\begin{proof}
  Note that the symmetric group $\mathcal{S}_{n+1}$ acts on the set
  $\T_{[n+1]}^\bullet$ of labeled rooted trees by permuting the
  labels. Moreover, each labeled rooted tree is uniquely constructed
  by choosing a plane tree $T \in \T_{[n+1]}^{\ast,\bullet}$, a label
  for its root, and for each $v \in V$ a subset $L_v \subset [n+1]$ of
  labels of its children. There are
  \[\frac{(n+1)!}{\deg(v_\bullet)!\prod_{v \in
        V\setminus\{v_{\bullet}\}}(\deg(v)-1)!}\]
  choices for such labelings. Therefore acting by the permutation
  group on the labels and comparing it with the formula
  \eqref{eq:RootedCayley} we have got
\[(n+1)!(x_1+\cdots+x_{n+1})^{n} = \!\!\! \sum_{T \in \T_{[n+1]}^{\ast,\bullet}
  }\sum_{\varsigma \in \mathcal{S}_{n+1}}(n{+}1)!
  \frac{x_{\varsigma(v_\bullet)}^{\deg(v_\bullet)}}{\deg(v_\bullet)!}\prod_{v \in
    V\setminus\{v_{\bullet}\}}\frac{x_{\varsigma(v)}^{\deg(v)-1}}{(\deg(v)-1)!},
\]
which finishes the proof.
\end{proof}

\begin{corollary}\label{cor:leibniz}
  For any $(n{+}1)$-tuple $f_0,...,f_n$ of differentiable functions one has
  \begin{align*}
  	n! \sum_{k_0+...+k_n=n}   \prod_{i=0}^n \frac{f_{i}^{(k_i)}(x)}{k_i!}
  	&\equiv
  	(f_0f_1\cdots f_n)^{(n)}(x)
        \\[-1ex]
  	&=\sum_{T \in \T_{[n+1]}^{\ast,\bullet}}
  	\sum_{\varsigma \in \mathcal{S}_{n+1}}
        \frac{f^{(\deg(v_\bullet))}_{\varsigma (v_\bullet)}(x)}{\deg(v_\bullet)}
  	\prod_{v\in V\setminus \{v_\bullet\}}
        \frac{f_{\varsigma(v)}^{(\deg(v)-1)}(x)}{(\deg(v)-1)!}\;.
  \end{align*}
  \begin{proof}
    Set $x_i\mapsto \partial_{x_i}$ in
    (\ref{eq:multinomialDolega}), apply it to
    $f_0(x_0)f_1(x_1)\cdots f_n(x_n)$ and substitute
    $x_0=x_1=\dots=x_n=x$.
  \end{proof}
\end{corollary}

Here is a corollary from Theorem~\ref{theo:Cayley} which gives an identity
expressed in terms of set-partitions.
  \begin{corollary}
  \label{theo:conj1}
  For any positive integer $n$ and family of indeterminates $x_1,\dots,x_n$the
  following formula holds true:
  \begin{equation}
    \label{eq:identity}
(1+\sum_{i=1}^nx_i)^{n-1} = \sum_{\pi \in \PPP([n])}\prod_{B \in
  \pi}(\sum_{b \in B}x_b)^{|B|-1}.
    \end{equation}
\begin{proof}
For any $T \in \T_{[n+1]}$ removing the
vertex $n+1$ from it yields the decomposition into a collection of
disjoint rooted labeled
trees on the set $[n]$. This decomposition establishes a bijection
between rooted, labeled forests $F$ on the vertex set $[n]$ and labeled
trees $T$ on $n+1$ vertices. Moreover, the degrees of the non-root
vertices of $F$ coincides with their degrees in $T$, and the degrees of the root
vertices of $F$ are equal to their degrees in $T$ minus one. This
decomposition gives the following identity by plugging
$x_{n+1}=1$ in \eqref{eq:Cayley}:
\[ (x_1+\cdots+x_n+1)^{n-1} = \sum_{F}\prod_{v \in
    V_\bullet}x_{v}^{\deg(v)}\prod_{v \in V\setminus
    V_\bullet}x_{v}^{\deg(v)-1}.\]
Note that for any set-partition $\pi \in \PPP([n])$ and for any
collection of rooted, labeled trees $\{T_B \in \T_B^\bullet: B \in
\pi\}$ there exists a rooted, labeled forest $F$ on $[n]$, which is
the disjoint union of $\{T_B \in \T_B^\bullet: B \in
\pi\}$ and every rooted, labeled forest $F$ on $[n]$ is obtained in
this way. Therefore
\begin{align*}
  (x_1+\cdots+x_n+1)^{n-1} &= \sum_{\pi \in \PPP([n])}\prod_{B \in
  \pi}\bigg(\sum_{T_B \in \T_B^\bullet}x_{v_\bullet(T_B)}^{\deg(v_\bullet(T_B))}\prod_{v \in
                         V(T_B)\setminus\{v_\bullet(T_B)\}}x_v^{\deg(v)-1}\bigg) \\
  &= \sum_{\pi \in \PPP([n])}\prod_{B \in
  \pi}(\sum_{b \in B}x_b)^{|B|-1}
  \end{align*}
by \eqref{eq:RootedCayley}, which finishes the proof.
\end{proof}
  \end{corollary}

\begin{corollary}\label{coro:combin-l}
  For any positive integers $\ell < n$ and family of indeterminates
  $x_1,\dots,x_n$ the
  following formula holds true:  
  \begin{align}
    \binom{n-1}{\ell-1}(x_1+\cdots +x_{n})^{n-\ell} &= \sum_{\substack{\pi \in
    \PPP([n]):\\ |\pi|=\ell}}\prod_{B \in
    \pi}(\sum_{b \in B}x_b)^{|B|-1}.
        \label{eq:partition-power}
    \end{align}
  \begin{proof}
    It is enough to apply the binomial formula for the left hand
    side of \eqref{eq:identity} and compare the 
    homogenous parts of degree $n-\ell$.
  \end{proof}
  \end{corollary}

Let $x^{\overline{n}} := \prod_{i=0}^{n-1}(x+i)$ and
$x^{\underline{n}} := \prod_{i=0}^{n-1}(x-i)$ denote the raising and
the falling factorials. 
  \begin{corollary}
  \label{cor:partition}
The following identities hold true:
\begin{align}
 \sum_{\substack{\pi \in \PPP([n])\\|\pi|=\ell}}
 \prod_{B \in \pi} \Big(\sum_{b \in B}x_b\Big)^{\overline{|B|-1}}
 =\binom{n-1}{\ell-1} (x_1+\cdots+x_n)^{\overline{n-\ell}} \;,
 \label{eq:partition-rising}
 \\
 \sum_{\substack{\pi \in \PPP([n])\\|\pi|=\ell}}
 \prod_{B \in \pi} \Big(\sum_{b \in B}x_b\Big)^{\underline{|B|-1}}
 =\binom{n-1}{\ell-1} (x_1+\cdots+x_n)^{n-\ell} \;,
\label{eq:partition-falling}
\end{align}
\begin{proof}
  It is enough to realise that 
  \[
    \Big(\sum_{b \in B}x_b\Big)^{\overline{k}}
    =\Big(- \sum_{b \in B}\partial_{t_b}\Big)^k \prod_{j=1}^nt_j^{- x_j}
    \Big|_{t_j=1}\;,\quad
    \Big(\sum_{b \in B}x_b\Big)^{\underline{k}}
    =\Big(\sum_{b \in B}\partial_{t_b}\Big)^k \prod_{j=1}^nt_j^{x_j}
    \Big|_{t_j=1}\;.
  \]
Substitute
 $x_j\mapsto \mp \partial_{t_j}$ in
    (\ref{eq:partition-power}), apply it to
    $\prod_{j\in J} t_j^{\mp x_j}$ and set all $t_j\equiv 1$.
\end{proof}
\end{corollary}

Let $a(x), b(y), c_1(y),\dots, c_n(y) \in C^\infty(\mathbb{R})$ be
smooth functions (in fact they might be formal elements of a ring
equipped with the formal derivations $\partial_x, \partial_y$, see
\cite{Dolega2017}). In the following we are going to prove an explicit
combinatorial formula for the expression
  \[
    \big(b(y)\partial_x+\partial_y\big)^{n-1}
    \big(a(x)\cdot b(y) \cdot c_1(y)\cdots
    c_n(y)\big)
  \]
in terms of special labeled trees, where we allow repetitions.

Consider the set of rooted, labeled
trees $T$ such that
\begin{itemize}
\item the root vertex $v_\bullet(T)$ has label $-1$,
  \item the set of vertices adjacent to the root is denoted by $V_0(T)$
    and for any $v \in V_0(T)$ one has $\deg(v) >1$ and $v$ is labeled by $0$,
    \item the set $V_{[n]}(T)$ of the remaining vertices has cardinality
      $n$ and its elements are labeled by distinct numbers $1,\dots,n$.
    \end{itemize}
    We denote the set of these trees by $\T^{\bullet;0}_{[n]}$.

    \begin{theorem}
  For any $(n+2)$-tuple of functions $a(x), b(y), c_1(y),\dots,
  c_n(y)$ the following identity holds true:
    \begin{align}
      &\big(b(y)\partial_x+\partial_y\big)^{n-1} {}
      \big(a(x)\cdot b(y) \cdot c_1(y)\cdots
      c_n(y)\big) \nonumber\\
      &= \sum_{T \in
            \T^{\bullet;0}_{[n]}}a(x)^{(\deg(v_\bullet(T))-1)}\prod_{v
            \in V_0(T)}b(y)^{(\deg(v)-2)}
          \prod_{v \in V_{[n]}(T)}c_{\lab(v)}(y) ^{(\deg(v)-1)}
          \label{CayleyDerivative},
    \end{align}
    where $f^{(n)}(z) = \partial_z^n f(z)$ with the convention
    $f^{(0)}(z) = f(z)$.
    \begin{proof}
    We can decompose a tree $T \in \T^{\bullet;0}_{[n]}$ as
    follows. Suppose that the degree of the root of $T$ is equal to
    $r$. Let $T' \in \T_{[n+1]}^\bullet$ be a tree obtained from $T$
    by identifying all the vertices labeled by $0$ with the root
    vertex $v_\bullet$. Note that the degree of the root of $T'$ is
    equal to $r \leq \deg(v_\bullet(T')) \leq n$. In particular $T$ is
    uniquely determined by $T'$ and by a set-partition
    $\pi \in \PPP(N(v_\bullet(T')))$, where $N(v_\bullet(T'))$ is the
    set of vertices in $T'$ adjacent to the root $v_\bullet(T')$. Each
    block of $B \in \pi$ corresponds to a vertex of $T$ labeled by
    $0$. This decomposition gives us the following equality:
    \begin{align*}
      &\sum_{T \in
        \T^{\bullet;0}_{[n]}}a(x)^{(\deg(v_\bullet(T))-1)}\prod_{v
        \in V_0(T)}b(y)^{(\deg(v)-2)}\prod_{v \in
        V_{[n]}(T)}c_{\lab(v)}(y) ^{(\deg(v)-1)} \\
      &=\sum_{r=1}^{n}\sum_{\substack{T' \in
          \T_{[n+1]}^\bullet,\\
          \deg(v_\bullet)=r}}\prod_{v \in
        V\setminus \{v_{\bullet}\}}c_{\lab(v)}(y) ^{(\deg(v)-1)}
      \sum_{\pi \in \PPP([r])}a(x)^{(|\pi|-1)}\prod_{B\in \pi}b(y)^{(|B|-1)}.
 \end{align*}

Define a transformation $f : \CC[y,x_1,\dots,x_n] \to \CC(y)$ by
 declaring its action on monomials
 \begin{equation}
            \label{eq:trans}
            f(y^k\cdot x_1^{\alpha_1}\cdots
          x_n^{\alpha_n}) :=
          b^{(k)}(y)\prod_{i=1}^nc_i^{(\alpha_i)}(y).
 \end{equation}
 Using Leibnitz rule we can compute
 \[ \partial_y^{|B|-k}\Big(\prod_{b \in B}c_b(y)\Big) =
   f\Big(\Big(\sum_{b \in B}x_b\Big)^{|B|-k}\Big),
 \]
 and using the proof of Corollary~\ref{theo:conj1} we rewrite the
 right hand side of \eqref{CayleyDerivative} as
  \begin{align*}
    \sum_{r=1}^{n}\binom{n-1}{r-1}
    \bigg(\prod_{i=1}^nc_i(y)\bigg)^{(n-r)}\sum_{\pi \in \PPP([r])}
    a(x)^{(r-1)}\prod_{B   \in \pi}b(y)^{(|B|-1)}.
\end{align*}
It is enough to notice that
\[ \sum_{\pi \in \PPP([r])}a(x)^{(|\pi|-1)}\prod_{B \in \pi}b(y)^{(|B|-1)}
  = (b(y)\partial_x+\partial_y)^{r-1}a(x)b(y),
\]
which is easy to prove by induction on $r$ (every set-partition
$\pi \in \PPP([r+1])$ is either constructed from a set partition
$\pi' \in \PPP([r+1])$ by adding a new block $\{r+1\}$, which
corresponds to the action of $\partial_x b(y)$ on
$ (b(y)\partial_x+\partial_y)^{r-1}a(x)b(y)$ or it is constructed from
$\pi' \in \PPP([r+1])$ by adding $r+1$ to one of its blocks, which
corresponds to the action of $\partial_y$ on $(b(y)\partial_x+\partial_y)^{r-1}a(x)b(y)$). Summing up, we have that
\begin{align*}
  &\sum_{r=1}^{n}\binom{n-1}{r-1}
  \bigg(\prod_{i=1}^nc_i(y)\bigg)^{(n-r)}\sum_{\pi \in \PPP([r])}
  a(x)^{(|\pi|-1)}\prod_{B \in \pi}b(y)^{(|B|-1)} 
  \\
  &=\sum_{r=1}^{n}\binom{n-1}{r-1}\bigg(\prod_{i=1}^nc_i(y)\bigg)^{(n-r)}
  (b(y)\partial_x+\partial_y)^{r-1}a(x)b(y)
    \\
   &=(b(y)\partial_x+\partial_y+\partial_z)^{n-1}
   \big(a(x)b(y)\prod_{i=1}^nc_i(z)\big)\big|_{z=y}
   \\
   &= (b(y)\partial_x+\partial_y)^{n-1}\big(a(x)b(y)\prod_{i=1}^nc_i(y)\big),
 \end{align*}
 which finishes the proof of \eqref{CayleyDerivative}.
 \end{proof}
  \end{theorem}

\begin{corollary}\label{coro-axby}
  For any $(n+2)$-tuple of functions $a(x), b(y), c_1(y),\dots,
  c_n(y)$ the following identity holds true:
  \begin{align}
    \big(b(y)\partial_x+\partial_y\big)^{n-1} \big(&a(x)\cdot b(y) \cdot c_1(y)\cdots
      c_n(y)\big)  \nonumber\\
    =\sum_{\pi \in \PPP([n])}\partial_x^{|\pi|-1}&a(x)\prod_{B
      \in \pi}\bigg(\partial_y^{|B|-1}\big(b(y)\prod_{b \in
    B}c_b(y)\big)\bigg).
        \label{eq:identity1}
    \end{align}
  \end{corollary}
          
\begin{proof}    
    Note that any $T \in \T^{\bullet;0}_{[n]}$ is uniquely determined
    by the following data: pick a set-partition $\pi \in
    \PPP([n])$. For each part $B \in \pi$ pick a labeled tree $T_B \in
    \T_{B \cup \{0\}}$. Take the disjoint union of $(T_B)_{B \in
        \pi}$ and connect all the vertices labeled by $0$ to a new
      vertex labeled by $-1$. In this way we obtain a tree $T \in
      \T^{\bullet;0}_{[n]}$ and conversely, every $T \in
      \T^{\bullet;0}_{[n]}$ decomposes into a collection of labeled
      trees $\big(T_B \in
      \T_{B \cup \{0\}}\big)_{B \in\pi}$.
      This decomposition yields the following identity:
          \begin{align}
         &\sum_{T \in
            \T^{\bullet;0}_{[n]}}a(x)^{(\deg(v_\bullet(T))-1)}\prod_{v
            \in V_0(T)}b(y)^{(\deg(v)-2)}\prod_{v \in
           V_{[n]}(T)}c_{\lab(v)}(y) ^{(\deg(v)-1)} \nonumber\\
            &= \sum_{\pi \in \PPP([n])}a^{({|\pi|-1})}(x)\prod_{B
      \in \pi}\sum_{T_B \in \T_{B \cup
              \{0\}}}\bigg(b(y)^{(\deg(v_0(T_B))-1)}\prod_{b \in
              B}c_b^{(\deg(v_b(T_B))-1)} (y)\bigg).
              \label{eq:pomocnicze}
  \end{align}
  Similarly as before we can compute
     \[ \partial_y^{|B|-1}\big(b(y)\prod_{b \in B}c_b(y)\big) =
       f\Big((y+\sum_{b \in B}x_b)^{|B|-1}\Big),
     \]
     where $f$ is a transformation given by \eqref{eq:trans}.
          Using \eqref{eq:Cayley} and the definition of $f$ we can further transform it into
          \[ \partial_y^{|B|-1}\big(b(y)\prod_{b \in B}c_b(y)\big) = \sum_{T_B \in \T_{B \cup \{0\}}}\bigg(b(y)^{(\deg(v_0(T_B))-1)}\prod_{b \in B}c_b^{(\deg(v_b(T_B))-1)} (y)\big)\bigg).\]
    Plugging it into the right hand side of \eqref{eq:pomocnicze} and
    using  \eqref{CayleyDerivative} we end up precisely with
    \eqref{eq:identity1}, which finishes the proof.
  \end{proof}

\section{An identity used in   Section~\ref{sec:iota-symm-II}}
\markboth{\hfill\textsc{Alexander Hock and Raimar Wulkenhaar}}{\textsc{{Blobbed topological
      recursion of the quartic Kontsevich model II: Genus=0}\hfill}}

We recall two well-known identities \cite[Vol.~4,  eq.~(10.18) \&
Vol.~5, eq.~(1.18)]{Gould}:
\begin{align}\label{Go2}
  &\sum_{k=0}^n(-1)^k\binom{n}{k}\binom{x+k}{r+k}=(-1)^n\binom{x}{r+n}
  \\\label{Go1}
  &\sum_{i=0}^k(-1)^i\binom{k}{i}\binom{x-i}{k}\frac{1}{y+i}
  =\frac{\binom{x+y}{k}}{y\binom{y+k}{k}},
\end{align}
which hold for $r\in \mathbb{N}$ and $x,y\in \mathbb{C}$.

Let $\mathcal{D}_n$ be the set of tuples $(n_1,n_2,n_3,n_4)$
of non-negative integers with $n_1+n_2+n_3+2n_4=n$ and
$n_3+n_4\neq 0$, that is
\begin{align}
  \mathcal{D}_n:=\{(n_1,n_2,n_3,n_4)|n_i\in \mathbb{N},
  n_3+n_4\neq 0,n_1+n_2+n_3+2n_4=n\}.
\label{Dn}
\end{align} 
Then the following decomposition is holds:
\begin{lemma}\label{combi}
  Let $y,\bar{y},w,\bar{w}\in \mathbb{C}$ and
  $e_1:=w+\bar{w}$ and $e_2:=y\bar{w}+\bar{y}w+w \bar{w}$.
  Then, we have for any $n$
\begin{align}
  &\sum_{k=0}^{n-1}\binom{n}{k}(y^k w^{n-k} +\bar{y}^k \bar{w}^{n-k})
  \nonumber
  \\
  &=\!\!\!\!\!\!\!\!\sum_{(n_1,n_2,n_3,n_4)\in \mathcal{D}_n}\!\!\!\!\!\!
  (-1)^{n_4}n\frac{\prod_{k=1}^{n_3+n_4-1}(n_1+k)(n_2+k)}{n_3!n_4!(n_3+n_4-1)!}
  y^{n_1}\bar{y}^{n_2}e_1^{n_3}e_2^{n_4}.
\label{eq:combi}
\end{align}
\begin{proof}
We expand $y^{n_1}\bar{y}^{n_2}e_1^{n_3}e_2^{n_4}$ into a linear combination of
$y^k\bar{y}^{\bar{k}} w^t\bar{w}^{\bar{t}}$. For given 
$n_1,n_2,n_3,n_4,k,\bar{k},t,\bar{t}$ at most one term of the multinomial
expansion of $e_1,e_2$ contributes.
The coefficient of
$y^k\bar{y}^{\bar{k}} w^t\bar{w}^{\bar{t}}$
in such a contribution is
\begin{align*}
&  [y^k\bar{y}^{\bar{k}} w^t\bar{w}^{\bar{t}}]
  (  y^{n_1}\bar{y}^{n_2}e_1^{n_3}e_2^{n_4})
  \\
  &= \frac{n_4!}{(k{-}n_1)!(\bar{k}{-}n_2)!(n_4{+}n_1{+}n_2{-}k{-}\bar{k})!}
   \frac{n_3!}{(t{+}k{-}n_4{-}n_1)!(\bar{t}{+}\bar{k}{-}n_4{-}n_2)!}\;,
 \end{align*}
where $k+\bar{k}+t+\bar{t}=n=n_1+n_2+n_3+2n_4$. It is only non-zero if
$n_1\in [k-\bar{t}..k]$,
$n_2\in [\bar{k}-t..\bar{k}]$ and
$n_3\in [k+\bar{k}-n_1-n_2..\min(k-n_1+t,\bar{k}-n_3+\bar{t})]$.

We thus need to evaluate the sum
\begin{align}\label{sumT}
  [y^k\bar{y}^{\bar{k}} w^t\bar{w}^{\bar{t}}]\eqref{eq:combi}
  = \sum_{n_1=\mathrm{max}(0,k-\bar{t})}^k
  \sum_{n_2=\mathrm{max}(0,\bar{k}-t)}^{\bar{k}}
  \sum_{n_4=k+\bar{k}-a-b}^{\mathrm{min} (k-n_1-t,\bar{k}-n_2+\bar{t})}
  T_{n_1,n_2,n_4}
\end{align}
where 
\begin{align}\label{T}
  &  T_{n_1,n_2,n_4}
  \\
&=n(-1)^{n_4}\frac{(n_1+n_3+n_4-1)!(n_2+n_3+n_4-1)!}{
    n_1!n_2!n_3!n_4!(n_3+n_4-1)!}
[y^k\bar{y}^{\bar{k}} w^t\bar{w}^{\bar{t}}]
  (  y^{n_1}\bar{y}^{n_2}e_1^{n_3}e_2^{n_4})  
 \nonumber
\end{align}
with $n_3=n-n_1-n_2-2n_4$.
The aim is to prove that
\eqref{sumT}+\eqref{T} breaks down to
\begin{align*}
  \binom{n}{k}\delta_{t,n-k}\delta_{\bar{k},0}
  \delta_{\bar{t},0}+\binom{n}{\bar{k}}\delta_{k,0}
  \delta_{t,0}\delta_{\bar{t},n-\bar{k}}.
\end{align*}
Shifting summation indices to
$n_1=a{+}k{-}\bar{t}$, $n_2=b{+}\bar{k}{-}t$,
$n_4=c{+}k{+}\bar{k}{-}n_1{-}n_2= c{+}t{+}\bar{t}{-}a{-}b$ leads to
\begin{align*}
  [y^k\bar{y}^{\bar{k}} w^t\bar{w}^{\bar{t}}]\eqref{eq:combi}
  &=\sum_{a=0}^{\bar{t}}\sum_{b=0}^{t}
  \sum_{c=0}^{\mathrm{min}(a,b)}
    \frac{(t+k+a-c-1)!(\bar{t}+\bar{k}+b-c-1)!}{
    (a+k-\bar{t})!(b+\bar{k}-t)!(t+\bar{t}-c-1)!}\\
  &\qquad \qquad\qquad\qquad\qquad\qquad
  \times \frac{n(-1)^{c+a+\bar{t}+b+t}}{(\bar{t}-a)!(t-b)!c!(b-c)!(a-c)!}\;.
\end{align*}
Next, change the order of the sums by
\begin{align*}
  \sum_{a=0}^{\bar{t}}\sum_{b=0}^{t}\sum_{c=0}^{\mathrm{min}(a,b)}f_{a,b,c}
  =\sum_{c=0}^{\mathrm{min}(t,\bar{t})}\sum_{a=c}^{\bar{t}}\sum_{b=c}^{t}f_{a,b,c}
  =\sum_{c=0}^{\mathrm{min}(t,\bar{t})}\sum_{a=0}^{\bar{t}-c}\sum_{b=0}^{t-c}f_{a+c,b+c,c}
\end{align*}
to derive
\begin{align*}
&[y^k\bar{y}^{\bar{k}} w^t\bar{w}^{\bar{t}}]\eqref{eq:combi}
\\
&=\sum_{c=0}^{\mathrm{min}(t,\bar{t})}\sum_{a=0}^{\bar{t}-c}\sum_{b=0}^{t-c}
\frac{n(-1)^{c+a+\bar{t}+b+t}(t+k+a-1)!(\bar{t}+\bar{k}+b-1)!}{
  (a{+}c{+}k{-}\bar{t})!(b{+}c{+}\bar{k}{-}t)!
  (t{+}\bar{t}{-}c{-}1)!(\bar{t}{-}c{-}a)!(t{-}c{-}b)!c!b!a!}
\\
&=\sum_{c=0}^{\mathrm{min}(t,\bar{t})}
\frac{n(-1)^{c+\bar{t}+t}(t+\bar{t}-1-c)!}{(\bar{t}-c)!(t-c)!c!}
\\
&\times \sum_{a=0}^{\bar{t}-c}(-1)^a 
\binom{\bar{t}-c}{a}\binom{t+k-1+a}{c+k-\bar{t}+a}
\sum_{b=0}^{t-c}(-1)^b
\binom{t-c}{b}\binom{\bar{t}+\bar{k}-1+b}{c+\bar{k}-t+b}.
\end{align*}
The sums over $a,b$ can be evaluated separately with the identity \eqref{Go2}.
Consequently, one concludes with identity \eqref{Go1}
\begin{align*}
[y^k\bar{y}^{\bar{k}} w^t\bar{w}^{\bar{t}}]\eqref{eq:combi}
&=n\binom{t+k-1}{k} \binom{\bar{t}+\bar{k}-1}{\bar{k}}
\sum_{c=0}^{\mathrm{min}(t,\bar{t})}\frac{(-1)^{c}(t+\bar{t}-1-c)!}{
  (\bar{t}-c)!(t-c)!c!}
\\
&=n\binom{t+k-1}{k} \binom{\bar{t}+\bar{k}-1}{\bar{k}}
\Big(\frac{\delta_{t,0}}{\bar{t}}+\frac{\delta_{\bar{t},0}}{t}\Big)
\\
&=n\Big(\delta_{t,0}\underbrace{\binom{k-1}{k}}_{=\delta_{k,0}}
\binom{\bar{t}+\bar{k}-1}{\bar{k}}\frac{1}{\bar{t}}
+\delta_{\bar{t},0}\binom{t+k-1}{k}
\underbrace{\binom{\bar{k}-1}{\bar{k}}}_{=\delta_{\bar{k},0}}\frac{1}{t}\Big)
\\
&=n\Big(\delta_{t,0}\delta_{k,0} \binom{n-1}{\bar{k}}
\frac{1}{n-\bar{k}}+\delta_{\bar{t},0}\delta_{\bar{k},0}
\binom{n-1}{k} \frac{1}{n-k}\Big)
\\
&=\delta_{t,0}\delta_{k,0} \binom{n}{\bar{k}}
+\delta_{\bar{t},0}\delta_{\bar{k},0}\binom{n}{k}.
\end{align*} \vspace*{-4ex}

\end{proof}
\end{lemma}


\end{document}